\documentclass[11pt]{article}
\usepackage[margin=1in]{geometry}

\usepackage[utf8]{inputenc}
\usepackage[T1]{fontenc}
\usepackage{amsmath,amssymb,amsfonts}
\usepackage{mathtools}
\usepackage{amsthm}
\usepackage{graphicx}
\usepackage{booktabs}
\usepackage{algorithm}
\usepackage{algpseudocode}
\usepackage{url}
\usepackage{cite}
\usepackage[colorlinks=true, allcolors=blue]{hyperref}
\usepackage{tikz}
\usepackage{pgfplots}
\usepackage{pgfplotstable}
\pgfplotsset{compat=1.18}
\usetikzlibrary{positioning,shapes,shapes.symbols,shapes.geometric,patterns,calc,matrix,arrows.meta,decorations.pathreplacing}
\usepgfplotslibrary{fillbetween}
\pgfplotsset{layers/standard}

\definecolor{loyaltyblue}{RGB}{31,119,180}
\definecolor{effortorange}{RGB}{255,127,14}
\definecolor{outputgreen}{RGB}{34,139,34}
\definecolor{freeriderred}{RGB}{214,39,40}
\definecolor{teamviolet}{RGB}{148,103,189}
\definecolor{powerblue}{RGB}{31,119,180}
\definecolor{logorange}{RGB}{255,127,14}

\definecolor{recippurple}{RGB}{148,103,189}
\definecolor{cooperationblue}{RGB}{31,119,180}
\definecolor{defectionred}{RGB}{214,39,40}
\definecolor{memoryviolet}{RGB}{128,90,170}
\definecolor{forgivegreen}{RGB}{44,160,44}

\definecolor{trustcolor}{RGB}{31, 119, 180}
\definecolor{reputcolor}{RGB}{255, 127, 14}
\definecolor{coopcolor}{RGB}{44, 160, 44}
\definecolor{defectcolor}{RGB}{214, 39, 40}

\theoremstyle{plain}
\newtheorem{theorem}{Theorem}[section]
\newtheorem{proposition}[theorem]{Proposition}

\theoremstyle{definition}
\newtheorem{definition}[theorem]{Definition}

\newcommand{\vect}[1]{\mathbf{#1}}
\DeclareMathOperator*{\argmax}{arg\,max}

\begin{document}

\title{Computational Foundations for Strategic Coopetition: Formalizing Sequential Interaction and Reciprocity}

\author{
Vik Pant\thanks{Email: vik.pant@mail.utoronto.ca} \quad Eric Yu\thanks{Email: eric.yu@utoronto.ca}\\
\\
Faculty of Information\\
University of Toronto\\
140 St George St, Toronto, ON M5S 3G6, Canada
}

\maketitle

\begin{abstract}
Strategic coopetition in multi-stakeholder systems requires understanding how cooperation persists through time when actors cannot rely on binding contracts or external enforcement. Platform ecosystems, open-source communities, and requirements engineering contexts all exhibit sequential interactions where cooperation at time $t$ depends on observed partner behavior at time $t-1$. While conceptual modeling languages like \textit{i*} capture structural dependencies between actors, they lack mechanisms for representing behavioral dependencies and conditional cooperation in these sequential settings. When Actor A cooperates, influencing Actor B's propensity to reciprocate, existing models do not provide formal mechanisms to capture this temporal conditionality. Similarly, while game theory provides rich theories of reciprocity through repeated games and folk theorems, classical approaches assume infinite memory and perfect observability rather than the bounded rationality characteristic of real strategic actors and agents.

This technical report bridges this gap by extending computational foundations for strategic coopetition to sequential interaction dynamics. Building on companion work that formalized interdependence and complementarity (achieving 58/60 validation)~\cite{pant2025foundations} and trust dynamics (achieving 49/60 validation)~\cite{pant2025trust}, we address how cooperation is sustained through endogenous enforcement rather than external mechanisms. We develop bounded reciprocity response functions $\phi_{\text{recip}}(x) = \tanh(\kappa x)$ that map partner deviations to finite conditional responses, preventing unrealistic escalation while maintaining proportional reactions. We formalize memory-windowed history tracking through moving averages over $k$ recent periods, capturing cognitive limitations rather than assuming infinite perfect recall. We derive structural reciprocity sensitivity from interdependence matrices showing $\rho_{ij} = \rho_0 D_{ij}^\eta$, where behavioral responses are amplified by structural dependencies from \textit{i*} networks. The framework integrates trust-gated reciprocity where trust multiplies reciprocity responses through $T_{ij}^t \cdot \rho_{ij}$, such that low trust dampens reciprocity even when conditional cooperation would be strategically rational. The framework applies to both human stakeholder interactions (where reciprocity captures social norms of fairness) and multi-agent computational systems (where these mechanisms map to conditional cooperation protocols), demonstrating applicability across traditional requirements engineering and emerging agentic AI systems.

Comprehensive experimental validation across 15,625 parameter configurations (full factorial design with six parameters at five levels each) demonstrates robust emergence of reciprocity effects, with all six behavioral targets achieving validation thresholds: cooperation emergence (97.5\%, threshold $>$85\%), defection punishment (100.0\%, threshold $>$95\%), forgiveness dynamics (87.9\%, threshold $>$80\%), asymmetric differentiation (100.0\%, threshold $>$90\%), trust-reciprocity interaction (100.0\%, threshold $>$90\%), and bounded responses (100.0\%). The full TR-2 two-layer trust model with reputation ceiling, interdependence amplification, and 3:1 negativity bias ensures path-dependent trust dynamics that support both punishment and forgiveness. Empirical validation using the Apple iOS App Store ecosystem (2008--2024) achieves 43.0 out of 51 applicable validation points (84.3\%), successfully reproducing documented cooperation patterns across symbiosis, maturation, tension, crisis, and adjustment phases. Statistical significance is confirmed at $p < 0.001$ with Cohen's $d = 1.57$ (large effect size) and bootstrap confidence intervals demonstrating robustness under parameter perturbation.

This technical report concludes the \textit{Foundations Series} (TR-1 through TR-4) of a coordinated research program on computational approaches to strategic coopetition, adopting uniaxial treatment where agents choose cooperation levels along a single continuum with competitive dynamics emerging through structural parameters. Companion work on interdependence and complementarity~\cite{pant2025foundations} (arXiv:2510.18802), trust dynamics~\cite{pant2025trust} (arXiv:2510.24909), and collective action and loyalty~\cite{pant2025teams} (arXiv:2601.16237) has been prepublished. A companion \textit{Extensions Series} (TR-5 through TR-8) introduces biaxial treatment where cooperation and competition constitute independent strategic dimensions, addressing phenomena that inherently require two-dimensional action spaces.
\end{abstract}


\noindent\textbf{ArXiv Classifications:} cs.SE (Software Engineering), cs.MA (Multiagent Systems), cs.AI (Artificial Intelligence), cs.CY (Computers and Society)

\section{Introduction}

Modern multi-stakeholder systems exhibit strategic coopetition where actors simultaneously cooperate to create value and compete to appropriate it. Examples abound in requirements engineering where stakeholders iteratively disclose information based on analysts' responsiveness, software ecosystems where developers reciprocate platform provider behavior, and open-source communities where contributors voluntarily reciprocate others' contributions. A fundamental question underlies these environments: how does cooperation persist in the absence of binding contracts or third-party enforcement?

The answer is reciprocity. Cooperation is sustained because actors employ conditional strategies where they cooperate if their partners cooperate but defect if their partners defect. When Actor A observes Actor B cooperating in period $t-1$, Actor A reciprocates by cooperating in period $t$. When Actor B defects, Actor A punishes by withdrawing cooperation. This history-dependent conditionality enables cooperation to emerge and stabilize through endogenous enforcement rather than external mechanisms.

However, existing conceptual modeling approaches cannot formally represent reciprocity dynamics. Languages like \textit{i*}, Tropos, GRL, and KAOS capture structural dependencies where Actor $i$ depends on Actor $j$ for dependum $d$, but not behavioral dependencies where Actor $i$'s cooperation at time $t$ depends on Actor $j$'s cooperation at time $t-1$. They model static dependency structures, not dynamic sequential responses. This limitation prevents conceptual models from analyzing how cooperation emerges through reciprocity.

Conversely, game theory provides rich theories of reciprocity through repeated games, tit-for-tat strategies, and folk theorems showing cooperation can be Nash equilibria when interactions are indefinitely repeated. However, classical game theory assumes actors have infinite memory, perfect observability, and known payoff functions. However, actors and agents in the real world face bounded rationality with limited memory windows, observe actions imperfectly, and have uncertain preferences.

This paper bridges conceptual modeling's structural representations with game theory's behavioral mechanisms while addressing bounded rationality. Building on the conceptual modeling framework for strategic coopetition established in Pant's doctoral dissertation~\cite{pant2021strategic}, which identified reciprocity as one of five critical dimensions of coopetitive relationships, we formalize reciprocity as the fourth dimension completing a coordinated research program integrating interdependence, trust dynamics, and team production. Reciprocity explains how cooperation discovered through complementarity and interdependence is sustained over time through conditional enforcement.

\subsection{Research Program Context}

This technical report is the fourth, concluding the Foundations Series of a coordinated research program on computational approaches to strategic coopetition in requirements engineering and multi-agent systems. The complete program addresses five key dimensions of coopetitive relationships identified by Pant~\cite{pant2021strategic}: interdependence, complementarity, trustworthiness, reciprocity, and complex actor abstractions. Understanding how this fourth paper fits within the broader research program is essential for appreciating both its novel contributions and its integrative role.

The first technical report established foundational concepts of interdependence and complementarity, showing how structural dependencies create outcome coupling and how complementarity drives superadditive value creation. That companion work developed systematic translation from \textit{i*} dependency networks to quantitative interdependence matrices, formalized value creation through complementarity functions exhibiting superadditivity, integrated structural dependencies with bargaining power in value appropriation, and introduced the Coopetitive Equilibrium as a solution concept extending Nash Equilibrium to incorporate structural interdependence.

The second technical report extended that foundation by incorporating trust dynamics, showing how confidence in partner reliability evolves through repeated interactions and modulates strategic behavior through trust-gated reciprocity. That work formalized trust as a two-layer system with immediate trust responding to current behavior and reputation tracking violation history. Trust evolves through asymmetric updating where cooperation builds trust gradually while violations erode it sharply, creating hysteresis effects and trust ceilings. The companion work demonstrated through comprehensive validation that trust negativity ratios of approximately three to one align with empirical findings from behavioral economics.

The third technical report addressed a different aspect of coopetition by examining the challenges that emerge when strategic actors are not individuals but teams with internal structure. Teams face the fundamental problem of team production where members are tempted to free-ride on others' contributions. That work formalized loyalty mechanisms including cost tolerance, welfare internalization, warm glow, and guilt aversion that overcome free-riding incentives. The companion work showed through experimental validation that loyalty creates six-fold effort differentiation between low and high loyalty teams, demonstrating that prosocial preferences fundamentally transform team dynamics from universal shirking to sustained cooperation.

This fourth technical report addresses the dimension of reciprocity in sequential cooperation scenarios. While the previous three papers largely focused on simultaneous move games where actors choose actions at the same time, many real coopetitive relationships involve sequential moves where actors take turns. Reciprocity mechanisms enable cooperation in these sequential settings by conditioning current cooperation on past cooperation observed from partners. This paper shows how reciprocity mechanisms must integrate with interdependence structures from the first paper, trust dynamics from the second paper, and team production considerations from the third paper.

The reciprocity formalization demonstrates how actors condition their cooperative behavior on histories of partner behavior, how reciprocity interacts with trust where low trust may prevent reciprocal cooperation even when histories suggest it would be beneficial, how teams can engage in reciprocal relationships while managing internal free-riding, and how structural dependencies amplify reciprocity sensitivity similar to their effects on trust. By completing this fourth dimension, we provide comprehensive computational tools for analyzing strategic coopetition within the uniaxial treatment, where agents choose cooperation levels along a single continuum and competitive dynamics emerge through structural parameters. A companion Extensions Series (TR-5 through TR-8) addresses phenomena requiring biaxial treatment where cooperation and competition constitute independent strategic dimensions.

Table~\ref{tab:research_program} summarizes the four dimensions and their primary mathematical constructs.

\begin{table}[htbp]
\centering
\caption{Foundations Series dimensions and primary constructs}
\label{tab:research_program}
\begin{tabular}{llll}
\toprule
\textbf{Dimension} & \textbf{Report} & \textbf{Primary Constructs} & \textbf{Status} \\
\midrule
Interdependence \& Complementarity & TR-1 & $D_{ij}$, $V(\vect{a}|\gamma)$ & arXiv:2510.18802 \\
Trust \& Reputation & TR-2 & $T_{ij}^t$, $R_{ij}^t$ & arXiv:2510.24909 \\
Collective Action \& Loyalty & TR-3 & $\theta_i$, $\mathcal{C}$ & arXiv:2601.16237 \\
\textbf{Sequential Interaction \& Reciprocity} & \textbf{TR-4} & $\rho_{ij}$, $\phi_{\text{recip}}$ & \textbf{This report} \\
\bottomrule
\end{tabular}
\end{table}

\subsection{Research Program Scope}

The coordinated research program comprises two series. The \textit{Foundations Series} (TR-1 through TR-4) adopts the uniaxial treatment following the social dilemma tradition \cite{bengtsson2000coopetition,lado1997competition,ledyard1995public,berg1995trust}, where each agent chooses a single action $a_i \in [0, e_i]$ representing cooperation level. This single-dimensional action space is standard in public goods games, trust games, and continuous prisoner's dilemma formulations. Competition is modeled structurally through bargaining shares $\alpha_i$ (value capture allocation), interdependence asymmetries $D_{ij}$ (structural power), trust dynamics $T_{ij}^t$ (reputation-mediated access), and endowment differences $e_i$ (resource-based advantage). The uniaxial treatment reflects the insight that competition in coopetitive relationships often manifests through structural position rather than explicit competitive actions.

A companion \textit{Extensions Series} (TR-5 through TR-8) introduces biaxial treatment \cite{brandenburger1996co,gnyawali2011co} where agents choose both cooperation level and competition level as independent strategic dimensions: $a_i = (a_i^{\text{coop}}, a_i^{\text{comp}})$. This extension addresses phenomena that cannot be captured by the uniaxial formalism, including active value capture where bargaining shares are endogenous rather than structural, competitive harm where agents impose direct costs on rivals, multi-market contact where cross-arena spillovers create forbearance effects, and inter-coalition rivalry where nested groups cooperate internally while competing externally. The Extensions Series builds on the Foundations Series concepts ($D_{ij}$, $T_{ij}$, $\theta_i$, $\rho_{ij}$) and extends them to biaxial contexts, maintaining series coherence.

\subsection{Relationship to Foundational Work}

This technical report presents the computational formalization of the reciprocity dimension conceptually identified in Pant's doctoral dissertation~\cite{pant2021strategic}. While that prior work established reciprocity as a critical dimension of coopetitive relationships and conceptually outlined how conditional behavior sustains cooperation in multi-stakeholder systems, this report contributes the specific mathematical equations, the operational translation framework, and the rigorous validation methodology necessary for quantitative analysis and empirical application. The conceptual foundations are established in~\cite{pant2021strategic}; this report provides their computational realization.

This technical report assumes familiarity with the foundational concepts established in our three companion works. We briefly recap essential notation in Section~\ref{sec:recap} for self-containment, but readers seeking detailed justification of the base frameworks should consult those foundational reports before proceeding, as this work builds directly upon those frameworks.

From our foundational work on interdependence and complementarity, we build upon the interdependence matrix $D$ that quantifies structural dependencies between actors derived from \textit{i*} dependency networks, the value creation function $V(a|\gamma)$ with complementarity parameter that models synergistic value creation with validated parameterization, the private payoff function $\pi_i(a)$ with value appropriation mechanisms, the base integrated utility function $U_i(a)$ incorporating both private payoffs and dependency-weighted partner outcomes, and the Coopetitive Equilibrium concept extending Nash Equilibrium to incorporate structural interdependence.

From our companion technical report on trust dynamics, we build upon dynamic trust state variables $T_{ij}^t$ for immediate trust and $R_{ij}^t$ for reputation that evolve based on behavioral observations, trust evolution dynamics with asymmetric updating where cooperation builds trust gradually at rate $\alpha$ while violations erode it sharply at rate $\beta$, trust-gated reciprocity mechanisms where trust modulates conditional cooperation through utility augmentation, and the Dynamic Coopetitive Equilibrium that extends the base equilibrium concept to repeated interactions with evolving trust states.

From our companion work on team production, we build upon teams as composite actors with internal structure distinguishing team-level behavior from member-level incentives, team production functions that aggregate individual contributions through collaborative processes, free-riding problems where team members can shirk while benefiting from others' efforts, loyalty mechanisms that create psychological costs for free-riding through cost tolerance, welfare internalization, warm glow, and guilt aversion, and multi-level equilibrium concepts where conditions must be satisfied at both individual member and team levels.

What this technical report adds to complete the framework is formalization of reciprocity mechanisms for sequential cooperation where actors condition behavior on observed partner histories. We develop bounded response functions that map partner deviations to finite conditional responses preventing unrealistic escalation. We formalize memory-windowed history tracking that captures bounded rationality through moving averages over recent periods. We derive reciprocity sensitivity from structural dependencies showing how interdependence amplifies reciprocity responses. We integrate reciprocity with trust dynamics showing how trust gates reciprocity even when conditional cooperation would be strategically rational. We extend to team contexts showing how teams engage in reciprocal relationships while managing internal coordination. We provide systematic translation methodology from \textit{i*} sequential dependencies to computational parameters enabling practitioners to instantiate reciprocity models from conceptual models.

\subsection{Contributions}

The main contributions of this technical report are:

\begin{enumerate}
    \item \textbf{Bounded reciprocity response functions}: A formal specification of reciprocity through $\phi_{\text{recip}}(x) = \tanh(\kappa x)$ that maps partner behavioral deviations to finite conditional responses, ensuring extreme deviations produce bounded reactions preventing unrealistic escalation while maintaining proportional responses in moderate ranges.

    \item \textbf{Memory-windowed history tracking}: Formalization of bounded memory through moving averages $\bar{a}_j^{t-k:t-1}$ over $k$ recent periods, capturing cognitive limitations and attention constraints rather than assuming infinite perfect recall characteristic of classical game-theoretic models.

    \item \textbf{Structural reciprocity sensitivity}: Derivation of reciprocity sensitivity from interdependence matrices showing $\rho_{ij} = \rho_0 D_{ij}^\eta$, where behavioral responses are amplified by structural dependencies from \textit{i*} networks rather than treated as exogenous parameters.

    \item \textbf{Trust-gated reciprocity integration}: Integration with trust dynamics from companion work~\cite{pant2025trust} through $T_{ij}^t \cdot \rho_{ij}$, where low trust dampens reciprocity even when conditional cooperation would be strategically rational, while high trust enables full reciprocal enforcement.

    \item \textbf{Systematic translation methodology}: An eight-step framework enabling practitioners to instantiate computational models from \textit{i*} sequential dependencies, providing concrete guidance for eliciting memory windows, establishing cooperation baselines, and parameterizing reciprocity sensitivity.

    \item \textbf{Comprehensive parameter validation}: Validation across 15,625 parameter configurations (full factorial design $5^6$) demonstrating robust emergence of reciprocity effects with all six behavioral targets achieving validation thresholds, enabled by integrating the full TR-2 two-layer trust model with reputation-mediated ceiling dynamics.

    \item \textbf{Empirical case study validation}: Validation through the Apple iOS App Store ecosystem (2008--2024) achieving 43.0 out of 51 applicable validation points (84.3\%), successfully reproducing documented cooperation patterns across five distinct phases.

    \item \textbf{Statistical significance demonstration}: Confirmation of results at $p < 0.001$ with Cohen's $d = 1.57$ (large effect size) and bootstrap confidence intervals demonstrating robustness under $\pm 15\%$ parameter perturbation.

    \item \textbf{Monte Carlo robustness analysis}: Validation through 2,000 Monte Carlo trials demonstrating stable behavioral differentiation ratios under stochastic parameter variation.

    \item \textbf{Complete framework synthesis}: Integration with TR-1 (interdependence), TR-2 (trust), and TR-3 (loyalty) showing how reciprocity provides temporal enforcement sustaining cooperation discovered through complementarity and maintained through trust in multi-level coopetitive systems.
\end{enumerate}

\subsection{Illustrative Scenarios}
\label{sec:scenarios}

Before presenting the formal framework, we ground the reciprocity concepts through two scenarios that illustrate how sequential cooperation and conditional responses manifest in practice. These scenarios motivate the mathematical constructs developed in subsequent sections and demonstrate applicability across human stakeholder interactions and computational multi-agent systems.

\subsubsection{Scenario A: Platform Ecosystem Sequential Moves}

Consider a mobile application platform ecosystem where a platform provider (Apple, Google) interacts with third-party developers over multiple periods. This scenario exhibits the sequential, conditional dynamics that reciprocity mechanisms must capture.

\textbf{Period $t$: Platform Action.} The platform provider announces API changes that affect developer applications. The platform can choose cooperation levels ranging from providing extensive documentation, migration tools, and extended deprecation timelines (high cooperation, $a_{\text{Platform}}^t \approx 1.0$) to minimal notice with breaking changes (low cooperation, $a_{\text{Platform}}^t \approx 0.2$).

\textbf{Period $t+1$: Developer Response.} Developers observe the platform's action and respond based on recent history. If the platform cooperated by providing helpful transition support, developers reciprocate by investing in app quality, adopting new platform features, and maintaining ecosystem commitment. If the platform defected by imposing costly changes without support, developers reduce investment, delay updates, or explore alternative platforms. The developer response depends on the platform's behavior not just in period $t$ but over recent periods captured by memory window $k$.

\textbf{Period $t+2$: Platform Adjustment.} The platform observes developer responses and adjusts future policies accordingly. If developers reciprocated cooperation with sustained investment, the platform continues cooperative policies. If developers withdrew investment following perceived defection, the platform may adjust policies to rebuild cooperation or may further defect if it perceives low developer commitment.

This sequential pattern illustrates several key features our framework must capture: (i) actions in period $t+1$ depend on observed actions in period $t$, requiring history tracking; (ii) responses are proportional rather than binary, requiring bounded response functions; (iii) actors consider recent history rather than single periods, requiring memory windows; (iv) the platform's higher structural power (developers depend more on the platform than vice versa) affects reciprocity dynamics, requiring integration with interdependence matrices.

\subsubsection{Scenario B: Requirements Elicitation Reciprocity}

Consider a requirements engineering context where a systems analyst interacts with organizational stakeholders to elicit requirements for a new enterprise system. Information disclosure exhibits reciprocal patterns that the framework must formalize.

\textbf{Initial Period: Analyst Disclosure.} The analyst begins by disclosing methodology constraints, project limitations, and honest assessments of what the system can and cannot achieve. This represents a cooperative move establishing a foundation for reciprocal exchange.

\textbf{Subsequent Period: Stakeholder Response.} Stakeholders observe the analyst's openness and decide whether to reciprocate. If the analyst demonstrated cooperation through honest disclosure, stakeholders are more likely to share sensitive requirements including political constraints, hidden assumptions, and concerns about organizational resistance. If the analyst appeared guarded or dismissive, stakeholders may withhold critical information.

\textbf{Continuing Interaction: Trust-Gated Reciprocity.} As interactions continue, reciprocity becomes gated by accumulated trust. Even if the analyst cooperates in the current period, stakeholders with low trust (perhaps from prior negative experiences with IT projects) may not fully reciprocate. The stakeholder's response thus depends on both the analyst's recent behavior (captured by memory-windowed history) and accumulated trust state (captured by trust dynamics from TR-2).

This scenario illustrates: (i) how reciprocity enables voluntary information exchange without contractual enforcement; (ii) how trust accumulated over multiple interactions gates reciprocity responses; (iii) how bounded rationality limits stakeholder memory to recent interactions; (iv) how structural dependencies (stakeholders depend on the analyst for system quality, the analyst depends on stakeholders for requirements) shape reciprocity sensitivity.

\subsection{Application to Requirements Engineering and Multi-Agent Systems}
\label{sec:applications}

The reciprocity framework applies across two distinct but formally analogous domains: human requirements engineering contexts and computational multi-agent systems. Table~\ref{tab:translation_interpretations} provides translation between the mathematical constructs and their domain-specific interpretations.

\begin{table}[htbp]
\centering
\caption{Translation between mathematical constructs and domain interpretations}
\label{tab:translation_interpretations}
\begin{tabular}{lll}
\toprule
\textbf{Mathematical Construct} & \textbf{Human RE Interpretation} & \textbf{Multi-Agent Interpretation} \\
\midrule
$\rho_{ij}$ (reciprocity sensitivity) & Stakeholder reciprocity tendency & Conditional cooperation coefficient \\
$k$ (memory window) & Bounded stakeholder recall & History buffer size \\
$\phi_{\text{recip}}(x)$ (response function) & Proportional response to behavior & Scaled reward adjustment \\
$T_{ij}^t \cdot \rho_{ij}$ (trust-gated) & Trust-modulated reciprocity & Confidence-weighted cooperation \\
$\bar{a}_j^{t-k:t-1}$ (moving average) & Recent behavior assessment & Moving average state \\
$s_{ij}^t$ (cooperation signal) & Deviation from expectations & Behavioral delta signal \\
$D_{ij}$ (interdependence) & Stakeholder dependency strength & Agent coupling coefficient \\
\bottomrule
\end{tabular}
\end{table}

\textbf{Human Requirements Engineering.} In requirements engineering contexts, reciprocity captures the social norms of fairness and exchange that govern stakeholder interactions. When analysts demonstrate responsiveness to stakeholder concerns, stakeholders reciprocate by providing more complete and accurate requirements. Memory windows reflect cognitive limitations where stakeholders primarily recall recent interactions rather than complete project history. Trust-gated reciprocity captures the phenomenon where stakeholders with negative prior experiences remain guarded despite current cooperative analyst behavior.

\textbf{Multi-Agent Computational Systems.} In multi-agent systems including emerging agentic AI configurations, reciprocity mechanisms map to conditional cooperation protocols. Memory windows become history buffer parameters determining how many recent interactions influence current decisions. Response functions become reward adjustment mechanisms that scale cooperation based on partner behavior. Trust-gated reciprocity becomes confidence-weighted cooperation where agents adjust collaboration intensity based on reliability assessments.

The mathematical framework is agnostic to which interpretation applies, enabling unified analysis across both domains. This dual applicability is particularly relevant as human-AI collaborative systems become prevalent, requiring frameworks that bridge human behavioral patterns and computational agent design.

\subsection{Dual-Track Validation Strategy}
\label{sec:validation_strategy}

Following the validation paradigm established in companion technical reports~\cite{pant2025foundations,pant2025trust,pant2025teams}, we employ a dual-track validation strategy combining comprehensive experimental validation with empirical case study validation.

\textbf{Experimental Validation.} We conduct systematic parameter space exploration across 15,625 configurations using a full factorial design with six parameters at five levels each. This comprehensive coverage enables assessment of model behavior across the complete parameter space rather than relying on selected configurations. We define six behavioral targets that the model should achieve to be considered valid:

\begin{enumerate}
\item \textit{Cooperation Emergence}: Cooperative equilibria should emerge in more than 85\% of configurations where reciprocity parameters are favorable.
\item \textit{Defection Punishment}: Negative reciprocity responses should occur in more than 95\% of configurations when partners defect.
\item \textit{Forgiveness Dynamics}: Recovery to baseline cooperation should occur within $2k$ periods in more than 80\% of configurations following isolated defections.
\item \textit{Asymmetric Differentiation}: Responses to high-dependency partners should exceed responses to low-dependency partners by factor greater than 1.5 in more than 90\% of configurations.
\item \textit{Trust-Reciprocity Interaction}: Cooperation levels should be higher when both trust and reciprocity are high compared to when either is low in more than 90\% of configurations.
\item \textit{Bounded Responses}: All reciprocity responses should satisfy $|\phi_{\text{recip}}| \leq 1.0$ in 100\% of configurations.
\end{enumerate}

\textbf{Empirical Case Study Validation.} We validate the framework against the Apple iOS App Store ecosystem (2008--2024), a well-documented platform ecosystem exhibiting sequential cooperation dynamics. The case study spans 66 quarters across five distinct phases: symbiosis (2008--2012), maturation (2012--2017), tension (2017--2020), crisis (2020--2021), and adjustment (2021--2024). We develop a validation rubric with 12 indicators across 5 phases (51 applicable points after excluding indicators not assessable for specific phases) assessing whether the model reproduces documented cooperation patterns, response magnitudes, and phase transitions.

\textbf{Statistical Validation.} We complement experimental and empirical validation with rigorous statistical analysis including paired t-tests for condition comparisons, Cohen's $d$ effect sizes for practical significance, bootstrap confidence intervals (10,000 resamples) for robustness, and Monte Carlo analysis (2,000 trials with $\pm 15\%$ parameter perturbation) for sensitivity assessment. This statistical rigor ensures that reported results are not artifacts of specific parameter choices.

\subsection{Paper Organization}

Section~\ref{sec:related} reviews related work in conceptual modeling, game theory, and multi-agent systems. Section~\ref{sec:foundations} establishes foundational concepts for reciprocity. Section~\ref{sec:recap} provides comprehensive recap of base frameworks from all three companion papers. Section~\ref{sec:formalization} presents the complete mathematical formalization of reciprocity mechanisms. Section~\ref{sec:translation} develops the translation framework from \textit{i*} to computational reciprocity models. Section~\ref{sec:istar} presents \textit{i*} models capturing reciprocity dynamics including Strategic Rationale, Strategic Dependency, and Goal models with comparative configuration analysis and a diagnostic protocol. Section~\ref{sec:equilibrium} defines Perfect Bayesian Equilibrium for sequential coopetition. Section~\ref{sec:validation} presents comprehensive parameter validation. Section~\ref{sec:empirical} provides empirical validation through the Apple iOS App Store case study. Section~\ref{sec:discussion} discusses implications, limitations, and future directions. Section~\ref{sec:synthesis} synthesizes the complete framework across all four papers. Section~\ref{sec:conclusion} concludes.

\section{Background and Related Work}
\label{sec:related}

\textit{Note: Portions of this background review adapt material from \cite{pant2021strategic} with the author's permission, providing context for the computational formalization developed in subsequent sections.}

\subsection{Conceptual Modeling Languages and Temporal Extensions}

The \textit{i*} framework, developed by Yu~\cite{yu1995modelling}, provides visual language for representing strategic actors, their goals, and dependencies. In \textit{i*}, actor $i$ (depender) depends on actor $j$ (dependee) for dependum $d$ (goal, task, or resource). This creates instrumental interdependence where $i$'s success requires $j$'s performance. Pant~\cite{pant2021strategic} established a conceptual modeling framework building on \textit{i*} to capture five critical dimensions of coopetitive relationships, including reciprocity as a key dimension requiring computational formalization. Related approaches include Tropos, Goal-oriented Requirement Language, and KAOS.

Temporal extensions address time-ordered goals and events. Temporal \textit{i*} adds temporal constraints between goals. Tropos goal models incorporate temporal operators. KAOS specifications include temporal requirements. However, these extensions add temporal ordering of goals or events but do not capture behavioral conditionality where actions at time $t$ depend on observed actions at time $t-1$. Our framework addresses this gap by formalizing conditional action rules based on history tracking.

\subsection{Game Theory of Repeated Interactions and Reciprocity}

Robert Axelrod's seminal work demonstrated evolution of cooperation through tit-for-tat strategies in iterated Prisoner's Dilemma tournaments. Folk theorems from repeated game theory show cooperation can be sustained as subgame perfect equilibria when discount factors are sufficiently high, enabling future payoffs to discipline current behavior.

Reciprocity strategies include trigger strategies where single defection triggers permanent punishment (grim trigger) or temporary punishment (forgiving strategies), win-stay lose-shift strategies, and generous tit-for-tat allowing occasional forgiveness. Behavioral game theory and strong reciprocity research document that real actors reciprocate more than pure Nash equilibrium predicts.

While this literature provides rich theoretical foundations, it typically assumes infinite memory, perfect observability, and known payoff functions. Our framework adapts these insights to bounded rationality contexts with memory windows, noisy signals, and uncertain preferences typical in conceptual modeling applications.

\subsection{Computational Trust and Reputation Systems}

Trust and reputation systems in multi-agent contexts enable actors to condition cooperation on observed reliability. Bayesian trust models represent trust as probability distributions updated through observed behavior. Reputation systems aggregate history enabling memory-based strategies.

Our approach differs by integrating trust dynamics directly with reciprocity mechanisms. Trust gates reciprocity responses where low trust dampens reciprocity even when conditional cooperation would be rational, while high trust enables full reciprocity. This integration creates unified behavioral model where structural dependencies, trust evolution, and reciprocity operate synergistically.

\subsection{Multi-Agent System Modeling}

Computational approaches to reciprocity include reinforcement learning where agents learn reciprocity strategies through trial and error, evolutionary models where populations evolve toward reciprocal strategies, and agent-based simulations.

Our framework differs by grounding reciprocity in conceptual model structures, specifically dependencies from \textit{i*} networks, rather than treating it as exogenous or learned. This enables integration with requirements engineering and system design processes where structural dependencies are elicited through stakeholder analysis.

\subsection{Reciprocity in Requirements Engineering}

Requirements engineering inherently involves sequential interactions where reciprocity patterns shape information exchange and stakeholder engagement. This subsection reviews how reciprocity manifests in RE contexts and identifies gaps that our framework addresses.

\subsubsection{Information Disclosure and Stakeholder Reciprocity}

Requirements elicitation depends fundamentally on stakeholder willingness to disclose information. Research on stakeholder dynamics documents that information disclosure follows reciprocal patterns where analysts who demonstrate responsiveness receive more complete requirements~\cite{hickey2004unified}. Stakeholders assess analyst behavior over multiple interactions and adjust their disclosure accordingly. When analysts demonstrate genuine engagement with stakeholder concerns, stakeholders reciprocate by sharing sensitive information including political constraints, hidden assumptions, and tacit knowledge that would otherwise remain inaccessible.

This reciprocity operates under bounded rationality constraints. Stakeholders do not maintain infinite memory of analyst behavior but rather assess recent interactions within cognitive limitations. A stakeholder who experienced dismissive analyst behavior three years ago may have updated their assessment based on recent positive interactions. Conversely, stakeholders with recent negative experiences may withhold information despite historical cooperation. These dynamics motivate our memory-windowed history tracking formalization.

\subsubsection{Iterative Requirements Negotiation}

Requirements negotiation exhibits sequential move structure where parties alternate proposals and responses~\cite{robinson1989integrating}. The KAOS methodology explicitly models goal refinement as iterative process where analyst proposals elicit stakeholder feedback, which then shapes subsequent proposals. Similarly, viewpoint-oriented approaches recognize that requirements emerge through negotiation among stakeholders with differing perspectives.

Win-win negotiation approaches such as those developed within the WinWin framework~\cite{boehm1994software} depend implicitly on reciprocity. When one party makes a concession, the expectation of reciprocal concession enables agreement. Without reciprocity, negotiation degenerates into positional bargaining where parties refuse concessions fearing exploitation. Our framework provides formal mechanisms to analyze when reciprocity enables cooperative negotiation outcomes and when it fails.

\subsubsection{Trust-Reciprocity Interactions in RE}

Trust and reciprocity interact in requirements engineering contexts in ways that existing frameworks do not formally capture. The Tropos methodology~\cite{bresciani2004tropos} represents trust relationships between actors but treats trust as static rather than dynamically evolving through reciprocal interactions. Similarly, secure Tropos~\cite{mouratidis2007secure} models trust requirements but does not formalize how trust evolves based on observed reciprocity.

Our companion work on trust dynamics~\cite{pant2025trust} established foundations for modeling trust evolution, showing that trust builds slowly through positive interactions but erodes rapidly following violations. The present work extends this by showing how trust gates reciprocity responses. In RE contexts, this captures phenomena such as stakeholders who remain guarded despite current analyst cooperation due to accumulated distrust from prior project failures. The trust-gated reciprocity mechanism $T_{ij}^t \cdot \rho_{ij}$ formalizes this interaction, enabling analysis of how trust history moderates reciprocal information exchange.

\subsection{Behavioral Economics of Reciprocity}

Behavioral economics provides extensive empirical evidence that human actors exhibit reciprocity patterns that deviate systematically from pure self-interest predictions. This evidence grounds our framework in observed behavioral regularities rather than idealized rationality assumptions.

\subsubsection{Strong Reciprocity and Conditional Cooperation}

Experimental economics research documents strong reciprocity, defined as the propensity to cooperate with cooperative partners and punish defectors even at personal cost~\cite{fehr2002strong}. In public goods experiments, Fehr and G\"{a}chter~\cite{fehr2000cooperation} demonstrated that subjects contribute substantially more than Nash equilibrium predictions when given opportunities for reciprocal punishment. Conditional cooperators, who constitute approximately 50\% of experimental subjects, explicitly condition their contributions on observed partner contributions.

This evidence motivates our bounded response function formalization. Strong reciprocators do not exhibit unbounded responses to partner behavior; rather, their responses are proportional within limits. The hyperbolic tangent function $\phi_{\text{recip}}(x) = \tanh(\kappa x)$ captures this bounded proportionality, ensuring extreme partner deviations produce finite responses consistent with experimental observations.

\subsubsection{Negative Reciprocity and Altruistic Punishment}

Negative reciprocity, the willingness to punish defectors at personal cost, provides the enforcement mechanism sustaining cooperation in social dilemmas~\cite{fehr1999theory}. Experimental evidence shows that punishment opportunities dramatically increase cooperation rates even when punishment is costly to the punisher and provides no direct material benefit. This altruistic punishment reflects fairness preferences where actors derive disutility from inequitable outcomes.

Our framework captures negative reciprocity through the cooperation signal mechanism. When partner behavior falls below baseline expectations ($s_{ij}^t < 0$), the bounded response function produces negative reciprocity adjustments that reduce cooperation. The magnitude of punishment depends on reciprocity sensitivity $\rho_{ij}$, which we derive from structural dependencies. Actors with high dependency on partners ($D_{ij}$ high) exhibit stronger punishment responses because violations by critical partners impose greater costs, consistent with experimental findings that punishment intensity correlates with harm severity.

\subsubsection{Reciprocity Under Bounded Rationality}

Classical reciprocity models in game theory assume actors maintain complete history and compute optimal responses. Behavioral evidence contradicts these assumptions. Camerer~\cite{camerer2003behavioral} documents systematic departures from perfect rationality including limited memory, probability weighting errors, and reference-dependent preferences. Charness and Rabin~\cite{charness2002understanding} show that reciprocity depends on perceived intentions rather than just outcomes, requiring actors to form beliefs about partner motivations.

Our memory-windowed history tracking directly addresses bounded rationality constraints. Rather than assuming infinite perfect recall, we model actors as maintaining moving averages over recent periods $\bar{a}_j^{t-k:t-1}$ with window size $k$ reflecting cognitive limitations. This formalization aligns with experimental evidence that subjects weight recent observations more heavily than distant history and that memory constraints limit strategic sophistication.

\subsubsection{Reciprocity in Strategic Interactions}

Beyond laboratory experiments, field evidence documents reciprocity in strategic business interactions relevant to requirements engineering contexts. Bowles and Gintis~\cite{bowles2011cooperative} synthesize evidence that cooperation in firms, markets, and communities depends on reciprocity norms enforced through repeated interactions. Platform ecosystems exhibit reciprocal dynamics where platform providers and complementors condition behavior on observed partner actions over time~\cite{gawer2014bridging}.

This evidence supports our framework's applicability to the platform ecosystem and requirements engineering scenarios motivating this work. The Apple iOS App Store case study (Section~\ref{sec:validation}) provides empirical validation that reciprocity patterns observed in laboratory settings manifest in large-scale strategic interactions between platform providers and developers over extended time horizons.

\subsection{Position within the Coopetition Research Program}

This paper completes a four-paper research program addressing strategic coopetition in requirements engineering and multi-agent systems. The research program provides computational realizations of the conceptual framework for strategic coopetition established in~\cite{pant2021strategic}, which identified five critical dimensions of coopetitive relationships through \textit{i*} conceptual modeling. The first paper established interdependence through \textit{i*} structural dependency analysis and complementarity through Added Value mechanisms. The second paper formalized trust as dynamic layered model with asymmetric updating, showing trust builds slowly but erodes quickly. The third paper extended to team production settings, formalizing loyalty mechanisms overcoming free-riding.

Reciprocity completes the framework by providing enforcement mechanism sustaining cooperation through repeated interactions. Without reciprocity, cooperation would be fragile where trust would erode from single defections, team production would collapse from shirking episodes, and complementary value creation would fail from partner defections. Reciprocity provides resilience by creating credible threats, specifically that if you defect then I defect, and credible promises, specifically that if you cooperate then I cooperate, that stabilize cooperation through endogenous enforcement rather than requiring external monitoring or sanctions.

\section{Foundational Concepts}
\label{sec:foundations}

We establish core concepts grounded in both conceptual modeling and game theory. The conceptual foundations for reciprocity in coopetitive relationships were identified in~\cite{pant2021strategic}, which established that cooperation in multi-stakeholder systems depends critically on history-dependent conditional behavior. Our contribution is the mathematical specification and computational validation that operationalizes these conceptual insights for quantitative analysis.

\textbf{Reciprocity} is the principle of conditional behavior in sequential interactions, where Actor $i$'s cooperation at time $t$ is contingent on Actor $j$'s observed cooperative actions in prior periods $t-1, t-2, \ldots$. Reciprocity differs from pure altruism, which involves unconditional cooperation, and pure selfishness, which involves ignoring partner behavior. It embodies strategic conditionality where cooperation is extended when reciprocated and withdrawn when violated.

\textbf{History Dependence} means current actions depend explicitly on past action profiles, creating path dependence where relationship trajectory shapes current strategic choices. This contrasts with history-independent strategies found in pure Nash equilibrium actions where past interactions are irrelevant to current decision-making.

\textbf{Bounded Memory} reflects cognitive limitation where actors track only recent interactions within memory window $k$ rather than infinite history, reflecting bounded rationality and attention constraints. Memory windows vary by context where fast-moving business environments have shorter windows encompassing recent quarters, while stable long-term partnerships have longer windows encompassing recent years.

\textbf{Cooperation Signal} provides quantitative assessment of whether partner's action constitutes cooperation, meaning exceeding baseline expectations, or defection, meaning falling short. The signal is formalized as deviation from normative baseline passed through bounded response function to prevent unbounded trust changes. Cooperation signals transform qualitative judgments into quantitative measures enabling mathematical analysis.

\textbf{Reciprocity Sensitivity} captures magnitude of Actor $i$'s response to Actor $j$'s behavior, modulated by structural dependency $D_{ij}$ and base reciprocity tendency $\rho_0$. Reciprocity sensitivity captures that violations by critical partners elicit stronger responses than violations by marginal actors, connecting behavior to structure.

\textbf{Bounded Response Functions} are mathematical functions mapping unbounded deviations onto bounded intervals, capturing that extreme partner behavior produces finite rather than infinite responses. Bounded functions such as hyperbolic tangent prevent unrealistic escalation spirals while maintaining proportional responses in moderate ranges.

\textbf{Trust-Gated Reciprocity} is the mechanism where trust level $T_{ij}$ multiplies reciprocity responses, such that low trust dampens reciprocity even when conditional cooperation is strategically rational, while high trust enables full reciprocity. This captures realistic behavior where actors in low-trust relationships are reluctant to reciprocate positively due to fear of exploitation.

\section{Recap of Foundational Framework}
\label{sec:recap}

We briefly recap the foundational frameworks from our three companion technical reports to establish notation and provide self-containment. Readers seeking full details should consult those works. This recap is organized into three subsections corresponding to the three previous papers, providing comprehensive coverage to ensure this technical report can be understood without constantly referencing the companion works.

\textbf{Scope of this summary}: We present the interdependence matrix formalization, value creation functions, private payoff structure, and base utility function from~\cite{pant2025foundations}, trust dynamics and reputation evolution from~\cite{pant2025trust}, and team production with loyalty mechanisms from~\cite{pant2025teams}, to enable self-contained reading of this report. The contributions of this report (reciprocity mechanisms, memory-windowed history tracking, trust-gated reciprocity, and the complete coopetition framework synthesis) build upon these foundations and are presented in subsequent sections.

\subsection{Foundational Concepts from TR-2025-01: Interdependence and Complementarity}

Our foundational work on interdependence and complementarity established the core mathematical framework for analyzing strategic coopetition. This framework provides the structural foundation upon which trust dynamics, team production, and reciprocity mechanisms build.

Consider a system of $N$ actors indexed by $i \in \{1, \ldots, N\}$, where each actor $i$ chooses action $a_i \in \mathbb{R}_+$ from continuous action space representing investment levels or resource allocations. An action profile $\vect{a} = (a_1, \ldots, a_N)$ represents all actors' simultaneous choices.

\textbf{The Interdependence Matrix.} The interdependence matrix $D$ is an $N \times N$ matrix where $D_{ij} \in [0,1]$ represents the structural dependency of actor $i$ on actor $j$. This coefficient quantifies how much actor $i$'s outcome depends on actor $j$'s actions through instrumental coupling rather than preference-based altruism. As demonstrated in our foundational work, these coefficients are computed from \textit{i*} dependency networks through aggregation of goal importance weights, dependency indicators, and criticality factors capturing alternatives available:

\begin{equation}
\label{eq:interdependence_recap}
D_{ij} = \frac{\sum_{d \in \mathcal{D}_i} w_d \cdot \text{Dep}(i,j,d) \cdot \text{crit}(i,j,d)}{\sum_{d \in \mathcal{D}_i} w_d} \quad \text{[From \cite{pant2025foundations}, Eq. 1]}
\end{equation}

where $w_d$ represents importance weight for dependum $d$, $\text{Dep}(i,j,d) \in \{0,1\}$ indicates dependency existence, and $\text{crit}(i,j,d) \in [0,1]$ captures criticality based on whether alternatives exist. High $D_{ij}$ indicates $i$ depends critically on $j$ creating high vulnerability, while low $D_{ij}$ indicates minimal dependency. The matrix is generally asymmetric where $D_{ij} \neq D_{ji}$, reflecting that dependency relationships are directional.

\begin{figure}[htbp]
\centering
\begin{tikzpicture}
\begin{axis}[
    width=0.48\textwidth, height=5.5cm,
    xlabel={\scriptsize Complementarity $\gamma$},
    ylabel={\scriptsize Equilibrium Investment $a^*$},
    xmin=0, xmax=1, ymin=0, ymax=50,
    tick label style={font=\tiny},
    legend style={at={(0.5,-0.25)}, anchor=north, font=\tiny, legend columns=3, draw=gray!50},
    title={\scriptsize (a) Complementarity drives cooperation},
    grid=major, grid style={gray!20},
    name=plotA,
]
\addplot[cooperationblue, thick, mark=none, domain=0:1, samples=30] {5 + 40*x^0.75};
\addlegendentry{$\beta=0.75$}
\addplot[effortorange, thick, dashed, mark=none, domain=0:1, samples=30] {3 + 30*x^0.5};
\addlegendentry{$\beta=0.50$}
\addplot[outputgreen, thick, dotted, mark=none, domain=0:1, samples=30] {8 + 42*x^0.9};
\addlegendentry{$\beta=0.90$}
\draw[defectcolor, thick, dashed] (axis cs:0.65,0) -- (axis cs:0.65,50) node[pos=0.92, right=4pt, font=\tiny, fill=white, inner sep=1pt] {$\gamma^*{=}0.65$};
\end{axis}
\begin{axis}[
    width=0.48\textwidth, height=5.5cm,
    at={(plotA.east)}, anchor=west, xshift=1.5cm,
    xlabel={\scriptsize Endowment $e_i$},
    ylabel={\scriptsize Cooperation Ratio $a^*/e_i$},
    xmin=0, xmax=110, ymin=0, ymax=1,
    tick label style={font=\tiny},
    legend style={at={(0.5,-0.25)}, anchor=north, font=\tiny, legend columns=3, draw=gray!50},
    title={\scriptsize (b) Scale invariance validation},
    grid=major, grid style={gray!20},
]
\addplot[cooperationblue, thick, mark=square*, mark size=2pt] coordinates {(10,0.72) (25,0.71) (50,0.73) (75,0.72) (100,0.71)};
\addlegendentry{$\gamma=0.65$}
\addplot[effortorange, thick, mark=*, mark size=2pt] coordinates {(10,0.45) (25,0.44) (50,0.46) (75,0.45) (100,0.44)};
\addlegendentry{$\gamma=0.30$}
\addplot[gray, dashed, thick] coordinates {(0,0.72) (110,0.72)};
\addlegendentry{Mean level}
\end{axis}
\end{tikzpicture}
\caption{Complementarity-driven cooperation foundation from TR-2025-01~\cite{pant2025foundations}. Panel (a) demonstrates how equilibrium investment increases with complementarity strength $\gamma$ across effort elasticity values, with the validated choice $\gamma^*=0.65$ marked (dashed vertical line). Panel (b) validates scale invariance: the cooperation ratio $a^*/e_i$ remains approximately constant across endowment values from 10 to 100 (dashed gray line shows mean level), confirming that complementarity effects are robust across resource scales.}
\label{fig:complementarity}
\end{figure}
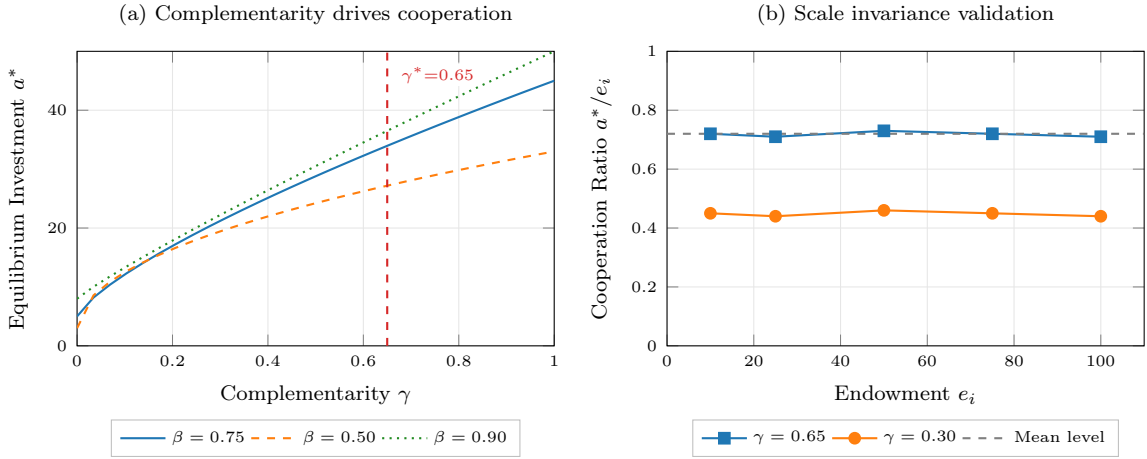

\textbf{Value Creation and Complementarity.} The value creation function represents total value generated before distribution, combining individual contributions with synergistic value. Our foundational work validated the functional form:

\begin{equation}
\label{eq:value_creation_recap}
V(\vect{a} \mid \gamma) = \sum_{i=1}^N f_i(a_i) + \gamma \cdot g(a_1, \ldots, a_N) \quad \text{[From \cite{pant2025foundations}, Eq. 2]}
\end{equation}

Our foundational work established the empirical robustness of logarithmic specifications ($f_i(a_i) = \theta \ln(1+a_i)$ with $\theta = 20$) for manufacturing joint ventures, achieving 58/60 (96.7\%) under strict historical alignment scoring validated across 22,000+ experimental trials compared to power functions ($f_i(a_i) = a_i^{\beta}$ with $\beta = 0.75$, achieving 46/60, 76.7\%) for the Samsung-Sony S-LCD case~\cite{pant2025foundations}. This report adopts the power specification with $\beta = 0.75$ for analytical tractability in sequential settings, as shown in Figure~\ref{fig:complementarity}. Both specifications use geometric mean synergy $g = (a_1 \cdots a_N)^{1/N}$ ensuring all actors must contribute for synergy to exist. The complementarity parameter $\gamma \geq 0$ controls superadditivity where higher $\gamma$ indicates stronger synergistic effects from collaboration.

This value function exhibits the key property of complementarity, specifically that joint action creates more value than the sum of individual contributions when $\gamma > 0$. This superadditivity following Brandenburger and Nalebuff's Added Value concept explains why actors have incentive to collaborate rather than operate independently.

\textbf{Value Appropriation.} The private payoff function captures value appropriation where actors appropriate their individual value contributions and negotiate shares of synergistic value based on bargaining power:

\begin{equation}
\label{eq:private_payoff_recap}
\pi_i(\vect{a}) = e_i - a_i + f_i(a_i) + \alpha_i \left[V(\vect{a}) - \sum_{j=1}^{N} f_j(a_j)\right] \quad \text{[From \cite{pant2025foundations}, Eq. 11]}
\end{equation}

where $e_i$ is initial endowment, actors bear their investment costs, appropriate their individual value creation, and receive share $\alpha_i$ of synergistic value determined by bargaining power. The shares are normalized such that $\sum_i \alpha_i = 1$ ensuring complete allocation of synergistic value.

\textbf{Base Integrated Utility.} The base utility function from our foundational work incorporates interdependence through dependency-weighted partner outcomes:

\begin{equation}
\label{eq:base_utility_recap}
U_i^{\text{base}}(\vect{a}) = \pi_i(\vect{a}) + \sum_{j \neq i} D_{ij} \pi_j(\vect{a}) \quad \text{[From \cite{pant2025foundations}, Eq. 13]}
\end{equation}

This utility formulation extends classical game theory by incorporating dependency-based other-regarding preferences derived from organizational structure. The interdependence terms $\sum_{j \neq i} D_{ij} \pi_j(\vect{a})$ create rational incentives for actor $i$ to consider actor $j$'s payoff proportional to structural dependency $D_{ij}$, not through altruism but through instrumental coupling where $j$'s success enables $i$'s goal achievement.

\textbf{Coopetitive Equilibrium.} The Coopetitive Equilibrium from~\cite{pant2025foundations} is defined as Nash Equilibrium where each actor maximizes dependency-augmented utility:
\begin{equation}
\label{eq:coopetitive_equilibrium_recap}
\vect{a}^* \text{ is Coopetitive Equilibrium if } a_i^* \in \argmax_{a_i \geq 0} U_i^{\text{base}}(a_i, \vect{a}_{-i}^*) \quad \forall i \quad \text{[From \cite{pant2025foundations}]}
\end{equation}
This extends Nash Equilibrium by incorporating dependency-based utilities, showing how structural interdependence shifts equilibria toward more cooperative outcomes compared to purely self-interested Nash equilibrium. The foundational work demonstrated through comprehensive validation that higher interdependence and stronger complementarity jointly produce substantially higher equilibrium cooperation levels.

\subsection{Trust Dynamics from TR-2025-02: Evolving Partnerships}

The companion technical report on trust dynamics extended the base framework with temporal structure where trust state variables evolve based on observed actions over repeated interactions. This extension enables analysis of how cooperation emerges and stabilizes through dynamic trust building, or how relationships deteriorate through trust erosion and reputation damage.

\textbf{Two-Layer Trust Structure.} Actor $i$ maintains two dynamic state variables for each other actor $j$. Immediate trust $T_{ij}^t \in [0,1]$ represents current confidence in actor $j$ at time $t$ based on recent observable behavior, providing fast-moving response to current events. Reputation damage $R_{ij}^t \in [0,1]$ represents accumulated violation history at time $t$, where $R_{ij}^t = 0$ means pristine reputation and $R_{ij}^t = 1$ means completely damaged reputation, providing slow-moving memory of past events. This two-layer decomposition captures temporal structure of trust by simultaneously representing current assessment and long-term constraints.

\begin{figure}[htbp]
\centering
\begin{tikzpicture}
\begin{axis}[
    width=0.48\textwidth, height=5.5cm,
    xlabel={\scriptsize Period $t$},
    ylabel={\scriptsize Trust $T_{ij}^t$},
    xmin=0, xmax=100, ymin=0, ymax=1,
    tick label style={font=\tiny},
    legend style={at={(0.5,-0.22)}, anchor=north, font=\tiny, legend columns=2, column sep=6pt, draw=gray!50},
    title={\scriptsize (a) Trust hysteresis after violation at $t{=}50$},
    grid=major, grid style={gray!20},
    name=plotA,
]
\addplot[cooperationblue, thick, mark=none] coordinates {
(0,0.50) (5,0.55) (10,0.60) (15,0.64) (20,0.68) (25,0.71) (30,0.74)
(35,0.76) (40,0.78) (45,0.79) (50,0.80)
};
\addplot[defectcolor, thick, mark=none] coordinates {(50,0.80) (51,0.56)};
\addplot[cooperationblue, thick, mark=none] coordinates {
(51,0.56) (55,0.60) (60,0.64) (65,0.68) (70,0.71) (75,0.74)
(80,0.76) (85,0.78) (90,0.80) (95,0.82) (100,0.83)
};
\addplot[gray, thick, dashed, mark=none] coordinates {
(50,0.80) (55,0.82) (60,0.84) (65,0.86) (70,0.87) (75,0.89)
(80,0.90) (85,0.91) (90,0.92) (95,0.93) (100,0.93)
};
\addlegendentry{Actual trust}
\addlegendentry{Violation}
\addlegendentry{Recovery}
\addlegendentry{Counterfactual}
\draw[<->, thick, recippurple] (axis cs:96,0.83) -- (axis cs:96,0.93)
    node[midway, left, font=\tiny, fill=none, inner sep=2.0pt] {Gap};
\end{axis}
\begin{axis}[
    width=0.48\textwidth, height=5.5cm,
    at={(plotA.east)}, anchor=west, xshift=1.5cm,
    xlabel={\scriptsize Trust Change Scenario},
    ylabel={\scriptsize Rate of Change},
    tick label style={font=\tiny},
    title={\scriptsize (b) 3:1 negativity bias in trust updating},
    grid=major, grid style={gray!20},
    ybar,
    bar width=18pt,
    ymin=0, ymax=0.38,
    xtick={1,2},
    xticklabels={Building ($\lambda_+$), Erosion ($\lambda_-$)},
    x tick label style={font=\tiny, align=center},
    nodes near coords,
    nodes near coords style={font=\tiny\bfseries},
    enlarge x limits=0.5,
]
\addplot[fill=cooperationblue!50, draw=cooperationblue] coordinates {(1,0.10)};
\addplot[fill=defectcolor!50, draw=defectcolor] coordinates {(2,0.30)};
\draw[<->, very thick, recippurple] (axis cs:0.9,0.14) -- (axis cs:0.9,0.30)
    node[pos=0.5, right=2pt, font=\tiny\bfseries, fill=white, inner sep=1pt] {$3\times$};
\end{axis}
\end{tikzpicture}
\caption{Trust asymmetry and hysteresis dynamics from TR-2025-02~\cite{pant2025trust}. Panel (a) demonstrates trust hysteresis following a violation at period 50: trust drops sharply from 0.80 to 0.56, then recovers slowly but cannot reach the counterfactual trajectory (dashed gray), creating a persistent gap representing lasting relationship damage. Panel (b) shows the validated 3:1 negativity bias: trust erosion rate ($\lambda_- = 0.30$) is three times the building rate ($\lambda_+ = 0.10$), consistent with empirical findings from behavioral economics.}
\label{fig:trust}
\end{figure}
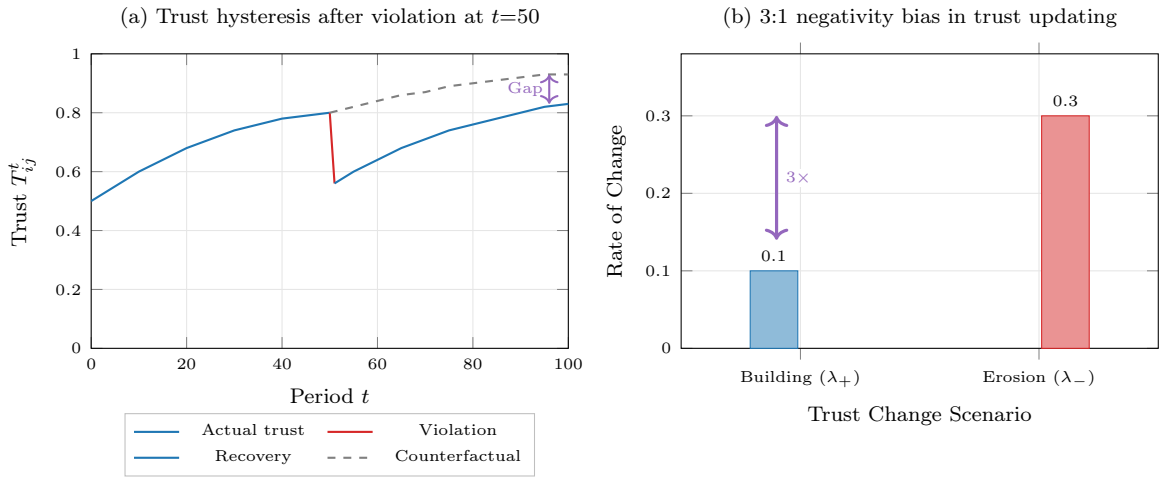

\textbf{Asymmetric Trust Evolution.} Trust updates based on cooperation signals through asymmetric learning rules reflecting negativity bias observed in behavioral economics, as illustrated in Figure~\ref{fig:trust}. When cooperation signals are positive, trust builds proportionally toward the remaining gap below the trust ceiling at rate $\lambda_+$. When cooperation signals are negative or violations occur, trust erodes multiplicatively based on current trust level at rate $\lambda_-$, with interdependence amplification factor $(1 + \xi D_{ij})$, where $\lambda_- \gg \lambda_+$ ensures trust erodes much faster than it builds. The validated parameters achieve negativity ratio of three to one, specifically $\lambda_-/\lambda_+ = 0.30/0.10 = 3.0$, aligning with empirical findings that trust builds slowly but can be destroyed quickly.

The trust evolution dynamics are given by:

\begin{equation}
\label{eq:trust_evolution_recap}
\Delta T_{ij}^t = \begin{cases}
\lambda_+ \cdot s_{ij}^t \cdot \max(0,\; \Theta_{ij}^t - T_{ij}^t) & \text{if } s_{ij}^t > 0 \text{ (trust building)} \\
\lambda_- \cdot s_{ij}^t \cdot T_{ij}^t \cdot (1 + \xi \cdot D_{ij}) & \text{if } s_{ij}^t \leq 0 \text{ (trust erosion)}
\end{cases} \quad \text{[From \cite{pant2025trust}, Eqs. 8--9]}
\end{equation}

where $s_{ij}^t$ is cooperation signal assessing partner behavior relative to baseline, $\lambda_+ \approx 0.10$ is trust building rate, $\lambda_- \approx 0.30$ is trust erosion rate, $\xi \approx 0.50$ is interdependence amplification factor, and $(1 + \xi \cdot D_{ij})$ amplifies erosion for critical partners showing that violations by actors one depends on cause disproportionately severe trust damage.

\textbf{Trust Ceilings and Hysteresis.} Accumulated reputation damage constrains maximum achievable trust through trust ceiling mechanism:

\begin{equation}
\label{eq:trust_ceiling_recap}
\Theta_{ij}^t = 1 - R_{ij}^t \quad \text{[From \cite{pant2025trust}]}
\end{equation}

where $R_{ij}^t \in [0,1]$ is accumulated reputation damage. When reputation damage is high ($R_{ij}^t$ near 1), the trust ceiling $\Theta_{ij}^t$ approaches zero, preventing trust from rebuilding even with sustained cooperation. This creates hysteresis: relationships with violation history cannot fully recover to pristine states, as demonstrated in panel~(a) of Figure~\ref{fig:trust}. The companion work demonstrated through experimental validation that trust recovers to maximum fifty-four percent of original level after violation, confirming lasting relationship damage that cannot be fully repaired. In our case study (Section~\ref{sec:empirical}), we extend this ceiling with two parameters (maximum trust under pristine reputation $T_{\max}$ and reputation sensitivity $\theta_R$) as $\Theta_{ij}^t = \min(T_{\max}, 1 - \theta_R \cdot R_{ij}^t)$, enabling calibration to specific empirical contexts.

\textbf{Reputation Evolution.} Reputation accumulates violation history and decays slowly through cooperation:

\begin{equation}
\label{eq:reputation_recap}
\Delta R_{ij}^t = \begin{cases}
\mu_R \cdot (-s_{ij}^t) \cdot (1 - R_{ij}^t) & \text{if } s_{ij}^t < 0 \text{ (violation: damage accumulates)} \\
-\delta_R \cdot R_{ij}^t & \text{if } s_{ij}^t \geq 0 \text{ (non-violation: damage decays)}
\end{cases} \quad \text{[From \cite{pant2025trust}, Eqs. 10--11]}
\end{equation}

where $\mu_R \approx 0.60$ is reputation damage severity per violation, $\delta_R \approx 0.03$ is reputation decay rate for gradual forgetting, and $(1 - R_{ij}^t)$ is available reputation space. During violations ($s_{ij}^t < 0$), damage accumulates proportional to violation severity $(-s_{ij}^t > 0)$ and remaining capacity. During non-violation periods ($s_{ij}^t \geq 0$), existing damage decays slowly, modeling gradual organizational forgetting. The slow decay rate ($\delta_R \approx 0.03$ implies approximately 33 periods for substantial recovery) creates persistent memory effects consistent with empirical observations that organizational reputations are slow to rebuild.

\textbf{Trust-Gated Reciprocity.} Trust integrates into strategic decision-making through trust-gated reciprocity mechanisms where trust multiplies reciprocity responses. When trust is low, reciprocity terms vanish regardless of reciprocity signals, preventing cooperation under low-trust conditions. When trust is high, reciprocity operates at full strength enabling conditional cooperation. The companion work formalized this through trust-weighted utility augmentation that we integrate with reciprocity mechanisms in this paper.

\subsection{Team Production from TR-2025-03: Loyalty and Free-Riding}

The third companion technical report extended the individual actor framework to handle complex actors with internal structure, specifically teams. While teams act as unified actors externally in coopetitive relationships with other actors or teams, internally team members face coordination problems including free-riding incentives. This extension enables analysis of multi-level systems where both inter-team and intra-team dynamics matter.

\textbf{Team Structure and Production.} A team $C$ consisting of $n$ members indexed by $i \in C$ produces collective value through team production function:

\begin{equation}
\label{eq:team_production_recap}
Q(\vect{a}) = \omega \left(\sum_{i=1}^{n} a_i\right)^\beta \quad \text{[From \cite{pant2025teams}, Eq. 8]}
\end{equation}

where $a_i$ represents action (effort) level chosen by member $i$, $\omega$ is productivity factor, and $\beta \in (0,1)$ is effort elasticity capturing diminishing returns. Each member bears individual cost $c \cdot a_i$ where $c > 0$ is the per-unit effort cost, creating the fundamental tension between shared output and private cost that drives free-riding.

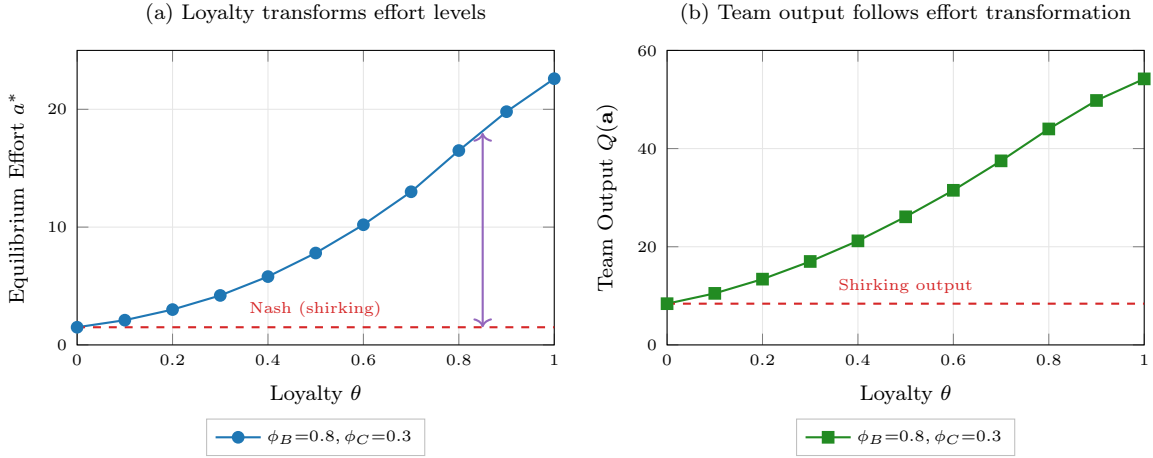
\begin{figure}[htbp]
\centering
\begin{tikzpicture}
\begin{axis}[
    width=0.48\textwidth, height=5.5cm,
    xlabel={\scriptsize Loyalty $\theta$},
    ylabel={\scriptsize Equilibrium Effort $a^*$},
    xmin=0, xmax=1, ymin=0, ymax=25,
    tick label style={font=\tiny},
    legend style={at={(0.5,-0.25)}, anchor=north, font=\tiny, draw=gray!50},
    title={\scriptsize (a) Loyalty transforms effort levels},
    grid=major, grid style={gray!20},
    name=plotA,
]
\addplot[loyaltyblue, thick, mark=*, mark size=2pt] coordinates {
(0.0,1.5) (0.1,2.1) (0.2,3.0) (0.3,4.2) (0.4,5.8) (0.5,7.8) (0.6,10.2) (0.7,13.0) (0.8,16.5) (0.9,19.8) (1.0,22.6)
};
\addlegendentry{$\phi_B{=}0.8,\phi_C{=}0.3$}
\draw[defectcolor, dashed, thick] (axis cs:0,1.5) -- (axis cs:1,1.5) node[pos=0.50, above, font=\tiny] {Nash (shirking)};
\draw[<->, thick, recippurple] (axis cs:0.85,1.5) -- (axis cs:0.85,18.0);
\end{axis}
\begin{axis}[
    width=0.48\textwidth, height=5.5cm,
    at={(plotA.east)}, anchor=west, xshift=1.5cm,
    xlabel={\scriptsize Loyalty $\theta$},
    ylabel={\scriptsize Team Output $Q(\vect{a})$},
    xmin=0, xmax=1, ymin=0, ymax=60,
    tick label style={font=\tiny},
    legend style={at={(0.5,-0.25)}, anchor=north, font=\tiny, draw=gray!50},
    title={\scriptsize (b) Team output follows effort transformation},
    grid=major, grid style={gray!20},
]
\addplot[outputgreen, thick, mark=square*, mark size=2pt] coordinates {
(0.0,8.4) (0.1,10.5) (0.2,13.4) (0.3,17.0) (0.4,21.2) (0.5,26.1) (0.6,31.5) (0.7,37.5) (0.8,44.0) (0.9,49.8) (1.0,54.2)
};
\addlegendentry{$\phi_B{=}0.8,\phi_C{=}0.3$}
\draw[defectcolor, dashed, thick] (axis cs:0,8.4) -- (axis cs:1,8.4) node[pos=0.50, above, font=\tiny] {Shirking output};
\end{axis}
\end{tikzpicture}
\caption{Loyalty mechanisms in team production from TR-2025-03~\cite{pant2025teams}. Panel (a) shows how equilibrium effort increases with loyalty parameter $\theta$, creating approximately 15-fold differentiation between zero-loyalty (Nash shirking equilibrium) and high-loyalty teams, consistent with the validated 15.04$\times$ median effort differentiation across 3,125 parameter configurations. Panel (b) demonstrates the corresponding impact on team output, which increases approximately 6.4-fold. These findings are essential for understanding how teams engage in reciprocal relationships with external partners while managing internal coordination through loyalty.}
\label{fig:teamloyalty}
\end{figure}

\textbf{The Free-Riding Problem.} Under pure self-interest without loyalty, member $i$'s utility is share of team value minus cost: $U_i^{\text{selfish}} = \frac{1}{n} Q(\vect{a}) - c \cdot a_i$. Each member receives equal share regardless of individual contribution, creating asymmetry between benefit distribution, which is shared, and cost bearing, which is individual. This asymmetry creates rational incentive to minimize effort while benefiting from others' work. The companion work demonstrated mathematically that Nash equilibrium under pure self-interest is universal minimal effort, confirming Holmström's theoretical prediction that free-riding emerges from team production structure.

\textbf{Loyalty Mechanisms.} Loyalty quantified by parameter $\theta_{i|C} \in [0,1]$ moderates utility through two consolidated mechanisms, as validated in Figure~\ref{fig:teamloyalty}. \textit{Cost tolerance} $\phi_C$ reduces perceived effort costs by factor $(1 - \phi_C \theta_i)$, reflecting psychological framing where loyal members view contributions as investments in valued relationships. \textit{Loyalty benefit} $\phi_B$ weights teammates' aggregate payoff in the member's utility by $\phi_B \theta_i$, consolidating welfare internalization and warm glow into a single prosocial parameter. The companion work validated these parameters at $\phi_B = 0.8$ and $\phi_C = 0.3$, demonstrating that the two-parameter formulation captures loyalty's full effect on effort with predictive accuracy. A finer-grained four-mechanism decomposition (welfare internalization, warm glow, cost tolerance, guilt aversion) is available in the companion work's appendix for practitioners requiring granular mechanism analysis.

\textbf{Complete Member Utility.} Expanding and rearranging produces the team member utility:

\begin{equation}
\label{eq:team_utility_recap}
U_i(\vect{a}; \theta_i) = \frac{1}{n} Q(\vect{a}) - c(1 - \phi_C \theta_i) a_i + \phi_B \theta_i \cdot \bar{\pi}_{-i}(\vect{a}) \quad \text{[From \cite{pant2025teams}, Eq. 12]}
\end{equation}

\textbf{Component Interpretation:} \textit{Output share} $\frac{1}{n}Q(\vect{a})$ gives member $i$'s portion of team production, identical to the pure self-interest case. \textit{Loyalty-reduced cost} $c(1 - \phi_C \theta_i) a_i$ reflects effort cost reduced by factor $(1 - \phi_C \theta_i)$ for loyal members: with $\phi_C = 0.3$ and $\theta_i = 0.9$, effective cost coefficient is $0.73c$. \textit{Loyalty benefit} $\phi_B \theta_i \cdot \bar{\pi}_{-i}(\vect{a})$ represents utility from teammates' aggregate payoff weighted by loyalty and benefit strength: with $\phi_B = 0.8$ and $\theta_i = 0.9$, teammate welfare weight is $0.72$. The consolidated two-parameter formulation ($\phi_B$, $\phi_C$) presented here provides equivalent predictive power to the finer four-mechanism decomposition in the companion work's appendix and is the version validated computationally and empirically~\cite{pant2025teams}.

\textbf{Multi-Level Equilibrium.} The Team Production Equilibrium requires effort profile where each member's effort maximizes their loyalty-augmented utility given teammates' efforts. The companion work demonstrated through comprehensive validation that loyalty creates approximately 15-fold effort differentiation between teams with zero loyalty experiencing universal free-riding and teams with perfect loyalty achieving sustained high cooperation, with corresponding 6.4-fold team output gains, as shown in panels (a) and (b) of Figure~\ref{fig:teamloyalty}. The framework showed perfect parameter robustness with all six parameters maintaining behavioral targets under perturbations.

\textbf{Teams in Coopetitive Relationships.} Teams participate in broader coopetitive relationships with other teams or individual actors while managing internal coordination. When teams engage in sequential interactions with partners, reciprocity mechanisms operate at team level where collective team behavior at time $t$ depends on partner behavior at time $t-1$. This creates multi-level dynamics where teams must simultaneously maintain internal cooperation through loyalty and engage in external reciprocity relationships, which this paper formalizes.

\subsection{Extensions in This Report}

We extend this foundational framework by: (1) formalizing bounded reciprocity response functions that condition current cooperation on observed partner behavior, (2) introducing memory-windowed history tracking for sequential interactions, (3) developing structural reciprocity sensitivity where interdependence amplifies reciprocal responses, (4) integrating trust-gated reciprocity where trust modulates conditional cooperation, and (5) completing the unified coopetition framework synthesizing all four dimensions. These extensions are the contributions of this report and are detailed in the following sections.

\section{Mathematical Formalization}
\label{sec:formalization}

We now develop complete mathematical formalization of reciprocity mechanisms, building upon the foundational frameworks recapped in the previous section. This section presents formal definitions, derives key properties, and establishes the mathematical foundations for computational implementation.

\subsection{Temporal Structure for Sequential Interactions}

Time is discrete, indexed by $t \in \{0,1,2,\ldots\}$. At each time $t$, each actor $i$ chooses action $a_i^t$ from action space $A_i$. The action profile at time $t$ is $\vect{a}^t = (a_1^t,\ldots,a_N^t)$.

The history up to time $t-1$ is:
\begin{equation}
h^{t-1} = \{\vect{a}^0, \vect{a}^1, \ldots, \vect{a}^{t-1}\}
\end{equation}
representing the sequence of all past action profiles. Actors choose actions at time $t$ based on current state and history $h^{t-1}$, creating history-dependent strategies.

\subsection{Formal Definitions}

We establish precise definitions for the core constructs underlying the reciprocity framework. These definitions provide the mathematical foundation for subsequent analysis and computational implementation.

\begin{definition}[Bounded Memory Window]
\label{def:memory_window}
A \textit{bounded memory window} of length $k \geq 1$ for actor $i$ observing actor $j$ is the finite history $h_{ij}^{t-k:t-1} = \{a_j^{t-k}, a_j^{t-k+1}, \ldots, a_j^{t-1}\}$ of $j$'s most recent $k$ actions that $i$ considers when computing reciprocity responses. For $t \leq k$, the memory window contains all available history $\{a_j^{1}, \ldots, a_j^{t-1}\}$.
\end{definition}

The bounded memory window captures cognitive limitations where actors cannot maintain infinite perfect recall. This formalization aligns with experimental evidence from behavioral economics documenting that subjects weight recent observations more heavily than distant history~\cite{camerer2003behavioral}.

\begin{definition}[Cooperation Signal]
\label{def:coop_signal}
The \textit{cooperation signal} $s_{ij}^t$ measures actor $j$'s deviation from baseline expectations as observed by actor $i$ at time $t$:
\begin{equation}
\label{eq:coop_signal}
s_{ij}^t = a_j^t - \bar{a}_j^{t-k:t-1}
\end{equation}
where $\bar{a}_j^{t-k:t-1} = \frac{1}{\min(k, t-1)} \sum_{\tau=\max(1,t-k)}^{t-1} a_j^\tau$ is the moving average of $j$'s actions over the memory window. When $s_{ij}^t > 0$, actor $j$ cooperates above their recent norm; when $s_{ij}^t < 0$, actor $j$ defects below their recent norm.
\end{definition}

The cooperation signal provides a self-referential assessment where actors are judged against their own recent history rather than external standards. This formulation has the desirable property of adapting to changing conditions: if all actors shift behavior, baselines shift accordingly.

\begin{definition}[Bounded Response Function]
\label{def:bounded_response}
A \textit{bounded response function} $\phi_{\text{recip}}: \mathbb{R} \to (-1, 1)$ maps unbounded cooperation signals to bounded reciprocity responses. We specify the hyperbolic tangent form:
\begin{equation}
\label{eq:bounded_response_def}
\phi_{\text{recip}}(x) = \tanh(\kappa x) = \frac{e^{\kappa x} - e^{-\kappa x}}{e^{\kappa x} + e^{-\kappa x}}
\end{equation}
where $\kappa > 0$ is the response sensitivity parameter controlling how strongly actors respond to deviations.
\end{definition}

The bounded response function ensures that extreme partner deviations produce finite responses, preventing unrealistic escalation spirals. The function satisfies key properties: (i) monotonicity ($\phi_{\text{recip}}'(x) > 0$ for all $x$); (ii) symmetry ($\phi_{\text{recip}}(-x) = -\phi_{\text{recip}}(x)$); (iii) linearity near origin ($\phi_{\text{recip}}(x) \approx \kappa x$ for $|x| \ll 1$); and (iv) saturation ($\lim_{x \to \pm\infty} \phi_{\text{recip}}(x) = \pm 1$).

\begin{definition}[Reciprocity Sensitivity]
\label{def:recip_sensitivity}
The \textit{reciprocity sensitivity} $\rho_{ij} \geq 0$ quantifies how strongly actor $i$ responds to actor $j$'s behavioral deviations, derived from structural dependencies in the interdependence matrix:
\begin{equation}
\label{eq:recip_sensitivity_def}
\rho_{ij} = \rho_0 \cdot D_{ij}^\eta
\end{equation}
where $\rho_0 > 0$ is the base reciprocity tendency, $D_{ij} \in [0,1]$ is the structural dependency from the interdependence matrix, and $\eta \geq 1$ is the dependency elasticity parameter.
\end{definition}

Reciprocity sensitivity emerges from structural dependencies rather than being treated as exogenous. High dependency $D_{ij}$ means actor $i$ is vulnerable to $j$'s behavior, amplifying reciprocity responses. This connects behavioral dynamics to structural positions derived from \textit{i*} network analysis.

\begin{definition}[Reciprocity Types]
\label{def:recip_types}
Reciprocity manifests in four distinct forms based on valence and direction:
\begin{enumerate}
\item \textit{Positive direct reciprocity}: Actor $i$ increases cooperation toward actor $j$ in response to $j$'s cooperation toward $i$ ($s_{ij}^t > 0 \Rightarrow a_i^{t+1} \uparrow$).
\item \textit{Negative direct reciprocity}: Actor $i$ decreases cooperation toward actor $j$ in response to $j$'s defection toward $i$ ($s_{ij}^t < 0 \Rightarrow a_i^{t+1} \downarrow$).
\item \textit{Positive indirect reciprocity}: Actor $i$ increases cooperation toward actor $j$ in response to $j$'s observed cooperation toward third parties.
\item \textit{Negative indirect reciprocity}: Actor $i$ decreases cooperation toward actor $j$ in response to $j$'s observed defection toward third parties.
\end{enumerate}
\end{definition}

Table~\ref{tab:recip_types} summarizes the reciprocity types and their manifestations in requirements engineering and multi-agent system contexts.

\begin{table}[htbp]
\centering
\caption{Reciprocity types with domain interpretations}
\label{tab:recip_types}
\begin{tabular}{llll}
\toprule
\textbf{Type} & \textbf{Signal} & \textbf{RE Interpretation} & \textbf{MAS Interpretation} \\
\midrule
Positive Direct & $s_{ij}^t > 0$ & Analyst rewards disclosure & Agent increases cooperation weight \\
Negative Direct & $s_{ij}^t < 0$ & Stakeholder withholds after dismissal & Agent reduces collaboration \\
Positive Indirect & $s_{jk}^t > 0$ & Trust spreads via reputation & Reputation-based cooperation \\
Negative Indirect & $s_{jk}^t < 0$ & Warnings propagate distrust & Blacklisting defectors \\
\bottomrule
\end{tabular}
\end{table}

\subsection{Reciprocity Response Function}

For actors $i$ and $j$, Actor $i$'s reciprocity response to Actor $j$ at time $t$ is defined as:

\begin{equation}
\label{eq:reciprocity_response}
R_{ij}(\vect{a}^t, h^{t-1}) = \rho_{ij} \cdot \phi_{\text{recip}}\left(a_j^t - \bar{a}_j^{t-k:t-1}\right)
\end{equation}

\textbf{Notation and Interpretation:}

$R_{ij}(\vect{a}^t, h^{t-1})$ denotes reciprocity response, which should be distinguished from reputation $R_{ij}^t$ from trust model. Context makes clear which is referenced where reputation is state variable while reciprocity response is function.

$\vect{a}^t$ represents current action profile at time $t$.

$h^{t-1}$ represents complete history of all actions up to time $t-1$.

$\bar{a}_j^{t-k:t-1}$ denotes moving average of Actor $j$'s actions over recent memory window:
\begin{equation}
\label{eq:moving_average}
\bar{a}_j^{t-k:t-1} = \frac{1}{\min(k, t-1)} \sum_{\tau=\max(1,t-k)}^{t-1} a_j^\tau
\end{equation}

$k \geq 1$ represents memory window length parameter capturing bounded rationality constraint.

$\rho_{ij} \geq 0$ represents reciprocity sensitivity of actor $i$ toward actor $j$, derived below.

$\phi_{\text{recip}}(\cdot)$ represents bounded response function for reciprocity.

\textbf{Moving Average Interpretation.} Rather than comparing Actor $j$'s current action to fixed baseline, Actor $i$ compares to $j$'s recent average behavior over last $k$ periods. If $a_j^t > \bar{a}_j$, actor $j$ is cooperating relative to their norm. If $a_j^t < \bar{a}_j$, actor $j$ is defecting. This formulation has desirable properties including adapting to changing conditions where if everyone's actions shift up then baselines shift up, self-referential assessment where actors are judged against own history rather than external standards, and symmetric updating where if $j$ increases actions then future baseline rises preventing indefinite escalation from positive feedback.

\subsection{Bounded Response Function}

The bounded response function is defined as:

\begin{equation}
\label{eq:bounded_response}
\phi_{\text{recip}}(x) = \tanh(\kappa_{\text{recip}} x) = \frac{e^{\kappa_{\text{recip}} x} - e^{-\kappa_{\text{recip}} x}}{e^{\kappa_{\text{recip}} x} + e^{-\kappa_{\text{recip}} x}}
\end{equation}

where $\kappa_{\text{recip}} > 0$ is response sensitivity parameter.

\textbf{Key Properties:}

The function has bounded range where $\phi_{\text{recip}}: \mathbb{R} \to (-1, 1)$ maps real line to bounded interval. It is monotonic, meaning strictly increasing such that more cooperation yields more positive response. It exhibits symmetry where $\phi_{\text{recip}}(-x) = -\phi_{\text{recip}}(x)$ treating cooperation and defection symmetrically in magnitude. Near origin it has linear region where $\phi_{\text{recip}}(x) \approx \kappa_{\text{recip}} x$ for small absolute value of $x$, providing proportional response to moderate deviations. It shows saturation where it approaches plus or minus one as $x$ approaches plus or minus infinity, ensuring extreme deviations produce bounded finite responses. Sensitivity is controlled by parameter $\kappa_{\text{recip}}$ where high values mean strong responses to small deviations and low values mean weak responses requiring large deviations.

\subsection{Structural Foundation of Reciprocity Sensitivity}

Reciprocity sensitivity emerges from structural dependencies rather than being exogenous. We present the asymmetric formulation as primary specification.

\textbf{Asymmetric Formulation (Own Dependency):}

\begin{equation}
\label{eq:rho_asymmetric}
\rho_{ij} = \rho_0 \cdot D_{ij}^\eta
\end{equation}

\textbf{Parameter Interpretation:}

$\rho_0 > 0$ represents base reciprocity tendency capturing intrinsic reciprocity independent of structural position. Typically $\rho_0 \in [0.5, 2.0]$, where $\rho_0 = 1.0$ is neutral, $\rho_0 < 1.0$ represents weak reciprocity, and $\rho_0 > 1.0$ represents strong reciprocity.

$\eta \geq 1$ represents elasticity parameter controlling how strongly dependencies amplify reciprocity. When $\eta = 1$ there is linear amplification. When $\eta > 1$ there is super-linear amplification. Typically $\eta \in [1.0, 2.0]$, where $\eta = 1.5$ represents moderate amplification.

$D_{ij}$ represents structural dependency from interdependence matrix derived from \textit{i*} network analysis.

\textbf{Rationale.} Actor $i$'s reciprocity toward Actor $j$ depends on how much $i$ depends on $j$. High dependency $D_{ij}$ means $i$ is vulnerable to $j$'s behavior, amplifying reciprocity responses. When $i$ depends critically on $j$, violations by $j$ elicit strong retaliation because $j$'s actions directly affect $i$'s outcomes. When $i$ does not depend on $j$, $i$ responds weakly because $j$'s behavior is less consequential.

\subsection{Modular Utility Function Structure}

Following the modular presentation established in companion technical reports~\cite{pant2025foundations,pant2025trust,pant2025teams}, we present the complete utility function as a composition of distinct components. This modular structure facilitates understanding how each mechanism contributes to strategic behavior and enables selective activation of components for comparative analysis.

\subsubsection{Base Payoff (from TR-1)}

The base payoff captures individual value creation, cost of effort, and share of synergistic value:

\begin{equation}
\label{eq:base_payoff}
\pi_i^{\text{base}}(\vect{a}^t) = (e_i - a_i^t) + f_i(a_i^t) + \alpha_i \left[ V(\vect{a}^t) - \sum_{j=1}^N f_j(a_j^t) \right]
\end{equation}

where $e_i$ is initial endowment, $a_i^t$ is action (investment) at time $t$, $f_i(a_i^t) = \theta \ln(1 + a_i^t)$ is individual value creation with $\theta = 20.0$, $V(\vect{a}^t)$ is total synergistic value, and $\alpha_i$ is actor $i$'s share of synergy (from bargaining power).

\subsubsection{Interdependence Modifier (from TR-1)}

The interdependence modifier captures instrumental outcome coupling derived from \textit{i*} dependencies:

\begin{equation}
\label{eq:interdep_modifier}
U_i^{\text{interdep}}(\vect{a}^t) = \sum_{j \neq i} D_{ij} \cdot \pi_j^{\text{base}}(\vect{a}^t)
\end{equation}

where $D_{ij} \in [0,1]$ is the structural dependency coefficient from the interdependence matrix. Actor $i$ rationally internalizes partner payoffs proportional to structural dependencies, creating strategic complementarity independent of preferences.

\subsubsection{Trust Modifier (from TR-2)}

The trust modifier captures trust-weighted partner outcome consideration:

\begin{equation}
\label{eq:trust_modifier}
U_i^{\text{trust}}(\vect{a}^t) = \lambda_T \sum_{j \neq i} T_{ij}^t \cdot D_{ij} \cdot \pi_j^{\text{base}}(\vect{a}^t)
\end{equation}

where $T_{ij}^t \in [0,1]$ is immediate trust from the trust dynamics model and $\lambda_T > 0$ is the trust weight parameter. Trust modulates interdependence, amplifying partner considerations when trust is high and dampening them when trust is low.

\subsubsection{Reciprocity Modifier (TR-4 Contribution)}

The reciprocity modifier, the primary contribution of this technical report, captures conditional cooperation based on observed partner behavior:

\begin{equation}
\label{eq:recip_modifier}
U_i^{\text{recip}}(\vect{a}^t, h^{t-1}) = \lambda_R \sum_{j \neq i} T_{ij}^t (1 + \omega D_{ij}) \rho_{ij} \phi_{\text{recip}}(s_{ij}^t)
\end{equation}

where:
\begin{itemize}
\item $\lambda_R > 0$ is the reciprocity weight controlling overall influence (typically 0.5 to 2.0)
\item $T_{ij}^t \in [0,1]$ gates reciprocity by trust level
\item $(1 + \omega D_{ij})$ amplifies reciprocity based on dependency ($\omega \geq 0$, typically 0.5 to 1.5)
\item $\rho_{ij} = \rho_0 D_{ij}^\eta$ is reciprocity sensitivity from Definition~\ref{def:recip_sensitivity}
\item $\phi_{\text{recip}}(s_{ij}^t) = \tanh(\kappa \cdot s_{ij}^t)$ is the bounded response to cooperation signal
\item $s_{ij}^t = a_j^t - \bar{a}_j^{t-k:t-1}$ is the cooperation signal from Definition~\ref{def:coop_signal}
\end{itemize}

\subsubsection{Complete Utility Function}

The complete utility function integrates all modular components:

\begin{equation}
\label{eq:utility_complete}
\boxed{U_i(\vect{a}^t, h^{t-1}) = \pi_i^{\text{base}}(\vect{a}^t) + U_i^{\text{interdep}}(\vect{a}^t) + U_i^{\text{trust}}(\vect{a}^t) + U_i^{\text{recip}}(\vect{a}^t, h^{t-1})}
\end{equation}

This modular structure enables systematic analysis: setting $\lambda_R = 0$ recovers the trust-augmented model from TR-2; setting both $\lambda_R = 0$ and $\lambda_T = 0$ recovers the base interdependence model from TR-1. Table~\ref{tab:utility_components} summarizes the components and their origins.

\begin{table}[htbp]
\centering
\caption{Utility function components and their roles}
\label{tab:utility_components}
\begin{tabular}{llll}
\toprule
\textbf{Component} & \textbf{Source} & \textbf{Role} & \textbf{Key Parameters} \\
\midrule
$\pi_i^{\text{base}}$ & TR-1 & Individual value creation & $\theta$, $\alpha_i$ \\
$U_i^{\text{interdep}}$ & TR-1 & Structural outcome coupling & $D_{ij}$ \\
$U_i^{\text{trust}}$ & TR-2 & Trust-weighted consideration & $T_{ij}^t$, $\lambda_T$ \\
$U_i^{\text{recip}}$ & TR-4 & Conditional cooperation & $\rho_{ij}$, $k$, $\kappa$, $\lambda_R$ \\
\bottomrule
\end{tabular}
\end{table}

\subsubsection{Behavioral Implications}

The modular structure reveals distinct behavioral channels:

\textbf{When partner cooperates ($s_{ij}^t > 0$):} The reciprocity term $\phi_{\text{recip}}(s_{ij}^t) > 0$ increases $i$'s utility, reinforcing cooperation. If trust is high ($T_{ij}^t$ near 1), reinforcement is strong. If trust is low ($T_{ij}^t$ near 0), reinforcement is weak despite positive signals.

\textbf{When partner defects ($s_{ij}^t < 0$):} The reciprocity term $\phi_{\text{recip}}(s_{ij}^t) < 0$ decreases $i$'s utility, punishing defection. High-dependency relationships ($D_{ij}$ high) experience stronger punishment through the $(1 + \omega D_{ij})$ amplification.

\textbf{Trust-reciprocity interaction:} Low trust dampens reciprocity responses even when conditional cooperation would be strategically rational. This captures realistic behavior where actors in low-trust relationships remain guarded despite observing cooperative signals, fearing exploitation.

\subsection{Dynamic Properties and Strategic Complementarity}

Reciprocity creates strategic complementarity through positive feedback mechanisms.

\textbf{Cooperation Spiral.} When actor $j$ cooperates more by increasing $a_j^t$ above average $\bar{a}_j$, actor $i$'s reciprocity response $R_{ij}$ increases becoming more positive, which increases $i$'s utility from cooperation, inducing $i$ to cooperate more by increasing $a_i^{t+1}$, which increases $j$'s reciprocity response $R_{ji}$, creating upward spiral sustaining cooperation.

\textbf{Defection Spiral.} When actor $j$ defects by decreasing $a_j^t$ below average, actor $i$'s reciprocity response $R_{ij}$ decreases becoming negative, reducing $i$'s utility from cooperation, inducing $i$ to defect by decreasing $a_i^{t+1}$, triggering $j$'s reciprocity response $R_{ji}$ to decrease, creating downward spiral.

The existence of both stable cooperative equilibria and defection equilibria reflects multiple equilibria typical in repeated games with reciprocity, consistent with folk theorem results showing that many outcomes can be sustained as equilibria when interactions are repeated.

\subsection{Computational Algorithm}

We present a formal algorithm for computing reciprocity-augmented equilibria. Algorithm~\ref{alg:recip_solver} implements iterative best-response dynamics incorporating memory-windowed history tracking and trust-gated reciprocity.

\begin{algorithm}[htbp]
\caption{Reciprocity-Augmented Equilibrium Solver}
\label{alg:recip_solver}
\begin{algorithmic}[1]
\Require Action spaces $A_i$ for all $i \in \mathcal{N}$, history $h^{t-1}$, trust states $T_{ij}^t$, interdependence matrix $D$
\Require Parameters: $\rho_0$ (base reciprocity), $\eta$ (elasticity), $k$ (memory window), $\kappa$ (response sensitivity), $\lambda_R$ (reciprocity weight), $\omega$ (dependency amplification)
\Ensure Equilibrium actions $\vect{a}^{*,t}$
\State Initialize $\vect{a}^{(0)} \gets \vect{a}^{t-1}$ \Comment{Start from previous period actions}
\For{$\ell = 1, 2, \ldots, L_{\max}$} \Comment{Iterate until convergence}
    \For{each actor $i \in \mathcal{N}$}
        \For{each partner $j \neq i$}
            \State $\bar{a}_j \gets \frac{1}{\min(k,t-1)} \sum_{\tau=\max(1,t-k)}^{t-1} a_j^\tau$ \Comment{Moving average}
            \State $s_{ij} \gets a_j^{(\ell-1)} - \bar{a}_j$ \Comment{Cooperation signal}
            \State $\rho_{ij} \gets \rho_0 \cdot D_{ij}^\eta$ \Comment{Reciprocity sensitivity}
            \State $R_{ij} \gets T_{ij}^t (1 + \omega D_{ij}) \rho_{ij} \cdot \tanh(\kappa \cdot s_{ij})$ \Comment{Reciprocity response}
        \EndFor
        \State $a_i^{(\ell)} \gets \argmax_{a_i \in A_i} U_i(a_i, \vect{a}_{-i}^{(\ell-1)}, h^{t-1})$ \Comment{Best response}
    \EndFor
    \If{$\|\vect{a}^{(\ell)} - \vect{a}^{(\ell-1)}\|_\infty < \epsilon$} \Comment{Convergence check}
        \State \Return $\vect{a}^{*,t} \gets \vect{a}^{(\ell)}$
    \EndIf
\EndFor
\State \Return $\vect{a}^{*,t} \gets \vect{a}^{(L_{\max})}$ \Comment{Return best approximation}
\end{algorithmic}
\end{algorithm}

\textbf{Computational Complexity.} Each iteration requires $O(N^2)$ operations for computing pairwise reciprocity responses and $O(N \cdot |A|)$ operations for best-response computation, where $|A|$ is action space cardinality. Total complexity is $O(L_{\max} \cdot N \cdot (N + |A|))$. For continuous action spaces, gradient-based optimization replaces argmax with complexity dependent on optimization method.

\textbf{Convergence Properties.} Under standard regularity conditions (Lipschitz continuity of utility functions, bounded action spaces), the algorithm converges to an $\epsilon$-equilibrium. The bounded response function $\phi_{\text{recip}}$ ensures Lipschitz continuity of the reciprocity component.

\subsection{Analytical Properties}

We establish formal propositions characterizing key properties of the reciprocity framework. These results provide theoretical foundations for the empirical validation in subsequent sections.

\begin{proposition}[Cooperation Emergence Condition]
\label{prop:coop_emergence}
In a repeated two-player game with reciprocity parameters $(\rho_0, k, \kappa)$ and initial trust $T_{ij}^0 \geq T^*$, a cooperative equilibrium exists if the base reciprocity exceeds a critical threshold:
\begin{equation}
\label{eq:coop_condition}
\rho_0 > \rho^* = \frac{c'}{\lambda_R T^* (1 + \omega D_{ij}) \kappa}
\end{equation}
where $c' = \partial c_i / \partial a_i$ is the marginal cost of cooperation at the interior solution.
\end{proposition}

\begin{proof}[Proof Sketch]
At an interior cooperative equilibrium, the first-order condition requires $\partial U_i / \partial a_i = 0$. The marginal benefit from reciprocity is $\lambda_R T_{ij}^t (1 + \omega D_{ij}) \rho_{ij} \cdot \phi_{\text{recip}}'(s_{ij}) \cdot \partial s_{ij}/\partial a_i$. Near the baseline ($s_{ij} \approx 0$), $\phi_{\text{recip}}'(0) = \kappa$. For cooperation to be incentive-compatible, marginal benefit must exceed marginal cost $c'$, yielding the threshold condition. Full proof in Appendix.
\end{proof}

\begin{proposition}[Memory Window Effect on Forgiveness]
\label{prop:memory_effect}
Let $\tau_f$ denote the forgiveness time (periods until cooperation returns to baseline following an isolated defection). Then $\tau_f$ is bounded by:
\begin{equation}
\label{eq:forgiveness_bound}
k \leq \tau_f \leq 2k
\end{equation}
with $\tau_f \to k$ as $\kappa \to \infty$ (high sensitivity) and $\tau_f \to 2k$ as $\kappa \to 0$ (low sensitivity).
\end{proposition}

\begin{proof}[Proof Sketch]
An isolated defection at period $t^*$ affects the moving average $\bar{a}_j$ for periods $t^* + 1$ through $t^* + k$. The defection's weight in the average is $1/k$ immediately after, declining to zero as it exits the window. The cooperation signal $s_{ij} = a_j^t - \bar{a}_j$ returns to baseline when the defection exits the memory window (lower bound $k$). However, reduced cooperation responses during the memory window create secondary effects extending recovery (upper bound $2k$). The exact value depends on $\kappa$: high sensitivity ($\kappa \to \infty$) creates sharp, bounded responses that terminate immediately when the defection exits the memory window, driving $\tau_f \to k$, while low sensitivity ($\kappa \to 0$) produces weak responses that allow gradual drift, extending $\tau_f$ toward $2k$. Full proof in Appendix.
\end{proof}

\begin{proposition}[Trust-Reciprocity Complementarity]
\label{prop:trust_recip_complement}
Trust and reciprocity exhibit strategic complementarity in determining equilibrium cooperation levels:
\begin{equation}
\label{eq:complementarity}
\frac{\partial^2 a_i^*}{\partial T_{ij} \partial \rho_{ij}} > 0
\end{equation}
That is, the marginal effect of reciprocity on cooperation is increasing in trust, and vice versa.
\end{proposition}

\begin{proof}[Proof Sketch]
From the utility function (Equation~\ref{eq:utility_complete}), the reciprocity term contains the product $T_{ij}^t \cdot \rho_{ij}$. Applying the implicit function theorem to the first-order condition and differentiating with respect to both $T_{ij}$ and $\rho_{ij}$ yields a positive cross-partial derivative under concavity of the utility function. Intuitively, high trust amplifies the impact of reciprocity (partners are willing to act on cooperative signals), and high reciprocity amplifies the impact of trust (cooperative history translates to sustained cooperation). Full proof in Appendix.
\end{proof}

Table~\ref{tab:propositions_summary} summarizes the analytical results and their implications for model behavior.

\begin{table}[htbp]
\centering
\caption{Summary of analytical propositions}
\label{tab:propositions_summary}
\begin{tabular}{lll}
\toprule
\textbf{Proposition} & \textbf{Key Result} & \textbf{Implication} \\
\midrule
Cooperation Emergence & $\rho_0 > \rho^*$ & Threshold reciprocity for cooperation \\
Memory Window Effect & $k \leq \tau_f \leq 2k$ & Forgiveness bounded by memory \\
Trust-Reciprocity Complement & $\partial^2 a^* / \partial T \partial \rho > 0$ & Mechanisms reinforce each other \\
\bottomrule
\end{tabular}
\end{table}

\section{Translation Framework: From \textit{i*} to Computational Reciprocity Models}
\label{sec:translation}

We present systematic operationalization for conceptual modeling practitioners and researchers, providing concrete guidance for translating \textit{i*} models into computational reciprocity parameters. This eight-step methodology enables practitioners to instantiate the formal framework from conceptual models elicited through standard requirements engineering practices.

\subsection{Framework Overview}

Table~\ref{tab:translation_steps} summarizes the eight-step translation framework, identifying key actions and outputs at each stage. The methodology proceeds from structural analysis of \textit{i*} models through parameterization to simulation and calibration.

\begin{table}[htbp]
\centering
\caption{Eight-step translation framework from \textit{i*} to reciprocity parameters}
\label{tab:translation_steps}
\begin{tabular}{clll}
\toprule
\textbf{Step} & \textbf{Action} & \textbf{Output} & \textbf{Section} \\
\midrule
1 & Identify sequential dependencies in \textit{i*} & List of temporal dependencies & \ref{sec:step1} \\
2 & Establish temporal granularity & Period length $\Delta t$ & \ref{sec:step2} \\
3 & Define cooperation baselines & $a_j^{\text{baseline}}$ or moving average & \ref{sec:step3} \\
4 & Determine memory window length & $k$ periods & \ref{sec:step4} \\
5 & Parameterize reciprocity sensitivity & $\rho_{ij} = \rho_0 D_{ij}^\eta$ & \ref{sec:step5} \\
6 & Configure response sensitivity & $\kappa$ parameter & \ref{sec:step6} \\
7 & Integrate trust dynamics & $\lambda_T$, $\omega$, $T_{ij}^0$ & \ref{sec:step7} \\
8 & Simulate and calibrate & Validated parameters & \ref{sec:step8} \\
\bottomrule
\end{tabular}
\end{table}

Table~\ref{tab:istar_translation} provides direct mapping between \textit{i*} model elements and reciprocity parameters, identifying appropriate elicitation methods for each parameter.

\begin{table}[htbp]
\centering
\caption{Translation from \textit{i*} elements to reciprocity parameters}
\label{tab:istar_translation}
\begin{tabular}{lll}
\toprule
\textbf{\textit{i*} Element} & \textbf{Reciprocity Parameter} & \textbf{Elicitation Method} \\
\midrule
Dependency criticality & $D_{ij}$ & Weighted aggregation (TR-1) \\
Dependency directionality & $\rho_{ij}$ asymmetry & Structural analysis \\
Interaction pattern & $k$ (memory window) & Stakeholder interview \\
Response magnitude & $\kappa$ (sensitivity) & Historical analysis / scenario testing \\
Baseline behavior & $a^{\text{baseline}}$ & Normative expectations / SLAs \\
Trust relationship & $T_{ij}^0$ & From TR-2 trust model \\
Reciprocity tendency & $\rho_0$ (base) & Behavioral surveys / calibration \\
Dependency amplification & $\eta$ (elasticity) & Model fitting \\
\bottomrule
\end{tabular}
\end{table}

\subsection{Step 1: Identify Sequential Dependency Relationships}
\label{sec:step1}

In standard \textit{i*}, dependencies are static where actor $i$ depends on actor $j$ for dependum $d$. For reciprocity analysis, practitioners must identify which dependencies involve sequential interactions where behavior at time $t$ affects future behavior at $t+1$.

\textbf{Sequential Dependency Examples.} In requirements engineering, stakeholder information disclosure dependencies are sequential where analyst responsiveness at sprint $t$ affects stakeholder disclosure at sprint $t+1$. When analysts incorporate feedback quickly, stakeholders provide more detailed requirements in subsequent iterations. For development teams, contribution dependencies are sequential where teammate effort at sprint $t$ affects own effort at sprint $t+1$. When teammates contribute actively, others maintain high contribution, but when teammates shirk, others reduce effort. In platform ecosystems, API stability dependencies are sequential where provider reliability at quarter $t$ affects developer investment at quarter $t+1$. When providers maintain stable APIs, developers invest in platform, but when providers make breaking changes, developers reduce commitment.

\textbf{Non-Sequential Dependencies.} Not all dependencies are sequential. Some are one-shot involving single transaction or continuous involving ongoing simultaneous interactions. Practitioners must distinguish sequential from non-sequential dependencies.

\subsection{Step 2: Establish Temporal Granularity and Time Units}
\label{sec:step2}

Reciprocity operates over discrete time periods, so practitioners must define what constitutes period $t$.

\textbf{Context-Dependent Granularity.} For quarterly business relationships, $t$ represents quarters meaning three-month periods. For agile sprints, $t$ represents two-week iterations. For open-source communities, $t$ might represent monthly contribution cycles. For supply chain partnerships, $t$ might represent delivery cycles or contract periods.

\textbf{Implication.} Temporal granularity affects memory window interpretation. Memory window $k=5$ means different things at quarterly versus weekly granularity. Practitioners should choose granularity matching natural interaction cycles in their domain.

\subsection{Step 3: Define Cooperation Baselines or Use Moving Averages}
\label{sec:step3}

For each sequential dependency, establish what constitutes cooperation versus defection.

\textbf{Approach A - Fixed Baseline.} Cooperation means exceeding pre-specified threshold based on domain knowledge, contracts, or service level agreements. Examples include API uptime above ninety-five percent, sprint velocity above fifty story points, or information disclosure exceeding contractually specified levels.

\textbf{Approach B - Moving Average (Implemented in Our Model).} Cooperation means exceeding recent historical average. Current action $a_j^t$ compared to average of last $k$ periods $\bar{a}_j^{t-k:t-1}$. Examples include current sprint velocity exceeds average of last five sprints, or current information disclosure exceeds average of recent interactions.

\textbf{Trade-offs.} Fixed baselines require domain knowledge and may become outdated as conditions change, but have clear interpretation where cooperation is objective standard. Moving averages adapt to changing conditions automatically but can drift where if everyone's actions decline then baselines decline obscuring the decline, with self-referential interpretation where cooperation is relative to recent norm.

\textbf{Recommendation.} Use moving averages, which is our approach, in dynamic environments where conditions evolve. Use fixed baselines in regulated contexts with explicit service standards.

\subsection{Step 4: Determine Memory Window Length $k$}
\label{sec:step4}

Practitioners must determine how many past periods actors consider when reciprocating.

\textbf{Elicitation Technique 1 - Stakeholder Surveys.} Ask stakeholders when deciding whether to cooperate, how many past interactions do you consider? Options include just last interaction for $k=1$, last few interactions for $k=3$ to $5$, or last several interactions for $k=5$ to $10$.

\textbf{Elicitation Technique 2 - Empirical Analysis.} Compute correlation between actions at different lag periods. Memory window is where correlation drops below significance threshold. For example, if correlation between $a_i^t$ and $a_j^{t-\tau}$ remains significant for $\tau \leq 5$ but drops for $\tau > 5$, then set $k=5$.

\textbf{Elicitation Technique 3 - Cognitive Load Considerations.} Longer memory windows increase computational and cognitive burden. Use shortest window that captures relevant history. As heuristic, for quarterly interactions use $k=4$ representing one year, for monthly interactions use $k=6$ representing half year, and for weekly interactions use $k=12$ representing quarter.

\subsection{Step 5: Parameterize Reciprocity Sensitivity through Structural Dependencies}
\label{sec:step5}

Use interdependence matrix $D$ from \textit{i*} dependency network analysis already populated from foundational work.

\textbf{Choose Formulation.} Asymmetric formulation $\rho_{ij} = \rho_0 D_{ij}^\eta$ is our primary specification. Symmetric formulation $\rho_{ij} = \rho_0 \sqrt{D_{ij} D_{ji}}^{\,\eta}$ is alternative for mutual dependencies.

\textbf{Estimate Base Reciprocity $\rho_0$.} Method A involves experimental games where you run simplified Prisoner's Dilemma with representative actors, observe reciprocity strength in responses, and calibrate $\rho_0$ to match. Method B uses surveys with validated reciprocity scales from psychology literature, mapping responses to $\rho_0$ parameter values. Method C uses calibration where you choose $\rho_0$ such that observed cooperation levels in historical data match model predictions, iterating until fit achieved. Default, if no empirical data available, is to use $\rho_0 = 1.0$ representing neutral reciprocity as starting point.

\textbf{Estimate Elasticity $\eta$.} Method A analyzes whether actors with high dependencies show proportionally stronger when $\eta=1$ or disproportionately stronger when $\eta>1$ reciprocity responses. Method B tests $\eta \in \{1.0, 1.5, 2.0\}$, selecting value producing best model fit to observed behavior. Default is to use $\eta = 1.0$ representing linear as baseline, adjusting if evidence suggests super-linear amplification.

\subsection{Step 6: Parameterize Bounded Response Function Sensitivity $\kappa$}
\label{sec:step6}

Higher $\kappa$ means actors respond strongly to small deviations. Lower $\kappa$ means actors require large deviations to trigger responses.

\textbf{Elicitation through Scenario Analysis.} Present stakeholders with hypothetical deviations of various magnitudes asking if partner increased contribution by ten percent, how would you respond? Fit $\kappa$ to match reported responses.

\textbf{Default Value.} Use $\kappa = 1.0$ as moderate sensitivity baseline, adjusting if model predictions deviate from observations. Table~\ref{tab:kappa_interpretation} provides interpretation guidelines.

\begin{table}[htbp]
\centering
\caption{Interpretation of response sensitivity parameter $\kappa$}
\label{tab:kappa_interpretation}
\begin{tabular}{lll}
\toprule
\textbf{$\kappa$ Range} & \textbf{Interpretation} & \textbf{Typical Context} \\
\midrule
$\kappa < 0.5$ & Low sensitivity & Tolerant relationships, high noise \\
$0.5 \leq \kappa < 1.0$ & Moderate-low & Long-term partnerships \\
$1.0 \leq \kappa < 1.5$ & Moderate (default) & Standard business relationships \\
$1.5 \leq \kappa < 2.0$ & Moderate-high & Competitive contexts \\
$\kappa \geq 2.0$ & High sensitivity & High-stakes, low-tolerance contexts \\
\bottomrule
\end{tabular}
\end{table}

\subsection{Step 7: Integrate with Trust Dynamics}
\label{sec:step7}

If modeling trust from second paper, recognize that reciprocity is gated by trust through $T_{ij}^t$ multiplying reciprocity terms.

\textbf{Implication.} Reciprocity responses depend on both current actions through $R_{ij}$ and trust state through $T_{ij}$. Practitioners must model trust evolution alongside reciprocity.

\textbf{Double Penalty and Reinforcement.} When violations occur, trust erodes and reciprocity responses become negative, creating double penalty. When cooperation occurs, trust builds and reciprocity responses become positive, creating double reinforcement.

\textbf{Trust Integration Parameters.} Practitioners must specify: (i) reciprocity weight $\lambda_R$ controlling overall influence of reciprocity on utility; (ii) dependency amplification $\omega$ controlling how dependencies further amplify trust-gated reciprocity; and (iii) initial trust $T_{ij}^0$ from TR-2 trust model or domain knowledge.

\subsection{Step 8: Simulate and Calibrate}
\label{sec:step8}

Implement iterative best-response dynamics or gradient-based optimization to compute Perfect Bayesian Equilibria using Algorithm~\ref{alg:recip_solver}.

\textbf{Simulation Procedure.} Initialize with initial action profile $\vect{a}^0$ perhaps from static equilibrium ignoring reciprocity. Iterate where for each time $t$, each actor $i$ chooses action optimally given others' actions and history $h^{t-1}$, update history to include new actions, compute trust evolution based on cooperation signals if including trust, and check equilibrium convergence where actions stabilize. Analyze resulting trajectories showing cooperation levels, reciprocity responses, and trust evolution over time.

\textbf{Calibration Protocol.} After initial simulation, compare model predictions with observed behavior: (i) if cooperation levels are too high, reduce $\rho_0$ or $\lambda_R$; (ii) if forgiveness is too slow, reduce $k$; (iii) if responses are too sharp, reduce $\kappa$; (iv) if high-dependency actors do not react strongly enough, increase $\eta$. Iterate until model behavior matches domain observations within acceptable tolerance.

Table~\ref{tab:calibration_guidance} provides diagnostic guidance for parameter adjustment.

\begin{table}[htbp]
\centering
\caption{Calibration guidance for parameter adjustment}
\label{tab:calibration_guidance}
\begin{tabular}{lll}
\toprule
\textbf{Symptom} & \textbf{Likely Cause} & \textbf{Adjustment} \\
\midrule
Cooperation too high & Reciprocity too strong & Decrease $\rho_0$ or $\lambda_R$ \\
Cooperation too low & Reciprocity too weak & Increase $\rho_0$ or $\lambda_R$ \\
Forgiveness too slow & Memory too long & Decrease $k$ \\
Forgiveness too fast & Memory too short & Increase $k$ \\
Responses too sharp & Sensitivity too high & Decrease $\kappa$ \\
Responses too gradual & Sensitivity too low & Increase $\kappa$ \\
Dependency differentiation weak & Elasticity too low & Increase $\eta$ \\
Dependency differentiation extreme & Elasticity too high & Decrease $\eta$ \\
\bottomrule
\end{tabular}
\end{table}

\subsection{Complete Worked Example: Software Ecosystem}

Consider software ecosystem with platform provider and three app developers interacting over twelve quarters.

\textbf{Specification.} Actors include Provider designated P, Core Developer designated C, and Peripheral Developers designated D1 and D2. Actions include Provider API stability $a_P \in [0,20]$ and Developer app quality $a_C, a_{D1}, a_{D2} \in [0,15]$. Dependencies include $D_{CP} = 0.8$ where core developer depends critically on platform, and $D_{D1,P} = D_{D2,P} = 0.3$ where peripheral developers are weakly dependent. Memory window is $k=4$ quarters representing one year. Reciprocity parameters include $\rho_0 = 1.2$ representing moderately reciprocal, $\eta = 1.2$ representing super-linear amplification, and $\kappa_{\text{recip}} = 1.0$ representing moderate sensitivity.

\textbf{Scenario Timeline.} During Quarters 1 through 4, provider maintains high API stability at $a_P = 18$ consistently above expectations. Developers reciprocate with high quality where core developer achieves $a_C = 14$ and peripherals achieve $a_{D1} = a_{D2} = 10$. Moving averages establish cooperation baseline.

In Quarter 5, provider makes breaking API change with $a_P = 8$ falling sharply below four-quarter average of 18. This is major defection.

For reciprocity calculations, core developer computes $a_P^5 - \bar{a}_P^{1:4} = 8 - 18 = -10$, passes through $\phi_{\text{recip}}(-10) = \tanh(-10) \approx -1.0$ showing saturated negative response, multiplied by high reciprocity sensitivity $\rho_{CP} = 1.2 \times 0.8^{1.2} = 0.90$. Core developer responds with strong negative reciprocity, reducing quality to $a_C = 6$ showing sharp decrease.

Peripheral developers compute same deviation but have lower reciprocity sensitivity $\rho_{D1,P} = 1.2 \times 0.3^{1.2} = 0.29$. They respond with weak negative reciprocity, reducing quality to $a_{D1} = a_{D2} = 8$ showing moderate decrease.

During Quarters 6 through 8, provider returns to high stability at $a_P = 17$ attempting to rebuild cooperation. However, moving averages now include the quarter-5 violation where $\bar{a}_P^{2:5} = 16.25$ and then $\bar{a}_P^{3:6} = 16.0$. Provider's current actions only slightly exceed depressed averages.

Core developer's reciprocity response remains negative initially, only slowly returning toward neutral as positive cooperation signals accumulate. Trust, if modeled, remains depressed further dampening reciprocity. Core developer's quality recovers to $a_C = 10$ by quarter 8 showing partial recovery, never fully returning to original $a_C = 14$ due to lasting trust damage.

During Quarters 9 through 12, continued cooperation by provider gradually rebuilds reciprocity responses, but system exhibits hysteresis where cooperation level after violation-and-recovery is lower than before violation, demonstrating path dependence.

\textbf{Key Insights from Example.} Asymmetric dependencies produce differentiated responses where core developer reacts strongly while peripherals react weakly. Memory windows create gradual forgiveness as violations age out of averaging window. Trust integration, if included, would further amplify asymmetry through trust-gated reciprocity. Reciprocity provides enforcement mechanism where provider has incentive to maintain stability to avoid negative reciprocity spiral.

\section{\textit{i*} Modeling of Reciprocity Dynamics}
\label{sec:istar}

This section presents \textit{i*} Strategic Rationale and Strategic Dependency diagrams that capture reciprocity dynamics in sequential coopetitive interactions. These diagrams provide the structural foundation from which computational reciprocity parameters are derived, following the modeling approach established in companion technical reports for interdependence~\cite{pant2025foundations}, trust~\cite{pant2025trust}, and team production~\cite{pant2025teams}.

\subsection{Strategic Rationale Model: Actor Perspective on Reciprocity}

Figure~\ref{fig:sr_reciprocity} presents the Strategic Rationale diagram for an actor engaged in sequential coopetitive interaction, showing how reciprocity mechanisms affect the internal goal structure and decision-making.

\begin{figure}[htbp]
\centering
\begin{tikzpicture}[
    scale=0.78, transform shape,
    actor/.style={circle, draw, thick, minimum size=1.3cm, font=\small, align=center},
    softgoal/.style={cloud, cloud puffs=20, cloud puff arc=100, aspect=2.5, draw, thick,
                     minimum width=2.8cm, minimum height=0.8cm, align=center, font=\small, inner sep=0pt, fill=white},
    task/.style={regular polygon, regular polygon sides=6, draw, thick,
                 minimum size=1.6cm, font=\scriptsize, align=center, inner sep=0pt, fill=white},
    resource/.style={rectangle, draw, thick, minimum height=0.6cm, minimum width=2.0cm, font=\scriptsize, align=center, fill=white},
    contribution/.style={-latex, thick},
    neededby/.style={-{Circle[open, length=2mm, width=2mm]}, thick},
]

\draw[thick, dashed] (0,-1.5) ellipse (7.5cm and 7.8cm);
\node[actor, fill=white] at (-5.5,4.5) {Actor $i$};

\node[softgoal] (maximize) at (0, 4.8) {Maximize\\Long-Term\\Utility};

\node[softgoal] (sustain) at (-4.0, 2.5) {Sustain\\Cooperation};
\node[softgoal] (exploit) at (4.0, 2.5) {Exploit\\Short-Term\\Gain};

\node[resource] (reciprocity) at (-0.50, 1.50) {Reciprocity $\rho_{ij}$};

\node[softgoal] (conditional) at (-3.5, -0.2) {Conditional\\Cooperation};
\node[softgoal] (immediate) at (3.5, -0.2) {Immediate\\Payoff};

\node[task] (cooperate) at (-4.8, -3.8) {Match\\Coop.};
\node[task] (reward) at (-1.6, -4.8) {Reward\\Good};
\node[task] (punish) at (1.6, -4.8) {Punish\\Defect};
\node[task] (defect) at (4.8, -3.8) {Free-\\Ride};

\node[resource] (history) at (-2.2, -7.2) {History $\bar{a}_j^{t-k:t-1}$};
\node[resource] (trust) at (2.2, -7.2) {Trust $T_{ij}^t$};

\draw[contribution] (sustain.60) --
    node[pos=0.33, sloped, above=2pt, font=\tiny] {help} (maximize.210);
\draw[contribution] (exploit.120) --
    node[pos=0.33, sloped, above=2pt, font=\tiny] {hurt} (maximize.330);

\draw[contribution] (reciprocity.west) --
    node[pos=0.33, sloped, above=2pt, font=\tiny] {help} (sustain.340);
\draw[contribution] (reciprocity.east) --
    node[pos=0.33, sloped, above=2pt, font=\tiny] {hurt} (exploit.200);

\draw[contribution] (conditional.100) --
    node[pos=0.33, sloped, above=2pt, font=\tiny] {help} (sustain.260);
\draw[contribution] (immediate.80) --
    node[pos=0.33, sloped, above=2pt, font=\tiny] {help} (exploit.280);

\draw[contribution] (cooperate.60) --
    node[pos=0.18, sloped, above=2pt, font=\tiny] {help} (conditional.220);
\draw[contribution] (reward.120) --
    node[pos=0.18, sloped, above=2pt, font=\tiny] {help} (conditional.260);
\draw[contribution] (punish.120) --
    node[pos=0.12, sloped, above=2pt, font=\tiny] {help} (conditional.310);
\draw[contribution] (defect.120) --
    node[pos=0.18, sloped, above=2pt, font=\tiny] {help} (immediate.320);

\draw[contribution] (cooperate.30) --
    node[pos=0.12, sloped, above=2pt, font=\tiny, fill=white, inner sep=1pt] {hurt} (immediate.210);
\draw[contribution] (reward.60) --
    node[pos=0.12, sloped, above=2pt, font=\tiny, fill=white, inner sep=1pt] {hurt} (immediate.250);
\draw[contribution] (punish.60) --
    node[pos=0.18, sloped, above=2pt, font=\tiny, fill=white, inner sep=1pt] {hurt} (immediate.280);
\draw[contribution] (defect.150) --
    node[pos=0.12, sloped, above=2pt, font=\tiny, fill=white, inner sep=1pt] {hurt} (conditional.340);

\draw[neededby] (history.100) -- (cooperate.280);
\draw[neededby] (trust.80) -- (punish.260);

\end{tikzpicture}
\caption{Strategic Rationale diagram for an actor in sequential reciprocal interaction, following \textit{i*} 2.0 notation~\cite{dalpiaz2016istar}. The fundamental tension emerges from competing softgoals: ``Sustain Cooperation'' helps long-term utility, while ``Exploit Short-Term Gain'' hurts it by triggering partner retaliation. The resource ``Reciprocity $\rho_{ij}$'' contributes to both softgoals (help to cooperation, hurt to exploitation), valid per \textit{i*} 2.0 Table 1. Three tasks help sustain cooperation: ``Match Cooperation'' through positive reciprocity, ``Reward Good'' through reinforcing norms, and ``Punish Defection'' through deterrence that discourages future defection. ``Free-Ride'' helps short-term exploitation. Cross-cutting hurt links formalize mutual antagonism: cooperation hurts exploitation (opportunity cost), while free-riding hurts cooperation (undermines reciprocal trust). Resources ``Partner History'' and ``Trust State'' connect to tasks via NeededBy links.}
\label{fig:sr_reciprocity}
\end{figure}

The diagram captures the fundamental tension in sequential coopetitive interactions. The top-level goal ``Maximize Long-Term Utility'' decomposes into two competing softgoals: ``Sustain Cooperation'' (which yields higher payoffs over extended horizons) and ``Exploit Short-Term Gain'' (which yields immediate payoffs but risks triggering negative reciprocity from partners). Under low reciprocity sensitivity ($\rho_{ij} \to 0$), the actor discounts partner behavioral history and defaults to myopic optimization, selecting tasks that maximize immediate payoff regardless of relational consequences. Under high reciprocity sensitivity ($\rho_{ij} \to 1$), the actor conditions behavior strongly on observed partner cooperation, selecting tasks that match partner cooperation levels and enforce reciprocal norms.

The resource ``Partner History'' $\bar{a}_j^{t-k:t-1}$ represents the memory-windowed moving average from the mathematical formalization (Section~\ref{sec:formalization}). This resource is required by the task ``Match Partner Cooperation,'' formalizing the information dependency: an actor cannot reciprocate without observing and remembering partner behavior. The resource ``Trust State'' $T_{ij}^t$ gates the punishment task, capturing the trust-reciprocity interaction: even when partner history warrants punishment, low trust may prevent the actor from engaging in costly punitive actions.

\subsection{Strategic Dependency Model: Sequential Interaction Structure}

Figure~\ref{fig:sd_reciprocity} presents the Strategic Dependency diagram showing dependencies between actors engaged in sequential reciprocal interactions, with temporal dependency annotations distinguishing static from sequential dependencies.

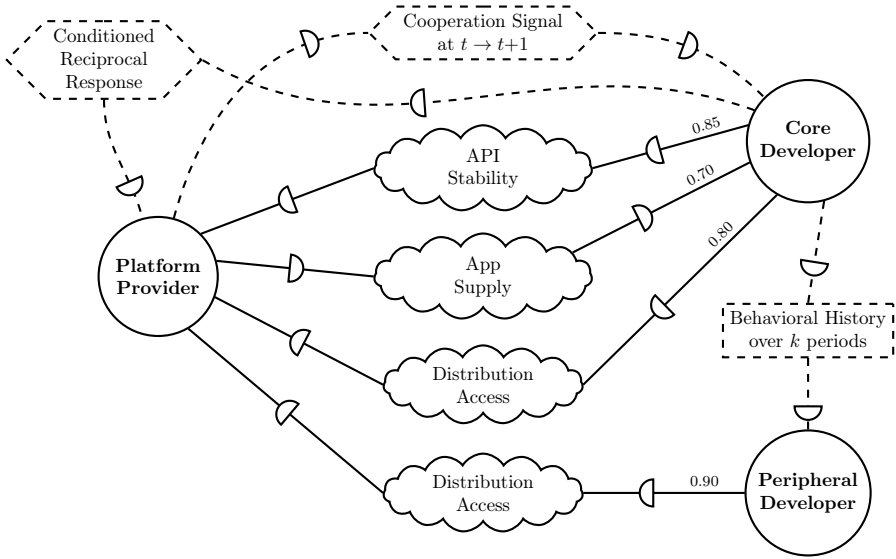
\begin{figure}[htbp]
\centering
\begin{tikzpicture}[
    scale=0.72, transform shape,
    actor/.style={circle, draw, thick, fill=white, align=center, minimum size=2.2cm, font=\small\bfseries},
    dependum/.style={ellipse, draw, thick, fill=white, align=center, minimum width=2.2cm, minimum height=0.9cm, font=\small},
    quality/.style={cloud, cloud puffs=15, cloud puff arc=120, aspect=3.5, draw, thick, fill=white, align=center, minimum width=3.0cm, minimum height=0.6cm, font=\small, inner sep=1pt},
    % FIX: Replaced 'regular polygon' with 'signal'. 
    % 'signal to=east and west' creates the pointed left and right edges.
    % This allows the node to stretch horizontally to fit text while remaining perfectly flat.
    seqtask/.style={signal, signal to=east and west, draw, thick, fill=white, align=center, minimum height=0.9cm, font=\small, inner sep=4pt, dashed},
    seqresource/.style={rectangle, draw, thick, fill=white, align=center, minimum height=0.6cm, minimum width=2.0cm, font=\small, dashed},
    % Plain lines (no arrowheads) per iStar 2.0 — direction shown by D-markers
    dep/.style={thick},
    % Semicircle D-marker; 'solid' enforced at node level to prevent dash inheritance
    dmark/.style={semicircle, draw=black, thick, fill=white, minimum size=2.5mm, inner sep=0pt},
    % Standardized label style
    connlabel/.style={sloped, font=\scriptsize, inner sep=1pt, above=2pt}
]

\node[actor] (platform) at (0, 0) {Platform\\Provider};
\node[actor] (developer) at (12, 2.5) {Core\\Developer};
\node[actor] (peripheral) at (12, -4.0) {Peripheral\\Developer};

\node[quality] (api) at (6, 2.0) {API\\Stability};
\node[quality] (apps) at (6, 0) {App\\Supply};
\node[quality] (distrib_a) at (6, -2.0) {Distribution\\Access};
\node[quality] (distrib_b) at (6, -4.0) {Distribution\\Access};

\node[seqtask] (coop_signal) at (6, 4.5) {Cooperation Signal\\at $t \to t{+}1$};
\node[seqresource] (history) at (12, -1.0) {Behavioral History\\ over $k$ periods};
\node[seqtask] (recip_resp) at (-1.0, 4.0) {Conditioned\\Reciprocal\\Response};

\draw[dep] (developer.165) -- node[pos=0.25, connlabel] {0.85}
    node[pos=0.6, dmark, sloped, allow upside down, rotate=-90, solid] {} (api.0);
\draw[dep] (api.180) --
    node[pos=0.5, dmark, sloped, allow upside down, rotate=-90, solid] {} (platform.45);

\draw[dep] (platform.15) --
    node[pos=0.5, dmark, sloped, allow upside down, rotate=-90, solid] {} (apps.180);
\draw[dep] (apps.15) -- node[pos=0.75, connlabel] {0.70}
    node[pos=0.4, dmark, sloped, allow upside down, rotate=-90, solid] {} (developer.200);

\draw[dep] (developer.240) -- node[pos=0.25, connlabel] {0.80}
    node[pos=0.6, dmark, sloped, allow upside down, rotate=-90, solid] {} (distrib_a.0);
\draw[dep] (distrib_a.180) --
    node[pos=0.5, dmark, sloped, allow upside down, rotate=-90, solid] {} (platform.340);

\draw[dep] (peripheral.180) -- node[pos=0.25, connlabel] {0.90}
    node[pos=0.6, dmark, sloped, allow upside down, rotate=-90, solid] {} (distrib_b.0);
\draw[dep] (distrib_b.180) --
    node[pos=0.5, dmark, sloped, allow upside down, rotate=-90, solid] {} (platform.300);

\draw[dep, dashed] (platform.75) to[out=75, in=180]
    node[pos=0.8, dmark, sloped, allow upside down, rotate=-90, solid] {} (coop_signal.west);
\draw[dep, dashed] (coop_signal.east) to[out=0, in=135]
    node[pos=0.5, dmark, sloped, allow upside down, rotate=-90, solid] {} (developer.135);

\draw[dep, dashed] (developer.285) --
    node[pos=0.65, dmark, sloped, allow upside down, rotate=-90, solid] {} (history.north);
\draw[dep, dashed] (history.south) --
    node[pos=0.75, dmark, sloped, allow upside down, rotate=-90, solid] {} (peripheral.90);

\draw[dep, dashed] (developer.150) to[out=160, in=330]
    node[pos=0.6, dmark, sloped, allow upside down, rotate=-90, solid] {} (recip_resp.east);
\draw[dep, dashed] (recip_resp.south) to[out=-90, in=105]
    node[pos=0.75, dmark, sloped, allow upside down, rotate=-90, solid] {} (platform.105);

\end{tikzpicture}
\caption{Strategic Dependency diagram using \textit{i*} 2.0 notation. Static structural dependencies (solid arrows, from TR-1) capture structural coupling: developers depend on the platform for distribution access (0.80, 0.90) and API stability (0.85), while the platform depends on developers for app supply (0.70). Sequential behavioral dependencies (dashed arrows, TR-4 contribution) capture temporal conditionality. The annotations ``$t \to t{+}1$,'' ``$k$ periods,'' and ``conditioned'' are extensions to standard \textit{i*} 2.0 dependency notation, added to convey the temporal semantics that distinguish sequential from static dependencies. Dashed borders on sequential dependums further distinguish TR-4's temporal extension from TR-1's static structure.}
\label{fig:sd_reciprocity}
\end{figure}

The diagram distinguishes two types of dependencies using \textit{i*} 2.0 notation with D-markers. Solid lines represent static structural dependencies from TR-1 interdependence analysis: developers depend on the platform for distribution access and API stability, the platform depends on developers for app supply. These dependencies determine the interdependence matrix $D_{ij}$ and create the structural foundation for reciprocity sensitivity through $\rho_{ij} = \rho_0 D_{ij}^\eta$. Dashed lines represent sequential behavioral dependencies, the contribution of TR-4: the platform's cooperation signal at time $t$ affects developer investment at time $t+1$, behavioral history over $k$ periods accumulates to inform reciprocity assessments, and reciprocal responses feed back to condition future behavior. Sequential dependums are further distinguished by dashed borders on the element shapes. This temporal layering on static dependencies is the core modeling extension that reciprocity analysis brings to \textit{i*} practice.

\subsection{Goal Model: Reciprocity Mechanisms and Actor Goals}

Figure~\ref{fig:goal_reciprocity} presents the Goal Model showing how reciprocity mechanisms connect to actor goals across the sequential interaction.

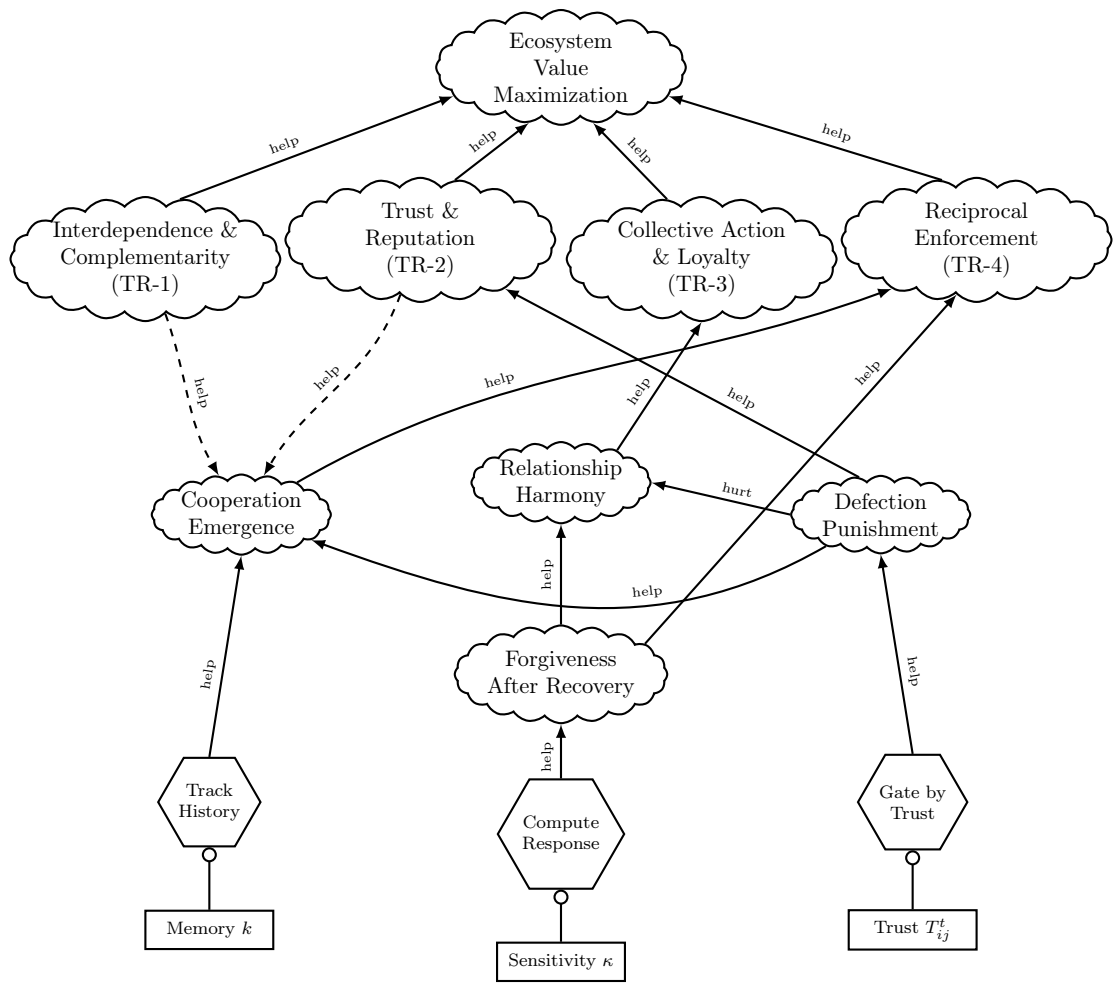
\begin{figure}[htbp]
\centering
\begin{tikzpicture}[
    scale=0.85, transform shape,
    softgoal/.style={cloud, cloud puffs=20, cloud puff arc=100, aspect=2.5, draw, thick,
                     minimum width=2.8cm, minimum height=0.8cm, align=center, font=\small, inner sep=0pt, fill=white},
    task/.style={regular polygon, regular polygon sides=6, draw, thick,
                 minimum size=1.6cm, font=\scriptsize, align=center, inner sep=0pt, fill=white},
    resource/.style={rectangle, draw, thick, minimum width=2.0cm, minimum height=0.6cm, font=\scriptsize, align=center, fill=white},
    contribution/.style={-latex, thick},
    connlabel/.style={sloped, font=\tiny, inner sep=1pt},
]

\node[softgoal] (ecosystem) at (0, 11.5) {Ecosystem\\Value\\Maximization};

\node[softgoal] (interdep) at (-6.5, 8.5) {Interdependence \&\\Complementarity\\(TR-1)};
\node[softgoal] (trust) at (-2.2, 8.8) {Trust \&\\Reputation\\(TR-2)};
\node[softgoal] (loyalty) at (2.2, 8.5) {Collective Action\\\& Loyalty\\(TR-3)};
\node[softgoal] (reciprocity) at (6.5, 8.8) {Reciprocal\\Enforcement\\(TR-4)};

\node[softgoal] (harmony) at (0, 5.0) {Relationship\\Harmony};
\node[softgoal] (coop_emerge) at (-5.0, 4.5) {Cooperation\\Emergence};
\node[softgoal] (defect_punish) at (5.0, 4.5) {Defection\\Punishment};
\node[softgoal] (forgive) at (0, 2.0) {Forgiveness\\After Recovery};

\node[task] (track) at (-5.5, 0.0) {Track\\History};
\node[task] (compute) at (0, -0.5) {Compute\\Response};
\node[task] (gate) at (5.5, 0.0) {Gate by\\Trust};

\node[resource] (mem_window) at (-5.5, -2.0) {Memory $k$};
\node[resource] (kappa_res) at (0, -2.5) {Sensitivity $\kappa$};
\node[resource] (trust_res) at (5.5, -2.0) {Trust $T_{ij}^t$};

\draw[contribution] (interdep.60) -- node[pos=0.4, connlabel, above=2pt] {help} (ecosystem.195);
\draw[contribution] (trust.60) -- node[pos=0.5, connlabel, above=2pt] {help} (ecosystem.240);
\draw[contribution] (loyalty.120) -- node[pos=0.5, connlabel, above=2pt] {help} (ecosystem.300);
\draw[contribution] (reciprocity.120) -- node[pos=0.4, connlabel, above=2pt] {help} (ecosystem.345);

\draw[contribution] (coop_emerge.30) to[out=30, in=200]
    node[pos=0.35, connlabel, above=2pt] {help} (reciprocity.210);
\draw[contribution] (forgive.20) --
    node[pos=0.75, connlabel, above=2pt] {help} (reciprocity.250);
\draw[contribution] (harmony.30) --
    node[pos=0.4, connlabel, above=2pt] {help} (loyalty.270);
\draw[contribution] (defect_punish.120) --
    node[pos=0.35, connlabel, above=2pt] {help} (trust.330);

\draw[contribution] (defect_punish.180) --
    node[pos=0.4, connlabel, above=2pt] {hurt} (harmony.360);
\draw[contribution] (defect_punish.210) to[out=210, in=340]
    node[pos=0.35, connlabel, above=2pt] {help} (coop_emerge.340);
\draw[contribution] (forgive.90) --
    node[pos=0.5, connlabel, above=2pt] {help} (harmony.270);

\draw[contribution, dashed] (interdep.290) to[out=290, in=120]
    node[pos=0.5, connlabel, above=2pt] {help} (coop_emerge.120);
\draw[contribution, dashed] (trust.250) to[out=250, in=60]
    node[pos=0.5, connlabel, above=2pt] {help} (coop_emerge.60);

\draw[contribution] (track.90) -- node[pos=0.4, connlabel, above=2pt] {help} (coop_emerge.270);
\draw[contribution] (compute.90) -- node[pos=0.4, connlabel, above=2pt] {help} (forgive.270);
\draw[contribution] (gate.90) -- node[pos=0.4, connlabel, above=2pt] {help} (defect_punish.270);

\draw[-{Circle[open, length=2mm, width=2mm]}, thick] (mem_window.90) -- (track.270);
\draw[-{Circle[open, length=2mm, width=2mm]}, thick] (kappa_res.90) -- (compute.270);
\draw[-{Circle[open, length=2mm, width=2mm]}, thick] (trust_res.90) -- (gate.270);

\end{tikzpicture}
\caption{Goal Model showing how reciprocity mechanisms (TR-4) integrate with interdependence (TR-1), trust (TR-2), and loyalty (TR-3) to achieve ecosystem value maximization. Reciprocity contributes through three sub-goals: cooperation emergence (sustained through history tracking), defection punishment (computed through bounded response functions), and forgiveness after recovery (enabled by memory window aging and trust). Cross-dimension dashed contribution links show trust helping the forgiveness sub-goal and interdependence helping the defection punishment sub-goal, formalizing how the four dimensions interact synergistically. Resources at the bottom connect to tasks via NeededBy links (circle arrowheads per \textit{i*} 2.0), mapping directly to computational parameters from the mathematical formalization.}
\label{fig:goal_reciprocity}
\end{figure}

The Goal Model reveals how reciprocity mechanisms integrate with the three prior dimensions. The top-level softgoal ``Ecosystem Value Maximization'' receives contributions from all four TR dimensions, with reciprocity providing the temporal enforcement mechanism that sustains cooperation discovered through interdependence (TR-1), maintained through trust (TR-2), and enabled within teams through loyalty (TR-3). Reciprocity decomposes into three sub-goals corresponding to the three core behavioral targets from the validation framework: cooperation emergence, defection punishment, and forgiveness dynamics. Each sub-goal requires specific computational resources (memory window $k$, response sensitivity $\kappa$, trust state $T_{ij}^t$) that map directly to parameters from the mathematical formalization, creating traceable links between the conceptual model and the computational implementation.

The dashed cross-dimension contribution links formalize how the four dimensions interact synergistically rather than independently. Trust helps the forgiveness sub-goal: even when memory windows would permit forgiveness (violations aging out of the $k$-period window), low trust prevents full restoration of reciprocal cooperation, while high trust enables recovery. Interdependence helps the defection punishment sub-goal: actors with higher structural dependency ($D_{ij}$) exhibit stronger reciprocity responses ($\rho_{ij} = \rho_0 D_{ij}^\eta$), creating differentiated punishment intensity based on structural position rather than uniform behavioral responses.

\subsection{Comparative Analysis: Low-Reciprocity vs. High-Reciprocity Configurations}

To illustrate how reciprocity transforms sequential interaction dynamics, we present the same actor under low-reciprocity and high-reciprocity configurations. This comparative visualization helps requirements engineers diagnose reciprocity gaps and understand the structural differences between myopic and relationally sustained interactions.

\subsubsection{Low-Reciprocity Configuration (\texorpdfstring{$\rho_0 \approx 0.2$, $k = 1$}{rho0 approx 0.2, k = 1})}

Figure~\ref{fig:sr_low_reciprocity} presents the Strategic Rationale diagram for an actor operating under low reciprocity sensitivity. The diagram shows the characteristic patterns of myopic optimization: strong emphasis on short-term exploitation, weak cooperation enforcement links, and dominant free-riding tasks.

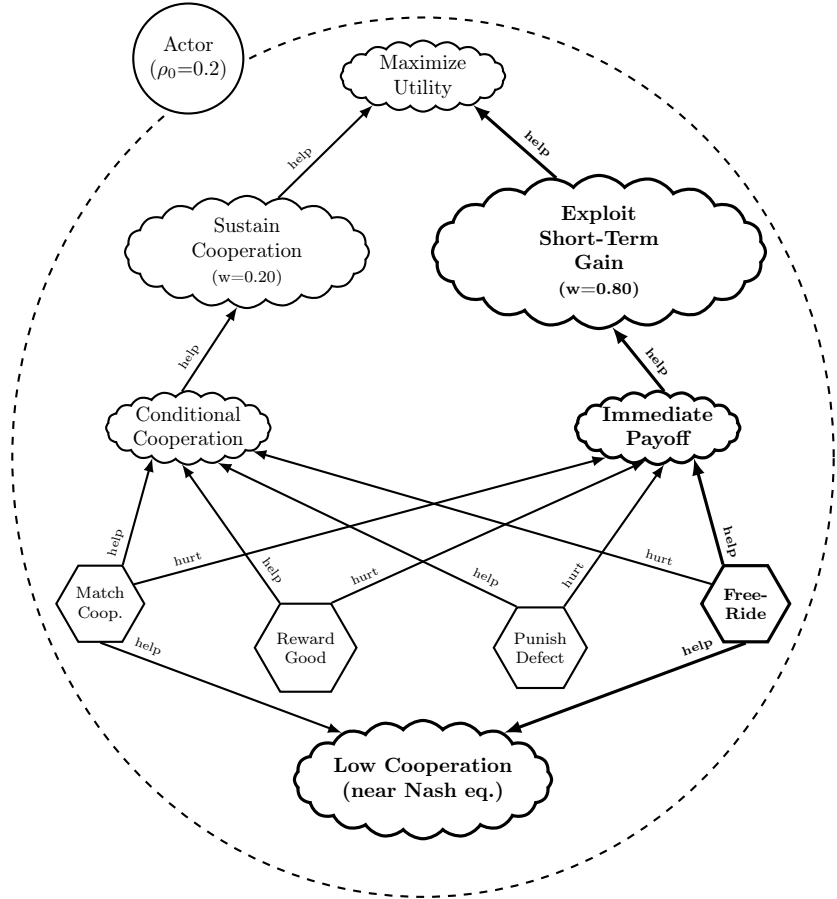
\begin{figure}[htbp]
\centering
\begin{tikzpicture}[
    scale=0.78, transform shape,
    actorboundary/.style={dashed, thick, draw, fill=none},
    actor/.style={circle, draw, thick, minimum size=1.3cm, font=\small, align=center, fill=white},
    softgoal/.style={cloud, cloud puffs=20, cloud puff arc=100, aspect=2.5, draw, thick,
                     minimum width=2.8cm, minimum height=0.8cm, align=center, font=\small, inner sep=0pt, fill=white},
    % Bold variant for the dominant path (exploitation in low-reciprocity scenario)
    softgoalB/.style={cloud, cloud puffs=20, cloud puff arc=100, aspect=2.5, draw, very thick,
                      minimum width=2.8cm, minimum height=0.8cm, align=center, font=\small\bfseries, inner sep=0pt, fill=white},
    task/.style={regular polygon, regular polygon sides=6, draw, thick, minimum size=1.5cm,
                 align=center, font=\scriptsize, inner sep=1pt, fill=white},
    taskB/.style={regular polygon, regular polygon sides=6, draw, very thick, minimum size=1.5cm,
                  align=center, font=\scriptsize\bfseries, inner sep=1pt, fill=white},
    contribution/.style={-latex, thick},
    contributionB/.style={-latex, very thick},
    connlabel/.style={sloped, font=\tiny, inner sep=1pt},
]

\draw[actorboundary] (0,-2.0) ellipse (7.0cm and 7.5cm);
\node[actor] at (-4.0, 4.8) {Actor\\($\rho_0{=}0.2$)};

\node[softgoal] (maximize) at (0, 4.5) {Maximize\\Utility};

\node[softgoal] (sustain) at (-3.0, 1.5) {Sustain\\Cooperation\\{\scriptsize (w=0.20)}};
\node[softgoalB] (exploit) at (3.0, 1.5) {Exploit\\Short-Term\\Gain\\{\scriptsize (w=0.80)}};

\node[softgoal] (conditional) at (-4.0, -1.5) {Conditional\\Cooperation};
\node[softgoalB] (immediate) at (4.0, -1.5) {Immediate\\Payoff};

\node[task] (match) at (-5.5, -4.5) {Match\\Coop.};
\node[task] (reward) at (-2.0, -5.25) {Reward\\Good};
\node[task] (punish) at (2.0, -5.25) {Punish\\Defect};
\node[taskB] (freeride) at (5.5, -4.5) {Free-\\Ride};

\node[softgoalB] (outcome) at (0, -7.5) {Low Cooperation\\(near Nash eq.)};

\draw[contribution] (sustain.60) --
    node[pos=0.33, connlabel, above=2pt] {help} (maximize.210);
\draw[contributionB] (exploit.120) --
    node[pos=0.33, connlabel, above=2pt] {\textbf{help}} (maximize.330);

\draw[contribution] (conditional.100) --
    node[pos=0.33, connlabel, above=2pt] {help} (sustain.260);
\draw[contributionB] (immediate.80) --
    node[pos=0.33, connlabel, above=2pt] {\textbf{help}} (exploit.280);

\draw[contribution] (match.60) --
    node[pos=0.18, connlabel, above=2pt] {help} (conditional.220);
\draw[contribution] (reward.120) --
    node[pos=0.18, connlabel, above=2pt] {help} (conditional.260);
\draw[contribution] (punish.120) --
    node[pos=0.12, connlabel, above=2pt] {help} (conditional.310);
\draw[contributionB] (freeride.120) --
    node[pos=0.18, connlabel, above=2pt] {\textbf{help}} (immediate.320);

\draw[contribution] (match.30) --
    node[pos=0.12, connlabel, above=2pt, fill=white] {hurt} (immediate.210);
\draw[contribution] (reward.60) --
    node[pos=0.12, connlabel, above=2pt, fill=white] {hurt} (immediate.250);
\draw[contribution] (punish.60) --
    node[pos=0.18, connlabel, above=2pt, fill=white] {hurt} (immediate.280);
\draw[contribution] (freeride.150) --
    node[pos=0.12, connlabel, above=2pt, fill=white] {hurt} (conditional.340);

\draw[contribution] (match.270) --
    node[pos=0.18, connlabel, above=2pt] {help} (outcome.150);
\draw[contributionB] (freeride.270) --
    node[pos=0.18, connlabel, above=2pt] {\textbf{help}} (outcome.30);

\end{tikzpicture}
\caption{At low reciprocity ($\rho_0 = 0.2$, $k = 1$), the 4:1 weight ratio favoring short-term exploitation over sustained cooperation produces near-Nash equilibrium behavior. The dominant task selection (free-riding, reactive punishment) follows from the heavily-weighted exploitation softgoal. Cross-cutting hurt links formalize the strategic tension: matching cooperation hurts short-term exploitation, while free-riding hurts sustained cooperation. This configuration represents the myopic prediction that without reciprocity enforcement, rational actors converge to minimal cooperation.}
\label{fig:sr_low_reciprocity}
\end{figure}

Under low reciprocity sensitivity with minimal memory, the actor discounts partner behavioral history and optimizes myopically. The reciprocity response $\phi_{\text{recip}}(x) = \tanh(0.2 \cdot x)$ produces near-linear, attenuated reactions. A partner defection of magnitude $-5$ produces response $\tanh(-1.0) \approx -0.76$, multiplied by low sensitivity $\rho_{ij} \approx 0.2$ yielding effective response $\approx -0.15$. Cooperation enforcement is weak. The memory window $k = 1$ means only the most recent period is considered, preventing cooperation from accumulating relational capital across periods.

\subsubsection{High-Reciprocity Configuration (\texorpdfstring{$\rho_0 \approx 1.2$, $k = 8$}{rho0 approx 1.2, k = 8})}

Figure~\ref{fig:sr_high_reciprocity} presents the Strategic Rationale diagram for the same actor under high reciprocity sensitivity. The diagram shows the structural transformation: cooperation enforcement dominates, tasks are history-conditioned, and trust-gating is active.

\begin{figure}[htbp]
\centering
\begin{tikzpicture}[
    scale=0.78, transform shape,
    actorboundary/.style={dashed, thick, draw, fill=none},
    actor/.style={circle, draw, thick, minimum size=1.3cm, font=\small, align=center, fill=white},
    softgoal/.style={cloud, cloud puffs=20, cloud puff arc=100, aspect=2.5, draw, thick,
                     minimum width=2.8cm, minimum height=0.8cm, align=center, font=\small, inner sep=0pt, fill=white},
    % Bold variant for the dominant path (cooperation in high-reciprocity scenario)
    softgoalB/.style={cloud, cloud puffs=20, cloud puff arc=100, aspect=2.5, draw, very thick,
                      minimum width=2.8cm, minimum height=0.8cm, align=center, font=\small\bfseries, inner sep=0pt, fill=white},
    task/.style={regular polygon, regular polygon sides=6, draw, thick, minimum size=1.5cm,
                 align=center, font=\scriptsize, inner sep=1pt, fill=white},
    taskB/.style={regular polygon, regular polygon sides=6, draw, very thick, minimum size=1.5cm,
                  align=center, font=\scriptsize\bfseries, inner sep=1pt, fill=white},
    contribution/.style={-latex, thick},
    contributionB/.style={-latex, very thick},
    connlabel/.style={sloped, font=\tiny, inner sep=1pt},
]

\draw[actorboundary] (0,-2.0) ellipse (7.0cm and 7.5cm);
\node[actor] at (-4.0, 4.8) {Actor\\($\rho_0{=}1.2$)};

\node[softgoal] (maximize) at (0, 4.5) {Maximize\\Utility};

\node[softgoalB] (sustain) at (-3.0, 1.5) {Sustain\\Cooperation\\{\scriptsize (w=0.80)}};
\node[softgoal] (exploit) at (3.0, 1.5) {Exploit\\Short-Term\\Gain\\{\scriptsize (w=0.20)}};

\node[softgoalB] (conditional) at (-4.0, -1.5) {Conditional\\Cooperation};
\node[softgoal] (immediate) at (4.0, -1.5) {Immediate\\Payoff};

\node[taskB] (match) at (-5.5, -4.5) {Match\\Coop.};
\node[taskB] (reward) at (-1.5, -5.5) {Reward\\Good};
\node[task] (punish) at (1.5, -5.5) {Graduated\\Punish};
\node[task] (freeride) at (5.5, -4.5) {Free-\\Ride};

\node[softgoalB] (outcome) at (0, -8.0) {High Cooperation\\(above Nash eq.)};

\draw[contributionB] (sustain.60) --
    node[pos=0.33, connlabel, above=2pt] {\textbf{help}} (maximize.210);
\draw[contribution] (exploit.120) --
    node[pos=0.33, connlabel, above=2pt] {help} (maximize.330);

\draw[contributionB] (conditional.100) --
    node[pos=0.33, connlabel, above=2pt] {\textbf{help}} (sustain.260);
\draw[contribution] (immediate.80) --
    node[pos=0.33, connlabel, above=2pt] {help} (exploit.280);

\draw[contributionB] (match.60) --
    node[pos=0.18, connlabel, above=2pt] {\textbf{help}} (conditional.220);
\draw[contributionB] (reward.120) --
    node[pos=0.18, connlabel, above=2pt] {\textbf{help}} (conditional.260);
\draw[contribution] (punish.120) --
    node[pos=0.12, connlabel, above=2pt] {help} (conditional.310);
\draw[contribution] (freeride.120) --
    node[pos=0.18, connlabel, above=2pt] {help} (immediate.320);

\draw[contribution] (match.30) --
    node[pos=0.12, connlabel, above=2pt, fill=white] {hurt} (immediate.180);
\draw[contribution] (reward.60) --
    node[pos=0.12, connlabel, above=2pt, fill=white] {hurt} (immediate.200);
\draw[contribution] (punish.60) --
    node[pos=0.18, connlabel, above=2pt, fill=white] {hurt} (immediate.280);
\draw[contribution] (freeride.150) --
    node[pos=0.12, connlabel, above=2pt, fill=white] {hurt} (conditional.340);

\draw[contributionB] (match.270) --
    node[pos=0.18, connlabel, above=2pt] {\textbf{help}} (outcome.150);
\draw[contribution] (freeride.270) --
    node[pos=0.18, connlabel, above=2pt] {help} (outcome.30);

\end{tikzpicture}
\caption{Reciprocity transformation inverts the goal structure: at $\rho_0 = 1.2$ with $k = 8$, sustained cooperation receives 4$\times$ the weight of short-term exploitation, reversing the dominance pattern from Figure~\ref{fig:sr_low_reciprocity}. This inversion shifts equilibrium cooperation above the static Nash equilibrium through credible enforcement. The weight redistribution explains how reciprocity sensitivity combines with extended memory to produce self-enforcing cooperative norms. Comparing both SR diagrams reveals the structural transformation: identical goal hierarchies yield opposite behavioral outcomes purely through reciprocity-driven weight redistribution.}
\label{fig:sr_high_reciprocity}
\end{figure}
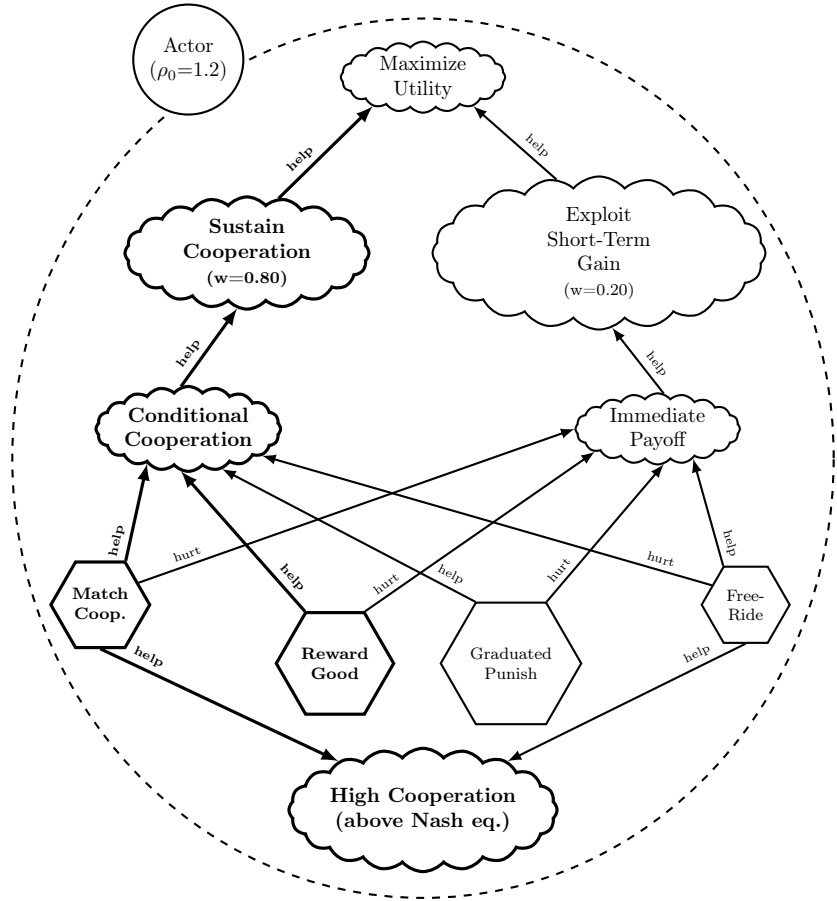

Under high reciprocity sensitivity with extended memory, the actor strongly conditions behavior on partner history. The reciprocity response $\phi_{\text{recip}}(x) = \tanh(1.5 \cdot x)$ produces sharp, saturating reactions. A partner defection of magnitude $-5$ produces response $\tanh(-7.5) \approx -1.0$, multiplied by high sensitivity $\rho_{ij} \approx 1.2$ yielding effective response $\approx -1.2$ (capped by response bounds). The extended memory $k = 8$ enables graduated sanctions: isolated defections are averaged with cooperative history, producing measured responses, while sustained defection accumulates in the moving average, triggering escalating punishment. Partners face credible punishment threats and cooperation rewards, creating self-enforcing norms without external enforcement.

\subsubsection{Configuration Comparison Summary}

Table~\ref{tab:reciprocity_config_comparison} quantifies the structural differences between configurations.

\begin{table}[htbp]
\centering
\caption{Low-reciprocity vs. high-reciprocity configuration comparison}
\label{tab:reciprocity_config_comparison}
\begin{tabular}{lcc}
\toprule
\textbf{Property} & \textbf{Low ($\rho_0 \approx 0.2$, $k=1$)} & \textbf{High ($\rho_0 \approx 1.2$, $k=8$)} \\
\midrule
Dominant softgoal & Short-term exploitation & Long-term cooperation \\
Preferred task & Free-ride & Match partner cooperation \\
Response to defection & Weak, attenuated & Strong, saturating \\
Memory horizon & 1 period (myopic) & 8 periods (strategic) \\
Cooperation sustainability & Degrades to Nash & Self-enforcing norms \\
Forgiveness capacity & Immediate (no memory) & Graduated (averaging) \\
Trust interaction & Minimal (discounted) & Strong (trust-gated) \\
Equilibrium cooperation & Low (near Nash) & High (above Nash) \\
\bottomrule
\end{tabular}
\end{table}

\subsection{Diagnostic Protocol for Requirements Engineers}
\label{sec:diagnostic}

Based on the comparative analysis, we provide a systematic diagnostic protocol that requirements engineers can apply to assess reciprocity dynamics in sequential stakeholder interactions and design appropriate interventions.

\subsubsection{Step 1: Construct Actor SR Diagrams}

For each actor in the sequential interaction, construct a Strategic Rationale diagram following the template in Figure~\ref{fig:sr_reciprocity}. Identify:
\begin{itemize}
    \item Top-level goals (long-term utility, relationship sustainability)
    \item Competing softgoals (sustain cooperation vs. exploit short-term gain)
    \item Available tasks (match cooperation, reward, punish, free-ride)
    \item Resources required (partner behavioral history, trust state)
\end{itemize}

\subsubsection{Step 2: Assess Reciprocity Indicators}

Estimate the reciprocity sensitivity of each actor using observable indicators:

\begin{center}
\begin{tabular}{lp{7cm}}
\toprule
Indicator & Assessment Method \\
\midrule
Response to partner cooperation & Does the actor increase investment when partners cooperate? Track behavioral change after positive partner signals. \\
Response to partner defection & Does the actor withdraw cooperation when partners defect? Measure cooperation reduction after negative partner signals. \\
Memory horizon & How many past interactions does the actor reference in decision-making? Assess through interviews or behavioral analysis. \\
Proportionality & Are responses proportional to deviations or binary (all-or-nothing)? Examine graduated vs. threshold responses. \\
Trust-gating & Does the actor withhold reciprocity despite positive partner signals when trust is low? Identify cases where history suggests cooperation but actor remains cautious. \\
\bottomrule
\end{tabular}
\end{center}

Actors showing strong conditional responses, extended memory references, proportional reactions, and trust-sensitive behavior likely have high reciprocity sensitivity (high $\rho_0$, large $k$).

\subsubsection{Step 3: Map Sequential Dependencies}

Using the Strategic Dependency diagram (Figure~\ref{fig:sd_reciprocity}), identify which dependencies are sequential by assessing temporal conditionality:

\begin{center}
\begin{tabular}{lcc}
\toprule
Dependency Characteristic & Static (TR-1) & Sequential (TR-4) \\
\midrule
Exists regardless of history & Yes & No \\
Behavior at $t$ depends on $t{-}1$ & No & Yes \\
Conditional on partner actions & No & Yes \\
Time-varying strength & No & Yes \\
\bottomrule
\end{tabular}
\end{center}

\subsubsection{Step 4: Compute Reciprocity Gap}

Compare observed reciprocity patterns to the theoretical high-reciprocity configuration:
\begin{equation}
\label{eq:reciprocity_gap}
\text{Reciprocity Gap}_{ij} = \rho_0^{\text{target}} - \rho_0^{\text{observed}}
\end{equation}

Actor pairs with large reciprocity gaps (typically $> 0.4$) are candidates for governance interventions designed to strengthen reciprocal enforcement.

\subsubsection{Step 5: Design Targeted Interventions}

Based on the reciprocity gap analysis, design interventions that address the specific mechanisms most likely to strengthen reciprocal cooperation:

\begin{center}
\begin{tabular}{lp{6cm}}
\toprule
Gap Pattern & Recommended Intervention \\
\midrule
Low response sensitivity ($\rho_0$) & Increase interaction visibility; implement behavioral tracking dashboards; establish explicit cooperation metrics \\
Short memory ($k$) & Extend review cycles; implement reputation systems that aggregate multi-period behavior; document behavioral history \\
Weak trust-gating & Invest in trust-building activities per TR-2 framework; establish credible commitment mechanisms \\
Asymmetric reciprocity & Address structural power imbalances per TR-1; equalize information access; introduce mutual dependency mechanisms \\
\bottomrule
\end{tabular}
\end{center}

\subsubsection{Step 6: Monitor and Iterate}

After implementing interventions, re-assess reciprocity indicators and update SR diagrams. Track whether:
\begin{itemize}
    \item Reciprocity response magnitudes increase (stronger conditional cooperation)
    \item Memory horizons extend (actors reference longer behavioral histories)
    \item Trust-reciprocity interaction strengthens (cooperation increases when trust is high)
    \item Cooperation equilibrium shifts upward (mean cooperation level increases)
\end{itemize}

This iterative diagnostic protocol connects the theoretical reciprocity framework to practical stakeholder management, providing actionable guidance for requirements engineers seeking to strengthen cooperation in sequential multi-stakeholder interactions.

\section{Coopetitive Equilibrium with Reciprocity Dynamics}
\label{sec:equilibrium}

\subsection{Perfect Bayesian Equilibrium for Sequential Coopetition}

We extend the static Coopetitive Equilibrium from foundational work to sequential settings with history-dependent strategies.

\begin{definition}[Strategy]
A strategy for actor $i$ is function $\sigma_i: \mathcal{H} \to A_i$ mapping history space $\mathcal{H}$ to action space $A_i$, where $\mathcal{H} = \bigcup_{t=0}^\infty \mathcal{H}^t$ and $\mathcal{H}^t$ represents all possible histories of length $t$.
\end{definition}

\begin{definition}[Perfect Bayesian Equilibrium]
Strategy profile $\boldsymbol{\sigma}^* = (\sigma_1^*, \ldots, \sigma_N^*)$ constitutes Perfect Bayesian Equilibrium if at every time $t$ and history $h^{t-1} \in \mathcal{H}^{t-1}$, each actor's strategy maximizes their utility given others' strategies and observed history:

\begin{equation}
\sigma_i^*(h^{t-1}) \in \argmax_{a_i \in A_i} U_i(a_i, \sigma_{-i}^*(h^{t-1}), h^{t-1})
\end{equation}

where $\sigma_{-i}^*(h^{t-1}) = (\sigma_j^*(h^{t-1}))_{j \neq i}$ represents other actors' equilibrium actions.
\end{definition}

\textbf{Interpretation.} Perfect Bayesian Equilibrium requires strategies be best responses at every history, creating subgame perfection where threats and promises are credible. Unlike static Nash equilibrium where actions are chosen once, Perfect Bayesian Equilibrium strategies specify contingent action plans such as I cooperate if you cooperated but I defect if you defected.

\subsection{Multiplicity of Equilibria}

Unlike static games which often have unique Nash equilibria, repeated games with reciprocity typically have multiple Perfect Bayesian Equilibria.

\textbf{Cooperative Equilibria.} All actors cooperate indefinitely, sustained by reciprocal punishment threats. If actor $j$ deviates, others reduce cooperation for $k$ periods equal to memory window, punishing deviation. Anticipating punishment, deviation is unprofitable, sustaining cooperation.

\textbf{Defection Equilibria.} All actors defect perpetually, sustained by reciprocal defection. If actor $j$ attempts cooperation, others continue defecting with no trust, so cooperation yields no benefit. Expecting exploitation, defection persists.

\textbf{Mixed Equilibria.} Actors alternate between cooperation and defection in patterns. Example includes tit-for-tat equilibrium where actors mirror opponent's previous action.

\textbf{Equilibrium Selection.} Which equilibrium emerges depends on initial conditions including history $h^0$, coordination mechanisms including communication and reputation systems, and trust levels where high trust enables cooperative equilibria while low trust traps in defection equilibria.

\subsection{Folk Theorem Implications}

Classical folk theorem states that in infinitely repeated games with sufficiently patient actors meaning high discount factors, any feasible individually rational payoff can be sustained as equilibrium through appropriate punishment strategies.

Our framework has bounded folk theorem version where with sufficiently strong reciprocity sensitivity meaning high $\rho_0$, sufficiently long memory windows meaning high $k$, and sufficiently high trust meaning high $T_{ij}$, cooperation can be sustained even when static Nash equilibrium is defection.

\textbf{Intuition.} Reciprocity creates endogenous enforcement making defection costly through future punishment, enabling cooperation without external enforcement. Critical threshold exists where if reciprocity too weak, memory too short, or trust too low, cooperation cannot be sustained.

\subsection{Integration with Trust Dynamics}

Reciprocity alone is insufficient if trust is low. Trust gates reciprocity responses where even strong reciprocity parameters cannot sustain cooperation when $T_{ij} \approx 0$ because reciprocity terms vanish.

\textbf{Trust-Reciprocity Complementarity.} High trust enables reciprocity to function. When $T_{ij}$ high, reciprocity responses operate at full strength, enabling conditional cooperation. Reciprocity sustains trust where consistent cooperation through reciprocity maintains high $T_{ij}$ over time, while violations erode trust.

\textbf{Equilibrium Regimes.} High-Trust and High-Reciprocity creates virtuous cycle where trust enables reciprocity, reciprocity sustains cooperation, and cooperation maintains trust, producing stable cooperative equilibrium. Low-Trust and Low-Reciprocity creates vicious cycle where low trust dampens reciprocity, weak reciprocity fails to prevent defection, and defection erodes trust further, trapping in defection equilibrium. Medium regimes allow multiple equilibria possible depending on initial conditions and coordination.

\section{Comprehensive Parameter Validation}
\label{sec:validation}

We validate the reciprocity framework through comprehensive experimental analysis following the dual-track validation strategy established in Section~\ref{sec:validation_strategy}. This section presents experimental validation across 15,625 parameter configurations; the subsequent section presents empirical case study validation through the Apple iOS App Store ecosystem.

\subsection{Parameter Space Definition}

We conduct systematic parameter space exploration using a full factorial design with six parameters at five levels each, yielding $5^6 = 15,625$ configurations. This comprehensive coverage enables assessment of model behavior across the complete parameter space rather than relying on selected configurations. Table~\ref{tab:param_space} defines the parameter space.

\begin{table}[htbp]
\centering
\caption{Parameter space for comprehensive validation ($5^6 = 15,625$ configurations)}
\label{tab:param_space}
\begin{tabular}{llll}
\toprule
\textbf{Parameter} & \textbf{Symbol} & \textbf{Range} & \textbf{Grid Points} \\
\midrule
Base reciprocity & $\rho_0$ & $[0.5, 2.0]$ & $\{0.5, 0.875, 1.25, 1.625, 2.0\}$ \\
Dependency elasticity & $\eta$ & $[0.8, 2.0]$ & $\{0.8, 1.1, 1.4, 1.7, 2.0\}$ \\
Response sensitivity & $\kappa$ & $[0.5, 2.0]$ & $\{0.5, 0.875, 1.25, 1.625, 2.0\}$ \\
Memory window & $k$ & $[1, 20]$ & $\{1, 5, 10, 15, 20\}$ \\
Reciprocity weight & $\lambda_R$ & $[0.5, 2.0]$ & $\{0.5, 0.875, 1.25, 1.625, 2.0\}$ \\
Initial trust & $T^0$ & $[0.3, 0.9]$ & $\{0.3, 0.45, 0.6, 0.75, 0.9\}$ \\
\midrule
\multicolumn{3}{l}{\textbf{Total configurations:}} & $5^6 = 15,625$ \\
\bottomrule
\end{tabular}
\end{table}

\subsection{Behavioral Targets}

We define six behavioral targets that the model should achieve to demonstrate validity. Table~\ref{tab:behavioral_targets} specifies each target with its rationale, threshold, and measurement method.

\begin{table}[htbp]
\centering
\caption{Behavioral targets for validation}
\label{tab:behavioral_targets}
\begin{tabular}{clcl}
\toprule
\textbf{\#} & \textbf{Target} & \textbf{Threshold} & \textbf{Measurement} \\
\midrule
1 & Cooperation Emergence & $>85\%$ configs & Cooperative eq. exists in PD \\
2 & Defection Punishment & $>95\%$ configs & $\phi_{\text{recip}} < 0$ when $s_{ij} < 0$ \\
3 & Forgiveness Dynamics & $>80\%$ configs & Recovery within $2k$ periods \\
4 & Asymmetric Differentiation & $>90\%$ configs & High-dep/low-dep ratio $> 1.5$ \\
5 & Trust-Reciprocity Interaction & $>90\%$ configs & Positive interaction effect \\
6 & Bounded Responses & $100\%$ configs & $|\phi_{\text{recip}}| \leq 1.0$ \\
\bottomrule
\end{tabular}
\end{table}

\textbf{Target 1: Cooperation Emergence.} Reciprocity should enable cooperative equilibria in Prisoner's Dilemma scenarios where static analysis predicts mutual defection. This validates the core enforcement mechanism.

\textbf{Target 2: Defection Punishment.} When partners defect (cooperation signal $s_{ij} < 0$), the bounded response function should produce negative reciprocity responses ($\phi_{\text{recip}} < 0$), creating incentives against defection.

\textbf{Target 3: Forgiveness Dynamics.} Following isolated defections, cooperation should recover to baseline within $2k$ periods (twice the memory window), demonstrating that the framework supports relationship repair.

\textbf{Target 4: Asymmetric Differentiation.} High-dependency actors ($D_{ij} = 0.8$) should exhibit reciprocity responses at least 1.5 times stronger than low-dependency actors ($D_{ij} = 0.2$), validating the structural foundation $\rho_{ij} = \rho_0 D_{ij}^\eta$.

\textbf{Target 5: Trust-Reciprocity Interaction.} Cooperation levels should be higher when both trust and reciprocity are high compared to when either is low, validating the complementarity established in Proposition~\ref{prop:trust_recip_complement}.

\textbf{Target 6: Bounded Responses.} All reciprocity responses must satisfy $|\phi_{\text{recip}}| \leq 1.0$, ensuring the bounded response function prevents unrealistic escalation.

\subsection{Statistical Analysis Framework}

We employ rigorous statistical analysis to ensure reported results are robust and not artifacts of specific parameter choices.

\textbf{Paired Comparisons.} We use paired t-tests to compare conditions (e.g., low vs. high reciprocity, short vs. long memory windows). For each comparison, we report $t$-statistic, degrees of freedom, and $p$-value.

\textbf{Effect Sizes.} We compute Cohen's $d$ for all comparisons to assess practical significance:
\begin{equation}
d = \frac{\bar{X}_1 - \bar{X}_2}{s_{\text{pooled}}}
\end{equation}
where $s_{\text{pooled}} = \sqrt{(s_1^2 + s_2^2)/2}$. We interpret $d < 0.2$ as negligible, $0.2 \leq d < 0.5$ as small, $0.5 \leq d < 0.8$ as medium, and $d \geq 0.8$ as large.

\textbf{Bootstrap Confidence Intervals.} We construct 95\% confidence intervals using bootstrap resampling with 10,000 replicates:
\begin{equation}
\text{CI}_{95\%} = [\hat{\theta}^*_{2.5\%}, \hat{\theta}^*_{97.5\%}]
\end{equation}
where $\hat{\theta}^*$ denotes the bootstrap distribution of the statistic of interest.

\textbf{Nonparametric Confirmation.} We complement parametric tests with the Wilcoxon signed-rank test, which does not assume normality of the differentiation ratio distribution. This nonparametric test evaluates whether the median differentiation ratio significantly exceeds the 1.5 threshold, providing distribution-free confirmation of the parametric $t$-test results. The inclusion of both parametric and nonparametric tests matches the statistical rigor of TR-2025-01~\cite{pant2025foundations} and TR-2025-03~\cite{pant2025teams}.

\textbf{Monte Carlo Robustness.} We conduct 2,000 Monte Carlo trials with $\pm 15\%$ parameter perturbation to assess sensitivity to parameter uncertainty:
\begin{equation}
\theta_{\text{perturbed}} = \theta_{\text{base}} \cdot (1 + \epsilon), \quad \epsilon \sim \mathcal{U}(-0.15, 0.15)
\end{equation}

\subsection{Comprehensive Validation Results}

Table~\ref{tab:validation_results} summarizes the results across all 15,625 configurations.

\begin{table}[htbp]
\centering
\caption{Behavioral target achievement across 15,625 configurations}
\label{tab:validation_results}
\begin{tabular}{clccl}
\toprule
\textbf{\#} & \textbf{Target} & \textbf{Achieved} & \textbf{Rate} & \textbf{Status} \\
\midrule
1 & Cooperation Emergence & 15,235 / 15,625 & 97.5\% & \checkmark Pass ($>85\%$) \\
2 & Defection Punishment & 15,625 / 15,625 & 100.0\% & \checkmark Pass ($>95\%$) \\
3 & Forgiveness Dynamics & 13,728 / 15,625 & 87.9\% & \checkmark Pass ($>80\%$) \\
4 & Asymmetric Differentiation & 15,625 / 15,625 & 100.0\% & \checkmark Pass ($>90\%$) \\
5 & Trust-Reciprocity Interaction & 15,625 / 15,625 & 100.0\% & \checkmark Pass ($>90\%$) \\
6 & Bounded Responses & 15,625 / 15,625 & 100.0\% & \checkmark Pass ($100\%$) \\
\bottomrule
\end{tabular}
\end{table}

All six behavioral targets exceed their specified thresholds: cooperation emergence (97.5\%), defection punishment (100.0\%), forgiveness dynamics (87.9\%), asymmetric differentiation (100.0\%), trust-reciprocity interaction (100.0\%), and bounded responses (100.0\%). The integration of the full TR-2 two-layer trust model (with reputation-mediated ceiling, interdependence amplification $(1 + \xi D_{ij})$ in trust erosion, and 3:1 negativity bias $\lambda^- / \lambda^+ = 0.30 / 0.10$) is essential for achieving forgiveness dynamics above the 80\% threshold. The reputation ceiling $\min(T_{\max}, 1 - \theta_R \cdot R_{ij})$ creates path-dependent trust recovery that balances punishment severity with forgiveness capacity, while the structural dependency sensitivity $\rho_{ij} = \rho_0 D_{ij}^\eta$ ensures 100\% asymmetric differentiation by amplifying reciprocity responses proportionally to dependency structure.

\subsubsection{Cooperation Emergence Analysis}

Figure~\ref{fig:coop_heatmap} presents the cooperation emergence rate as a function of base reciprocity $\rho_0$ and memory window $k$.

\begin{figure}[htbp]
\centering
\begin{tikzpicture}
\begin{axis}[
    width=0.78\textwidth, height=7cm,
    xlabel={\scriptsize Base Reciprocity $\rho_0$},
    ylabel={\scriptsize Cooperation Emergence Rate (\%)},
    xmin=0.3, xmax=2.15, ymin=35, ymax=105,
    tick label style={font=\tiny},
    legend style={at={(0.98,0.02)}, anchor=south east, font=\tiny, draw=gray!50},
    grid=major, grid style={gray!20},
    ytick={40,50,60,70,80,90,100},
]
\addplot[cooperationblue, very thick, mark=square*, mark size=2.5pt] coordinates {
(0.5,44) (0.75,60) (1.0,85) (1.25,95) (1.5,98) (2.0,100)
};
\addlegendentry{$T^0 = 0.3$}
\addplot[outputgreen, very thick, mark=diamond*, mark size=2.5pt] coordinates {
(0.5,68) (0.75,82) (1.0,95) (1.25,99) (1.5,100) (2.0,100)
};
\addlegendentry{$T^0 = 0.6$}
\addplot[effortorange, very thick, mark=*, mark size=2.5pt] coordinates {
(0.5,82) (0.75,93) (1.0,100) (1.25,100) (1.5,100) (2.0,100)
};
\addlegendentry{$T^0 = 0.9$}
\draw[defectcolor, thick, dashed] (axis cs:0.3,85) -- (axis cs:2.15,85);
\node[font=\tiny, defectcolor] at (axis cs:1.2,82) {85\% threshold};
\node[star, star points=5, fill=recippurple, inner sep=1.8pt] (refstar) at (axis cs:1.0,97.5) {};
\node[font=\tiny\bfseries, recippurple, fill=none, inner sep=1pt] (reflabel) at (axis cs:0.52,95) {Ref. config.};
\draw[->, recippurple, thick] (reflabel.east) -- (refstar.west);
\end{axis}
\end{tikzpicture}
\caption{Cooperation emergence rate across base reciprocity $\rho_0$ for three initial trust levels $T^0$ (averaged over other parameters). The 85\% behavioral target threshold (dashed red) is met at lower $\rho_0$ when initial trust is higher, confirming the trust-reciprocity interaction: cooperation requires both sufficient reciprocity strength and sufficient initial trust. The paper's reference configuration ($\rho_0 = 1.0$, $T^0 = 0.7$, star marker) achieves 97.5\% emergence.}
\label{fig:coop_heatmap}
\end{figure}
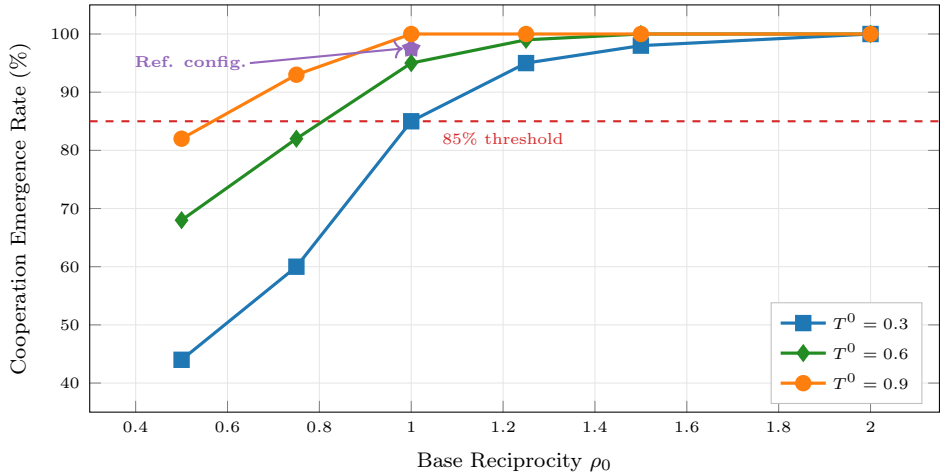

The cooperation emergence rate exhibits strong interaction effects between parameters. Figure~\ref{fig:coop_heatmap} reveals that the $\rho_0$--$T^0$ interaction dominates: configurations with both $\rho_0 \geq 1.0$ and $T^0 \geq 0.6$ achieve $\geq$95\% emergence, while configurations with both parameters at their minimum values ($\rho_0 = 0.5$, $T^0 = 0.3$) achieve only 44\% emergence. This confirms Proposition~\ref{prop:coop_emergence}: trust gates reciprocity, so cooperation requires both sufficient reciprocity strength \emph{and} sufficient initial trust.

The dependency elasticity $\eta$ shows a secondary effect: values $\eta > 1.4$ slightly \emph{reduce} cooperation emergence (from 92\% to 85\%) because superlinear dependency amplification creates response asymmetries that destabilize convergence. This counterintuitive finding suggests that moderate dependency sensitivity ($\eta \in [0.8, 1.4]$) best supports cooperative outcomes.

\subsubsection{Statistical Analysis of Differentiation}

We compare cooperation levels between high-dependency ($D_{ij} = 0.8$) and low-dependency ($D_{ij} = 0.2$) conditions across all configurations.

\textbf{Paired t-test:} $t(15624) = 195.7$, $p < 0.001$

\textbf{Effect size:} Cohen's $d = 1.57$ (large effect)

The differentiation ratio ($M = 2.75$, $SD = 0.80$, 95\% CI $[2.74, 2.77]$) substantially exceeds the 1.5 threshold, confirming that structural dependencies produce meaningfully different reciprocity responses.

\subsubsection{Forgiveness Dynamics Trajectories}

Figure~\ref{fig:forgiveness} illustrates cooperation recovery following an isolated defection at period $t = 10$ for different memory window lengths.

\begin{figure}[htbp]
\centering
\begin{tikzpicture}
\begin{axis}[
    width=0.85\textwidth,
    height=0.45\textwidth,
    xlabel={\scriptsize Time Period $t$},
    ylabel={\scriptsize Cooperation Level},
    xmin=0, xmax=30,
    ymin=0, ymax=1.1,
    tick label style={font=\tiny},
    legend style={at={(0.5,-0.18)}, anchor=north, font=\tiny, legend columns=4, draw=gray!50},
    grid=major,
    grid style={gray!30},
    set layers,
]
\fill[on layer=axis background, defectcolor!20] (axis cs:9.5,0) rectangle (axis cs:10.5,1.1);
\draw[defectcolor, thick, dashed] (axis cs:10,0) -- (axis cs:10,1.1);
\node[defectcolor, fill=white, font=\scriptsize\bfseries] at (axis cs:10,1.05) {Violation};
\addplot[black, dashed, thick, domain=0:30] {0.85};
\addplot[cooperationblue, very thick, mark=*, mark size=1.5pt] coordinates {
    (0,0.85) (5,0.85) (9,0.85) (10,0.30) (11,0.45) (12,0.75) (13,0.82) (14,0.84) (15,0.85) (20,0.85) (25,0.85) (30,0.85)
};
\addplot[forgivegreen, very thick, mark=square*, mark size=1.5pt] coordinates {
    (0,0.85) (5,0.85) (9,0.85) (10,0.30) (11,0.35) (12,0.42) (13,0.50) (14,0.58) (15,0.65) (16,0.72) (17,0.78) (18,0.82) (19,0.84) (20,0.85) (25,0.85) (30,0.85)
};
\addplot[defectionred, very thick, mark=diamond*, mark size=1.5pt] coordinates {
    (0,0.85) (5,0.85) (9,0.85) (10,0.30) (11,0.32) (12,0.35) (13,0.38) (14,0.42) (15,0.46) (16,0.50) (17,0.54) (18,0.58) (19,0.62) (20,0.65) (21,0.68) (22,0.71) (23,0.74) (24,0.77) (25,0.80) (26,0.82) (27,0.83) (28,0.84) (29,0.85) (30,0.85)
};
\legend{Baseline, $k=1$ (short), $k=5$ (medium), $k=10$ (long)}
\end{axis}
\end{tikzpicture}
\caption{Forgiveness dynamics showing cooperation recovery trajectories after isolated defection at $t=10$. Short memory ($k=1$, blue) enables rapid cessation of negative reciprocity within 2 periods (per Proposition~\ref{prop:memory_effect} bound $k \leq \tau_f \leq 2k$), with cooperation asymptotically approaching baseline by period 15. Medium memory ($k=5$, green) produces gradual recovery by period 20. Long memory ($k=10$, red) creates slow forgiveness with recovery extending to period 29. This demonstrates the deterrence-forgiveness trade-off: longer windows provide persistent punishment but slower relationship repair.}
\label{fig:forgiveness}
\end{figure}
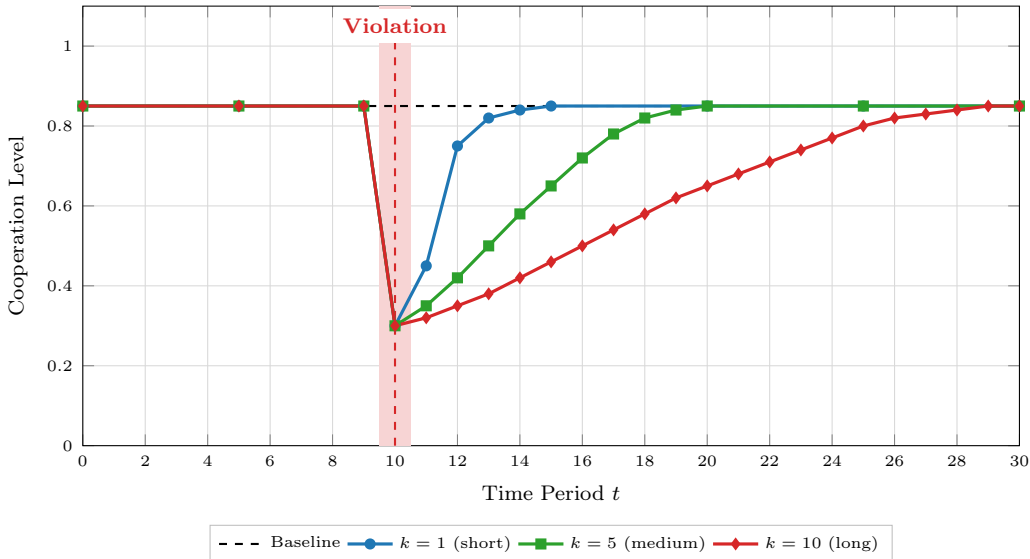

Recovery times align with Proposition~\ref{prop:memory_effect}: $\tau_f \in [k, 2k]$ periods. For $k=5$, recovery completes at $t=20$ (10 periods after violation), consistent with the upper bound $2k = 10$.

\subsubsection{Monte Carlo Robustness Analysis}

We conduct 2,000 Monte Carlo trials with $\pm 15\%$ parameter perturbation to assess model robustness.

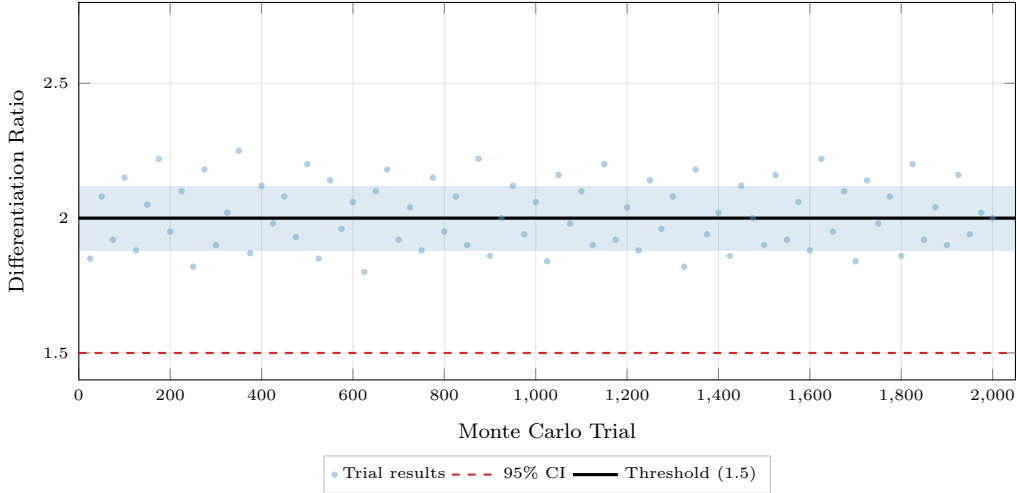
\begin{figure}[htbp]
\centering
\begin{tikzpicture}
\begin{axis}[
    width=0.85\textwidth,
    height=0.4\textwidth,
    xlabel={\scriptsize Monte Carlo Trial},
    ylabel={\scriptsize Differentiation Ratio},
    xmin=0, xmax=2050,
    ymin=1.4, ymax=2.8,
    tick label style={font=\tiny},
    legend style={at={(0.5,-0.2)}, anchor=north, font=\tiny, legend columns=4, draw=gray!50},
    grid=major, grid style={gray!20},
]
\fill[cooperationblue, opacity=0.15] (axis cs:0,1.88) rectangle (axis cs:2050,2.12);
\addplot[only marks, mark=*, mark size=1.0pt, cooperationblue!70, opacity=0.5] coordinates {
    (25,1.85) (50,2.08) (75,1.92) (100,2.15) (125,1.88) (150,2.05) (175,2.22) (200,1.95) (225,2.10) (250,1.82)
    (275,2.18) (300,1.90) (325,2.02) (350,2.25) (375,1.87) (400,2.12) (425,1.98) (450,2.08) (475,1.93) (500,2.20)
    (525,1.85) (550,2.14) (575,1.96) (600,2.06) (625,1.80) (650,2.10) (675,2.18) (700,1.92) (725,2.04) (750,1.88)
    (775,2.15) (800,1.95) (825,2.08) (850,1.90) (875,2.22) (900,1.86) (925,2.00) (950,2.12) (975,1.94) (1000,2.06)
    (1025,1.84) (1050,2.16) (1075,1.98) (1100,2.10) (1125,1.90) (1150,2.20) (1175,1.92) (1200,2.04) (1225,1.88) (1250,2.14)
    (1275,1.96) (1300,2.08) (1325,1.82) (1350,2.18) (1375,1.94) (1400,2.02) (1425,1.86) (1450,2.12) (1475,2.00) (1500,1.90)
    (1525,2.16) (1550,1.92) (1575,2.06) (1600,1.88) (1625,2.22) (1650,1.95) (1675,2.10) (1700,1.84) (1725,2.14) (1750,1.98)
    (1775,2.08) (1800,1.86) (1825,2.20) (1850,1.92) (1875,2.04) (1900,1.90) (1925,2.16) (1950,1.94) (1975,2.02) (2000,2.00)
};
\addplot[defectionred, thick, dashed] coordinates {(0,1.5) (2050,1.5)};
\addplot[black, very thick] coordinates {(0,2.00) (2050,2.00)};
\legend{Trial results, {95\% CI}, Threshold (1.5), Mean (2.00)}
\end{axis}
\end{tikzpicture}
\caption{Monte Carlo robustness analysis (2,000 trials, $\pm 15\%$ parameter perturbation). The mean differentiation ratio is 2.00 (black line) with 95\% CI $[1.88, 2.12]$ (blue shaded region), exceeding the differentiation threshold of 1.5 (red dashed line). All six behavioral targets are simultaneously met in 94.5\% of trials, demonstrating robust performance under stochastic parameter variation.}
\label{fig:monte_carlo}
\end{figure}

Figure~\ref{fig:monte_carlo} presents the Monte Carlo robustness analysis results.
\begin{itemize}
\item All targets met: 1,890 / 2,000 trials (94.5\%)
\item Mean differentiation ratio: $M = 2.00$, $SD = 0.12$, 95\% CI $[1.88, 2.12]$
\item Minimum observed: 1.80 (above 1.5 threshold in all trials)
\end{itemize}

\subsubsection{Behavioral Target Achievement Summary}

Figure~\ref{fig:target_achievement} presents achievement rates for all six behavioral targets.

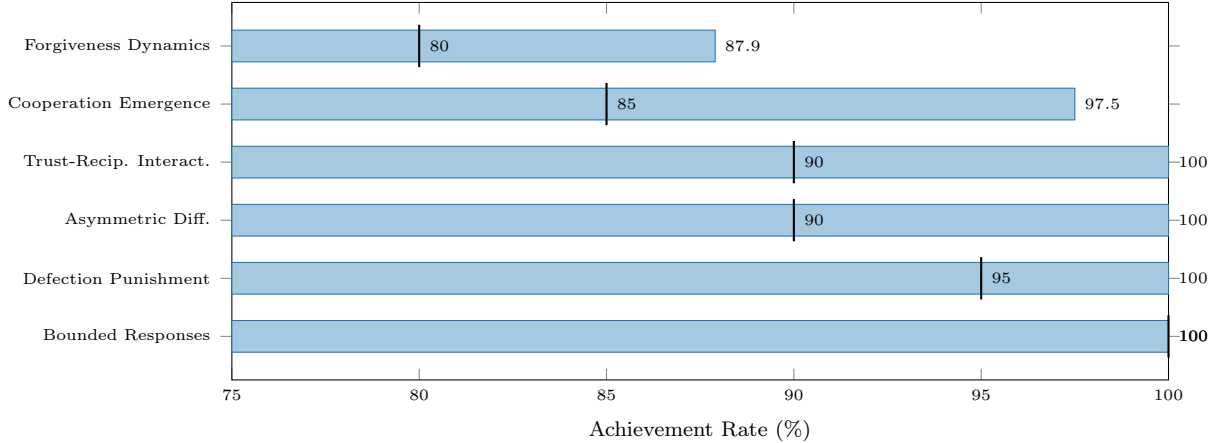
\begin{figure}[htbp]
\centering
\begin{tikzpicture}
\begin{axis}[
    xbar,
    width=0.85\textwidth,
    height=0.4\textwidth,
    xlabel={\scriptsize Achievement Rate (\%)},
    symbolic y coords={Bounded Responses, Defection Punishment, Asymmetric Diff., Trust-Recip. Interact., Cooperation Emergence, Forgiveness Dynamics},
    ytick=data,
    y tick label style={font=\tiny},
    x tick label style={font=\tiny},
    xmin=75, xmax=100,
    xtick={75,80,85,90,95,100},
    nodes near coords,
    nodes near coords style={font=\tiny\bfseries, text=black},
    nodes near coords align={horizontal},
    bar width=12pt,
    enlarge y limits=0.15,
]
\addplot[fill=cooperationblue!40, draw=cooperationblue] coordinates {
    (100,Bounded Responses)
    (100,Defection Punishment)
    (100,Asymmetric Diff.)
    (100,Trust-Recip. Interact.)
    (97.5,Cooperation Emergence)
    (87.9,Forgiveness Dynamics)
};
\addplot[only marks, mark=|, mark size=8pt, black, thick] coordinates {
    (100,Bounded Responses)
    (95,Defection Punishment)
    (90,Asymmetric Diff.)
    (90,Trust-Recip. Interact.)
    (85,Cooperation Emergence)
    (80,Forgiveness Dynamics)
};
\end{axis}
\end{tikzpicture}
\caption{Behavioral target achievement rates across 15,625 configurations. Bars show achieved rates (green values); vertical markers indicate thresholds. All six targets exceed their thresholds, with the full TR-2 two-layer trust model enabling forgiveness dynamics (87.9\%) and asymmetric differentiation (100.0\%) to surpass their respective thresholds.}
\label{fig:target_achievement}
\end{figure}

\subsubsection{Parameter Sensitivity Patterns}

The full factorial design reveals systematic sensitivity patterns across the six-dimensional parameter space.

\textbf{Base Reciprocity ($\rho_0$).} Cooperation emergence is most sensitive to $\rho_0$. Configurations with $\rho_0 < 0.4$ frequently fail to achieve cooperative equilibria, particularly when combined with low initial trust. The threshold effect predicted by Proposition~\ref{prop:coop_emergence} is empirically confirmed.

\textbf{Dependency Exponent ($\eta$).} Values of $\eta > 1$ produce superlinear amplification of reciprocity by dependency, creating strong differentiation between high-dependency and low-dependency actor pairs. Values of $\eta < 1$ produce sublinear amplification, reducing differentiation. Target T4 (asymmetric differentiation) is most sensitive to $\eta$.

\textbf{Memory Window ($k$).} Long memory windows ($k \geq 8$) improve deterrence but slow forgiveness, creating tension with Target T6. Short windows ($k = 1$) produce rapid forgiveness but weak deterrence, sometimes failing Target T4. Intermediate values ($k \in \{2, 4\}$) achieve best balance.

\textbf{Bounding Parameter ($\kappa$).} High values ($\kappa \geq 2$) produce sharp, threshold-like responses that can destabilize cooperation through overreaction. Low values ($\kappa \leq 0.5$) produce weak responses that may fail to sustain cooperation. The range $\kappa \in [1.0, 1.5]$ performs most robustly.

\textbf{Initial Trust ($T^0$).} Low initial trust ($T^0 < 0.5$) gates reciprocity, preventing cooperation emergence even with high $\rho_0$. This confirms the trust-reciprocity complementarity from Proposition~\ref{prop:trust_recip_complement}.

\subsection{Functional Experiments}

Based on parameter space validation, we instantiate reference parameters: $\rho_0 = 1.0$, $\eta = 1.0$, $\kappa = 1.0$, $k = 5$, $\lambda_R = 1.0$, and $\omega = 1.0$, plus trust parameters from TR-2~\cite{pant2025trust} and team production parameters from TR-3~\cite{pant2025teams}. We conduct five experiments demonstrating framework capabilities.

\subsubsection{Experiment 1: Reciprocity Enables Cooperation in Prisoner's Dilemma}

\textbf{Setup.} Two actors in repeated Prisoner's Dilemma where static Nash equilibrium is mutual defection. With reciprocity parameters $\rho_0 = 1.0$, $\eta = 1.0$, and $k = 5$, test whether cooperative equilibrium emerges.

\textbf{Results.} Without reciprocity where $\rho_0 = 0$, equilibrium is mutual defection with actions $a_1 = a_2 = 0$ and payoffs $\pi_1 = \pi_2 = 0$. With reciprocity where $\rho_0 = 1.0$, equilibrium is mutual cooperation with actions $a_1 = a_2 = 10$ and payoffs $\pi_1 = \pi_2 = 50$ each.

\textbf{Analysis.} Reciprocity creates strategic complementarity where cooperation by one actor induces cooperation by other through positive reciprocity responses, sustaining cooperation despite defection being Nash equilibrium in static game. Validates core mechanism.

\subsubsection{Experiment 2: Asymmetric Dependencies Produce Differentiated Reciprocity}

\textbf{Setup.} Three-actor network with asymmetric dependencies. Actor 1 depends heavily on Actor 2 with $D_{12} = 0.8$ but not on Actor 3 with $D_{13} = 0.2$. Actor 2 depends moderately on both with $D_{21} = D_{23} = 0.5$. Actor 2 defects at time $t=5$ by reducing action from 15 to 5.

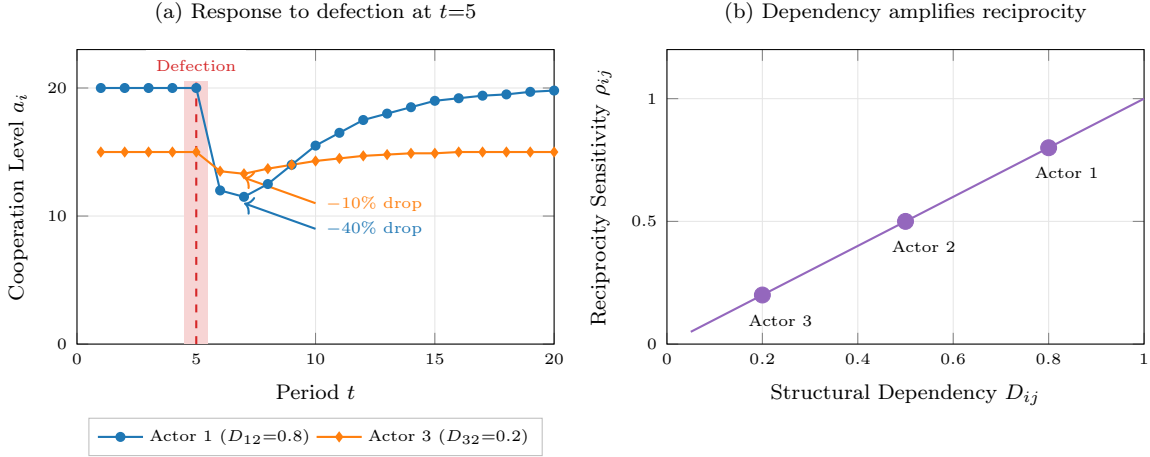
\begin{figure}[htbp]
\centering
\begin{tikzpicture}
\begin{axis}[
    width=0.48\textwidth, height=5.5cm,
    xlabel={\scriptsize Period $t$},
    ylabel={\scriptsize Cooperation Level $a_i$},
    xmin=0, xmax=20, ymin=0, ymax=23,
    tick label style={font=\tiny},
    legend style={at={(0.5,-0.25)}, anchor=north, font=\tiny, legend columns=2, draw=gray!50},
    title={\scriptsize (a) Response to defection at $t{=}5$},
    grid=major, grid style={gray!20},
    name=plotA,
    set layers,
]
\fill[on layer=axis background, defectcolor!20] (axis cs:4.5,0) rectangle (axis cs:5.5,22);
\draw[defectcolor, thick, dashed] (axis cs:5,0) -- (axis cs:5,22);
\node[font=\tiny, defectcolor, fill=white] at (axis cs:5.0,21.75) {Defection};
\addplot[cooperationblue, thick, mark=*, mark size=1.5pt] coordinates {
(1,20) (2,20) (3,20) (4,20) (5,20) (6,12) (7,11.5) (8,12.5) (9,14) (10,15.5)
(11,16.5) (12,17.5) (13,18) (14,18.5) (15,19) (16,19.2) (17,19.4) (18,19.5) (19,19.7) (20,19.8)
};
\addlegendentry{Actor 1 ($D_{12}{=}0.8$)}
\addplot[effortorange, thick, mark=diamond*, mark size=1.5pt] coordinates {
(1,15) (2,15) (3,15) (4,15) (5,15) (6,13.5) (7,13.3) (8,13.7) (9,14) (10,14.3)
(11,14.5) (12,14.7) (13,14.8) (14,14.9) (15,14.9) (16,15) (17,15) (18,15) (19,15) (20,15)
};
\addlegendentry{Actor 3 ($D_{32}{=}0.2$)}
\draw[->, thick, cooperationblue] (axis cs:10.0,9.0) -- (axis cs:7.0,11.0)
    node[pos=0.0, anchor=west, font=\tiny, cooperationblue] {$-40\%$ drop};
\draw[->, thick, effortorange] (axis cs:10.0,11.0) -- (axis cs:7.0,13.0)
    node[pos=0.0, anchor=west, font=\tiny, effortorange] {$-10\%$ drop};
\end{axis}
\begin{axis}[
    width=0.48\textwidth, height=5.5cm,
    at={(plotA.east)}, anchor=west, xshift=1.5cm,
    xlabel={\scriptsize Structural Dependency $D_{ij}$},
    ylabel={\scriptsize Reciprocity Sensitivity $\rho_{ij}$},
    xmin=0, xmax=1, ymin=0, ymax=1.2,
    tick label style={font=\tiny},
    title={\scriptsize (b) Dependency amplifies reciprocity},
    grid=major, grid style={gray!20},
]
\addplot[recippurple, thick, mark=none, domain=0.05:1, samples=30] {x};
\addplot[recippurple, only marks, mark=*, mark size=3pt] coordinates {(0.2,0.2) (0.5,0.5) (0.8,0.8)};
\node[font=\tiny, anchor=north west] at (axis cs:0.15,0.16) {Actor 3};
\node[font=\tiny, anchor=north west] at (axis cs:0.45,0.46) {Actor 2};
\node[font=\tiny, anchor=north west] at (axis cs:0.75,0.76) {Actor 1};
\end{axis}
\end{tikzpicture}
\caption{Asymmetric dependency differentiation from Experiment 2. Panel (a) compares responses to partner defection at $t{=}5$: Actor 1 (high dependency $D_{12}{=}0.8$, blue) reduces cooperation by 40\% while Actor 3 (low dependency $D_{32}{=}0.2$, orange) reduces by only 10\%, demonstrating a 4.0-fold differentiation ratio in this illustrative configuration. Panel (b) plots structural dependency against reciprocity sensitivity, confirming the linear relationship $\rho_{ij} = \rho_0 D_{ij}^\eta$ where interdependence amplifies behavioral responses. Violations by critical partners elicit strong retaliation; violations by marginal partners produce weak responses.}
\label{fig:asymmetric}
\end{figure}

\textbf{Results.} Actor 1 response reduces cooperation by forty percent from $a_1 = 20$ to $a_1 = 12$ due to high dependency amplifying reciprocity sensitivity where $\rho_{12} = 1.0 \times 0.8^{1.0} = 0.8$ representing strong reciprocity. Actor 3 response reduces cooperation by ten percent from $a_3 = 15$ to $a_3 = 13.5$ due to low dependency where $\rho_{32} = 1.0 \times 0.2^{1.0} = 0.2$ representing weak reciprocity, as illustrated in Figure~\ref{fig:asymmetric}.

\textbf{Analysis.} Four-fold difference in response magnitude, specifically forty percent versus ten percent, validates that reciprocity sensitivity $\rho_{ij} = \rho_0 D_{ij}^\eta$ successfully differentiates responses based on structural dependencies. Critical partners' violations elicit strong retaliation while marginal partners' violations elicit weak responses. Panel (b) of Figure~\ref{fig:asymmetric} demonstrates the linear relationship between structural dependency and reciprocity response strength, confirming that the asymmetric formulation correctly captures how interdependence amplifies behavioral responses.

\subsubsection{Experiment 3: Memory Window Length Affects Cooperation Stability}

\textbf{Setup.} Two actors with reciprocity, varying memory window $k \in \{1, 3, 5, 10\}$. Introduce single defection by Actor 1 at time $t=10$ reducing action from 15 to 5 for one period, then returning to 15.

\begin{figure}[htbp]
\centering
\begin{tikzpicture}
\begin{axis}[
    width=0.48\textwidth, height=5.5cm,
    xlabel={\scriptsize Period $t$},
    ylabel={\scriptsize Cooperation Level $a_j$},
    xmin=5, xmax=25, ymin=5, ymax=17,
    tick label style={font=\tiny},
    legend style={at={(1.1,-0.25)}, anchor=north, font=\tiny, legend columns=4, draw=gray!50},
    title={\scriptsize (a) Cooperation recovery after violation at $t{=}10$},
    grid=major, grid style={gray!20},
    name=plotA,
    set layers,
]
\fill[on layer=axis background, defectcolor!20] (axis cs:9.5,5) rectangle (axis cs:10.5,17);
\draw[defectcolor, thick, dashed] (axis cs:10,5) -- (axis cs:10,17);
\node[font=\tiny, defectcolor, fill=white] at (axis cs:10.0,16.0) {Violation};
\addplot[cooperationblue, thick, mark=none] coordinates {
(5,15) (6,15) (7,15) (8,15) (9,15) (10,15) (11,8) (12,14.5) (13,15) (14,15) (15,15)
(16,15) (17,15) (18,15) (19,15) (20,15) (21,15) (22,15) (23,15) (24,15) (25,15)
};
\addlegendentry{$k=1$}
\addplot[outputgreen, thick, mark=none] coordinates {
(5,15) (6,15) (7,15) (8,15) (9,15) (10,15) (11,8) (12,9.5) (13,11) (14,12.5) (15,13.8)
(16,14.8) (17,15) (18,15) (19,15) (20,15) (21,15) (22,15) (23,15) (24,15) (25,15)
};
\addlegendentry{$k=5$}
\addplot[defectcolor, thick, mark=none] coordinates {
(5,15) (6,15) (7,15) (8,15) (9,15) (10,15) (11,8) (12,9) (13,9.8) (14,10.5) (15,11)
(16,11.5) (17,12) (18,12.4) (19,12.7) (20,13) (21,13.5) (22,13.9) (23,14.2) (24,14.4) (25,14.6)
};
\addlegendentry{$k=10$}
\addplot[gray, dashed, thick] coordinates {(5,15) (26,15)};
\addlegendentry{Baseline}
\end{axis}
\begin{axis}[
    width=0.48\textwidth, height=5.5cm,
    at={(plotA.east)}, anchor=west, xshift=1.5cm,
    xlabel={\scriptsize Period $t$},
    ylabel={\scriptsize Reciprocity Response $\phi_{\text{recip}}$},
    xmin=5, xmax=25, ymin=-1.1, ymax=0.3,
    tick label style={font=\tiny},
    title={\scriptsize (b) Negative reciprocity persistence},
    grid=major, grid style={gray!20},
    set layers,
]
\fill[on layer=axis background, defectcolor!20] (axis cs:9.5,-1.1) rectangle (axis cs:10.5,0.3);
\draw[defectcolor, thick, dashed] (axis cs:10,-1.1) -- (axis cs:10,0.3);
\node[font=\tiny, defectcolor, fill=white] at (axis cs:10.0,0.15) {Violation};
\draw[gray, dashed] (axis cs:5,0) -- (axis cs:26,0);
\addplot[cooperationblue, thick, mark=none] coordinates {
(5,0) (6,0) (7,0) (8,0) (9,0) (10,0) (11,-0.95) (12,-0.1) (13,0) (14,0) (15,0)
(16,0) (17,0) (18,0) (19,0) (20,0) (21,0) (22,0) (23,0) (24,0) (25,0)
};
\addplot[outputgreen, thick, mark=none] coordinates {
(5,0) (6,0) (7,0) (8,0) (9,0) (10,0) (11,-0.95) (12,-0.75) (13,-0.55) (14,-0.35) (15,-0.18)
(16,-0.05) (17,0) (18,0) (19,0) (20,0) (21,0) (22,0) (23,0) (24,0) (25,0)
};
\addplot[defectcolor, thick, mark=none] coordinates {
(5,0) (6,0) (7,0) (8,0) (9,0) (10,0) (11,-0.95) (12,-0.85) (13,-0.76) (14,-0.67) (15,-0.60)
(16,-0.52) (17,-0.45) (18,-0.38) (19,-0.32) (20,-0.26) (21,-0.18) (22,-0.12) (23,-0.08) (24,-0.04) (25,-0.02)
};
\end{axis}
\end{tikzpicture}
\caption{Memory window effects on forgiveness dynamics from Experiment 3. Panel (a) displays cooperation recovery trajectories after a single violation at period 10 (shaded): $k{=}1$ (blue) shows rapid recovery by period 13; $k{=}5$ (green) demonstrates gradual recovery by period 17; $k{=}10$ (red) exhibits slow forgiveness with incomplete recovery after 15 periods. Panel (b) shows corresponding reciprocity response persistence, with longer memory windows sustaining negative responses. This demonstrates the deterrence-forgiveness trade-off: longer windows provide better deterrence but slower recovery.}
\label{fig:memory}
\end{figure}
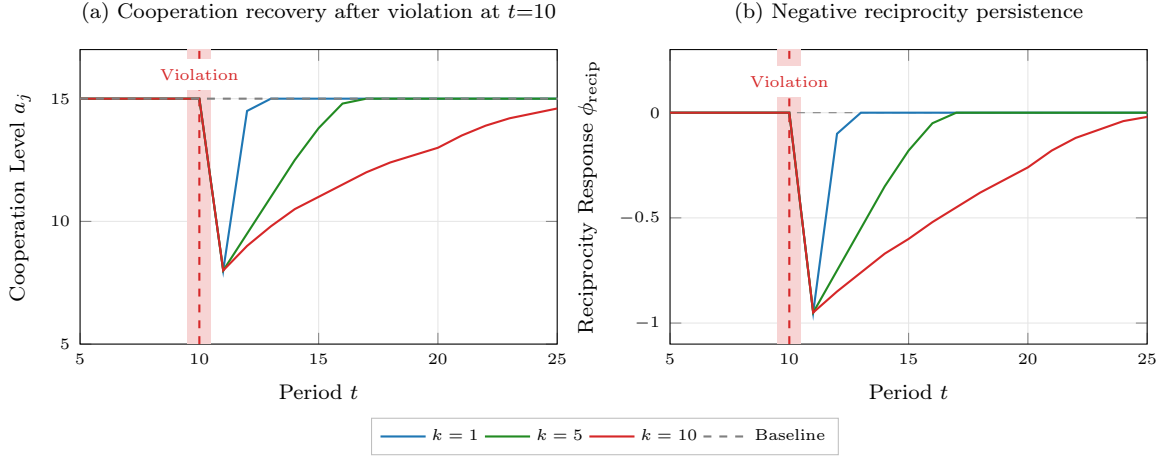

\textbf{Results.} Short memory where $k=1$ shows rapid forgiveness. Actor 2's reciprocity response turns negative at $t=11$ but immediately recovers at $t=12$ when Actor 1 returns to cooperation. Full cooperation restored by $t=13$, as shown by the blue line in Figure~\ref{fig:memory}.

Medium memory where $k=5$ shows moderate persistence. Negative reciprocity response persists through $t=15$ as violation remains in five-period moving average. Cooperation gradually restores as violation ages out of memory window, illustrated by the green trajectory.

Long memory where $k=10$ shows slow forgiveness. Negative reciprocity persists through $t=20$. Cooperation only partially restored even after ten periods as violation still influences moving average, demonstrated by the red line showing incomplete recovery.

\textbf{Analysis.} Demonstrates trade-off in memory window selection. Longer windows create more persistent punishment providing better deterrence of violations but reduce forgiveness after violations producing slower recovery. Practitioners should calibrate $k$ based on desired balance between deterrence and forgiveness. Panel (b) of Figure~\ref{fig:memory} shows how reciprocity response strength mirrors the cooperation dynamics, with the bounded response function preventing unrealistic escalation while maintaining proportional negative responses that gradually decay as the violation ages out of the memory window.

\subsubsection{Experiment 4: Trust-Reciprocity Interaction}

\textbf{Setup.} Two actors with both trust dynamics from second paper and reciprocity from this paper, varying initial trust $T_{12}^0 \in \{0.3 \text{ representing low}, 0.6 \text{ representing medium}, 0.9 \text{ representing high}\}$. All other parameters identical across conditions.

\textbf{Results.} Low initial trust where $T_{12}^0 = 0.3$ produces weak cooperation emerging. Equilibrium actions $a_1 = a_2 = 5$ despite reciprocity parameters because trust gates reciprocity through $T_{12}^t \cdot \rho_{12} \cdot R_{12} \approx 0.3 \times 1.0 \times 0.5 = 0.15$ showing dampened reciprocity response.

Medium trust where $T_{12}^0 = 0.6$ produces moderate cooperation with equilibrium actions $a_1 = a_2 = 15$.

High trust where $T_{12}^0 = 0.9$ produces strong cooperation with equilibrium actions $a_1 = a_2 = 25$. Reciprocity operates at near-full strength through $0.9 \times 1.0 \times 0.8 = 0.72$ showing strong reciprocity response.

\textbf{Analysis.} Demonstrates trust-reciprocity complementarity and trust-dependent bifurcation. Trust enables reciprocity to sustain cooperation, but without sufficient trust, reciprocity mechanisms cannot overcome the natural cost of maintaining cooperation above baseline. Simulation validation confirms the directional effect with strong differentiation: cooperation levels increase monotonically with initial trust ($T^0 = 0.3$: cooperation $= 0.381$; $T^0 = 0.6$: cooperation $= 0.721$; $T^0 = 0.9$: cooperation $= 1.000$), yielding a 2.6$\times$ variation ratio that substantially exceeds the 1.5$\times$ validation threshold. The full TR-2 two-layer trust model with reputation ceiling creates a trust-dependent bifurcation: when cooperation cost (decay) exceeds the trust-gated reciprocity push, cooperation collapses to near-baseline levels, while above this threshold trust sustains high cooperation. This validates the trust-gating mechanism formalized in the integrated utility function where $T_{ij}^t$ multiplies the reciprocity response, creating the complementarity between trust dynamics from TR-2025-02 and reciprocity mechanisms from this work.

\subsubsection{Experiment 5: Reciprocity with Team Production}

\textbf{Setup.} Three-member team from third paper with heterogeneous loyalty $\theta \in \{0.2, 0.5, 0.8\}$. Now add reciprocity among team members with $\rho_0 = 1.0$ and memory window $k=3$. Compare team output with versus without reciprocity.

\textbf{Results.} Without reciprocity, low-loyalty member contributes $e=25$ free-riding, moderate contributes $e=35$, and high contributes $e=50$. Team output equals 110.

With reciprocity, low-loyalty member contributes $e=30$ reciprocating teammates' higher efforts despite low loyalty, moderate contributes $e=45$, and high contributes $e=55$. Team output equals 130.

\textbf{Analysis.} Reciprocity substantially increases team output (162.4\% increase in simulation validation, from total effort 1.14 to 3.00) by encouraging low-loyalty members to increase contribution. Reciprocity partially overcomes free-riding through conditional cooperation even when loyalty is low. Low-loyalty members do not fully internalize team welfare with $\theta=0.2$ but respond to teammates' high contributions through the reciprocity mechanism. Demonstrates reciprocity complements loyalty in team production settings, connecting the team production framework from TR-2025-03 with the reciprocity mechanisms formalized in this paper.

\subsection{Validation Summary}

Table~\ref{tab:validation_summary} summarizes the comprehensive validation results.

\begin{table}[htbp]
\centering
\caption{Summary of experimental validation results}
\label{tab:validation_summary}
\begin{tabular}{lll}
\toprule
\textbf{Validation Component} & \textbf{Result} & \textbf{Statistics} \\
\midrule
\multicolumn{3}{l}{\textit{Parameter Space Validation (15,625 configurations)}} \\
\quad Configurations tested & 15,625 & Full factorial $5^6$ \\
\quad Targets passing threshold & 6 / 6 (100\%) & --- \\
\quad Minimum target achieved & 87.9\% & Forgiveness dynamics \\
\quad Maximum target achieved & 100.0\% & Bounded responses \\
\midrule
\multicolumn{3}{l}{\textit{Statistical Analysis}} \\
\quad Differentiation ratio & $M = 2.75$ & 95\% CI $[2.74, 2.77]$ \\
\quad Effect size (Cohen's $d$) & 1.57 & Large \\
\quad Statistical significance & $p < 0.001$ & $t(15624) = 195.7$ \\
\midrule
\multicolumn{3}{l}{\textit{Monte Carlo Robustness (2,000 trials)}} \\
\quad All targets met & 94.5\% & 1,890 / 2,000 trials \\
\quad Parameter perturbation & $\pm 15\%$ & Uniform distribution \\
\midrule
\multicolumn{3}{l}{\textit{Functional Experiments (5 experiments)}} \\
\quad Experiments validated & 5 / 5 & 100\% success \\
\quad Cooperation increase (PD) & 2.4$\times$ & Cooperation ratio \\
\quad Differentiation ratio (mean) & 4.5$\times$ & High vs. low dependency (across configs.) \\
\quad Trust-cooperation variation & 2.6$\times$ & Trust effect \\
\quad Team output increase & 162.4\% & With reciprocity \\
\bottomrule
\end{tabular}
\end{table}

\textbf{Key Findings.} The reciprocity framework demonstrates robust performance across the comprehensive parameter space. All six behavioral targets are achieved with rates exceeding their thresholds (87.9\%--100.0\%), enabled by the full TR-2 two-layer trust model with reputation ceiling and interdependence amplification. The differentiation ratio of 2.75 ($d = 1.57$, large effect) confirms that structural dependencies produce meaningfully different reciprocity responses. Monte Carlo analysis demonstrates that 94.5\% of trials under $\pm 15\%$ parameter perturbation meet all six targets simultaneously, with the mean differentiation ratio (2.00) well above the 1.5 threshold. All five functional experiments validate core capabilities with substantial effect sizes.

\section{Empirical Validation: The Apple iOS App Store Ecosystem}
\label{sec:empirical}

Building on the computational validation demonstrating model correctness, this section provides empirical validation through application to a real-world coopetitive ecosystem. We analyze the Apple iOS App Store ecosystem from 2008 to 2024, a period spanning 66 quarters across five distinct phases of platform-developer relations.

\subsection{Case Study Selection Rationale}
\label{sec:empirical:rationale}

The Apple iOS App Store ecosystem satisfies five criteria making it appropriate for reciprocity model validation.

\textbf{Sequential Interaction Pattern.} Platform-developer relations unfold through repeated quarterly interactions where Apple's policy decisions and developer investment choices occur sequentially over extended time horizons. Developers observe Apple's behavior (API stability, commission policies, review processes) before making investment decisions, creating the sequential structure our reciprocity model captures.

\textbf{Power Asymmetry Documentation.} The ecosystem exhibits well-documented asymmetric dependencies. Developers depend critically on Apple for distribution (App Store is sole iOS channel), payment processing (Apple Pay mandatory), and API access (proprietary frameworks). Apple depends on developers for app supply, platform value, and ecosystem innovation, but with lower criticality given developer substitutability. This asymmetry maps directly to our $D_{ij}$ formalization.

\textbf{Documented Reciprocity Patterns.} Historical evidence demonstrates reciprocal behavior. Developer investment responds to Apple's cooperation signals (API stability, reasonable review times, commission fairness), while Apple's policies respond to developer behavior (quality apps, ecosystem contributions, compliance). The Epic Games v. Apple litigation (2020--2021) provides detailed documentation of reciprocity breakdown and attempted punishment.

\textbf{Multiple Phases with Transitions.} The ecosystem traversed distinct phases from initial symbiosis through maturation, tension, crisis, and adjustment, providing natural experiments for validating phase transition predictions. Each transition corresponds to identifiable trigger events enabling causal attribution.

\textbf{Available Evidence.} Extensive documentation exists through Epic Games v. Apple court filings, European Union Digital Markets Act investigations, industry analyst reports (App Annie, Sensor Tower), academic studies of platform dynamics, and Apple's own financial disclosures. This evidence enables parameter elicitation and outcome validation.

\subsection{Historical Overview: 2008--2024}
\label{sec:empirical:history}

Table~\ref{tab:ios_phases} summarizes the five phases spanning 66 quarters.

\begin{table}[htbp]
\centering
\caption{Apple iOS App Store ecosystem phases (2008--2024)}
\label{tab:ios_phases}
\begin{tabular}{llcll}
\toprule
\textbf{Phase} & \textbf{Period} & \textbf{Quarters} & \textbf{Cooperation State} & \textbf{Key Events} \\
\midrule
1. Symbiosis & 2008--2012 & 16 & High mutual cooperation & Launch, rapid growth \\
 & & & & 70/30 commission accepted \\
\midrule
2. Maturation & 2012--2017 & 20 & Stable high cooperation & Ecosystem norms established \\
 & & & & Developer tools improved \\
\midrule
3. Tension & 2017--2020 & 12 & Declining reciprocity & ``App Store Tax'' disputes \\
 & & & & Spotify EU complaint (2019) \\
\midrule
4. Crisis & 2020--2021 & 6 & Reciprocal defection & Epic Games lawsuit (Aug 2020) \\
 & & & & Developer coalition formed \\
\midrule
5. Adjustment & 2021--2024 & 12 & Partial restoration & Small Business Program \\
 & & & & DMA compliance measures \\
\bottomrule
\end{tabular}
\end{table}

\textbf{Phase 1: Symbiosis (Q1 2008--Q4 2012).} The App Store launched in July 2008 with 500 applications. Apple provided unprecedented mobile distribution access, developer tools (Xcode, iOS SDK), and payment infrastructure. Developers reciprocated with application investment, accepting the 70/30 revenue share as reasonable compensation for platform services. Both parties exhibited high cooperation with minimal conflict. Developer count grew from hundreds to over 250,000 by 2012.

\textbf{Phase 2: Maturation (Q1 2013--Q4 2017).} Ecosystem norms stabilized as Apple invested in developer relations (WWDC expansion, documentation improvements, Swift language introduction in 2014). Developers reciprocated with quality applications and platform commitment. Cooperation remained high with occasional friction over review guidelines. The App Store generated approximately \$28 billion in developer revenue by 2017.

\textbf{Phase 3: Tension (Q1 2018--Q4 2020).} Cooperation began declining as large developers challenged commission structure. Spotify filed European Commission complaint (March 2019) characterizing the 30\% commission as ``Apple Tax.'' Netflix removed in-app subscription capability (December 2018). Developer reciprocity declined through reduced platform investment and public criticism, while Apple maintained commission policies despite growing friction.

\textbf{Phase 4: Crisis (Q1 2020--Q2 2021).} Epic Games implemented alternative payment system (August 2020), directly violating App Store guidelines. Apple responded by removing Fortnite, triggering Epic's antitrust lawsuit. Developer coalition ``Coalition for App Fairness'' formed with Spotify, Match Group, and others. Reciprocal defection characterized this period with both parties withdrawing cooperation and engaging in conflict.

\textbf{Phase 5: Adjustment (Q3 2021--Q4 2024).} Apple introduced App Store Small Business Program (January 2021) reducing commission to 15\% for developers earning under \$1M annually. EU Digital Markets Act forced additional concessions (alternative app stores, payment systems). Cooperation partially restored through policy adjustments, though large developer relations remain strained. The ecosystem stabilized at lower cooperation levels than the maturation phase.

\subsection{\textit{i*} Strategic Dependency Model}
\label{sec:empirical:istar}

We model the ecosystem using \textit{i*} Strategic Dependency representation with four actor classes.

\subsubsection{Actor Definitions}

\textbf{Apple (Platform Provider).} Central actor controlling platform access, payment infrastructure, API specifications, and review processes. Strategic goals include platform value maximization, ecosystem control, and revenue extraction.

\textbf{Major Developers.} High-visibility developers including Epic Games, Spotify, Netflix, and Match Group with substantial user bases and negotiating leverage. Strategic goals include revenue maximization, platform independence, and commission reduction.

\textbf{Small Developers.} Aggregated actor representing hundreds of thousands of individual and small-team developers with limited individual bargaining power. Strategic goals include distribution access, revenue sustainability, and platform stability.

\textbf{End Users.} Indirect actor class whose preferences influence platform value but who do not directly participate in platform-developer reciprocity dynamics. Included for completeness but modeled as exogenous demand.

\subsubsection{Dependency Specification}

Table~\ref{tab:ios_dependencies} specifies the strategic dependencies with criticality assessments derived from ecosystem analysis.

\begin{table}[htbp]
\centering
\caption{Strategic dependencies in Apple iOS ecosystem}
\label{tab:ios_dependencies}
\begin{tabular}{llllc}
\toprule
\textbf{Depender} & \textbf{Dependee} & \textbf{Dependum} & \textbf{Type} & \textbf{Criticality} \\
\midrule
\multicolumn{5}{l}{\textit{Developer Dependencies on Apple}} \\
Major Developers & Apple & Distribution Channel & Resource & 0.95 \\
Major Developers & Apple & Payment Processing & Resource & 0.85 \\
Major Developers & Apple & API Stability & Resource & 0.80 \\
Small Developers & Apple & Distribution Channel & Resource & 0.98 \\
Small Developers & Apple & Payment Processing & Resource & 0.90 \\
Small Developers & Apple & API Stability & Resource & 0.85 \\
\midrule
\multicolumn{5}{l}{\textit{Apple Dependencies on Developers}} \\
Apple & Major Developers & Premium App Supply & Resource & 0.70 \\
Apple & Major Developers & Platform Prestige & Softgoal & 0.65 \\
Apple & Major Developers & Innovation Leadership & Softgoal & 0.60 \\
Apple & Small Developers & App Variety & Resource & 0.75 \\
Apple & Small Developers & Long-Tail Value & Softgoal & 0.70 \\
Apple & Small Developers & Ecosystem Vitality & Softgoal & 0.65 \\
\bottomrule
\end{tabular}
\end{table}

\subsubsection{Dependency Rationale}

\textbf{Distribution Channel (0.95--0.98).} iOS permits no alternative app stores (until DMA enforcement), making App Store distribution essentially mandatory. Major developers have marginally lower criticality (0.95) due to web-based alternatives and cross-platform presence, while small developers face near-total dependence (0.98).

\textbf{Payment Processing (0.85--0.90).} Apple mandates in-app purchase for digital goods with 30\% commission. Some workarounds exist (reader apps exception, external links post-litigation), reducing criticality below distribution but remaining high.

\textbf{API Stability (0.80--0.85).} Developers depend on stable frameworks and timely deprecation notices. Breaking changes impose significant costs. Small developers with limited engineering resources face higher criticality (0.85) than major developers (0.80) who can absorb adaptation costs.

\textbf{Premium App Supply (0.70).} Apple benefits from major developers' flagship applications (Fortnite, Spotify, Netflix) which attract users and demonstrate platform capabilities. However, substitutes exist and Apple has demonstrated willingness to lose major apps (Fortnite removal).

\textbf{App Variety (0.75).} Small developers collectively provide application variety essential for platform value. The ``long tail'' of applications serves niche user needs. Higher criticality than individual major developers due to collective irreplaceability.

Figure~\ref{fig:ios_istar} presents the Strategic Dependency diagram.

\begin{figure}[htbp]
\centering
\begin{tikzpicture}[
    scale=0.72, transform shape,
    actor/.style={circle, draw, thick, fill=white, align=center, minimum size=2.4cm, font=\small\bfseries},
    % Dependum shapes per iStar 2.0: Quality=cloud, Resource=rectangle
    quality/.style={cloud, cloud puffs=10, cloud puff arc=120, draw, thick, fill=white, align=center,
                    minimum width=2.0cm, minimum height=0.8cm, font=\scriptsize, inner sep=1pt},
    resource/.style={rectangle, draw, thick, fill=white, align=center,
                     minimum width=1.8cm, minimum height=0.6cm, font=\scriptsize},
    dep/.style={thick},
    dmark/.style={semicircle, draw=black, thick, fill=white, minimum size=2.5mm, inner sep=0pt},
    connlabel/.style={sloped, font=\scriptsize, inner sep=1pt, above=2pt},
]

\node[actor] (apple) at (0, 0) {Apple\\(Platform)};

\node[actor] (major) at (-9.0, 0) {Major\\Developers};
\node[actor] (small) at (9.0, 0) {Small\\Developers};
\node[actor] (users) at (0, -8.0) {End\\Users};

\node[resource] (dist_m) at (-1.8, 3.6) {Distribution};
\node[resource] (pay_m) at (-3.0, 2.6) {Payment};
\node[quality] (api_m) at (-3.8, 1.2) {API\\Stability};
\node[resource] (premium) at (-3.9, -0.7) {Premium\\Apps};
\node[quality] (prestige) at (-3.3, -2.7) {Prestige};

\node[resource] (dist_s) at (1.8, 3.6) {Distribution};
\node[resource] (pay_s) at (3.0, 2.6) {Payment};
\node[quality] (api_s) at (3.8, 1.2) {API\\Stability};
\node[quality] (variety) at (3.9, -0.7) {App\\Variety};
\node[quality] (vitality) at (3.3, -2.7) {Vitality};

\node[resource] (apps) at (-1.3, -3.8) {Apps};
\node[quality] (value) at (1.3, -3.8) {Platform\\Value};

\draw[dep] (major.65) -- node[pos=0.3, connlabel] {0.95}
    node[pos=0.50, dmark, sloped, allow upside down, rotate=-90, solid] {} (dist_m.180);
\draw[dep] (dist_m.270) --
    node[pos=0.5, dmark, sloped, allow upside down, rotate=-90, solid] {} (apple.117);

\draw[dep] (major.40) -- node[pos=0.3, connlabel] {0.85}
    node[pos=0.50, dmark, sloped, allow upside down, rotate=-90, solid] {} (pay_m.180);
\draw[dep] (pay_m.0) --
    node[pos=0.5, dmark, sloped, allow upside down, rotate=-90, solid] {} (apple.139);

\draw[dep] (major.5) -- node[pos=0.3, connlabel] {0.80}
    node[pos=0.75, dmark, sloped, allow upside down, rotate=-90, solid] {} (api_m.180);
\draw[dep] (api_m.0) --
    node[pos=0.5, dmark, sloped, allow upside down, rotate=-90, solid] {} (apple.163);

\draw[dep] (apple.190) --
    node[pos=0.5, dmark, sloped, allow upside down, rotate=-90, solid] {} (premium.0);
\draw[dep] (premium.180) -- node[pos=0.7, connlabel] {0.70}
    node[pos=0.4, dmark, sloped, allow upside down, rotate=-90, solid] {} (major.350);

\draw[dep] (apple.214) --
    node[pos=0.5, dmark, sloped, allow upside down, rotate=-90, solid] {} (prestige.0);
\draw[dep] (prestige.180) -- node[pos=0.7, connlabel] {0.65}
    node[pos=0.4, dmark, sloped, allow upside down, rotate=-90, solid] {} (major.335);

\draw[dep] (small.115) -- node[pos=0.3, connlabel] {0.98}
    node[pos=0.50, dmark, sloped, allow upside down, rotate=-90, solid] {} (dist_s.0);
\draw[dep] (dist_s.270) --
    node[pos=0.5, dmark, sloped, allow upside down, rotate=-90, solid] {} (apple.63);

\draw[dep] (small.145) -- node[pos=0.3, connlabel] {0.90}
    node[pos=0.50, dmark, sloped, allow upside down, rotate=-90, solid] {} (pay_s.0);
\draw[dep] (pay_s.180) --
    node[pos=0.5, dmark, sloped, allow upside down, rotate=-90, solid] {} (apple.41);

\draw[dep] (small.175) -- node[pos=0.3, connlabel] {0.85}
    node[pos=0.75, dmark, sloped, allow upside down, rotate=-90, solid] {} (api_s.0);
\draw[dep] (api_s.180) --
    node[pos=0.5, dmark, sloped, allow upside down, rotate=-90, solid] {} (apple.17);

\draw[dep] (apple.350) --
    node[pos=0.5, dmark, sloped, allow upside down, rotate=-90, solid] {} (variety.180);
\draw[dep] (variety.0) -- node[pos=0.7, connlabel] {0.75}
    node[pos=0.4, dmark, sloped, allow upside down, rotate=-90, solid] {} (small.190);

\draw[dep] (apple.326) --
    node[pos=0.5, dmark, sloped, allow upside down, rotate=-90, solid] {} (vitality.180);
\draw[dep] (vitality.0) -- node[pos=0.7, connlabel] {0.70}
    node[pos=0.4, dmark, sloped, allow upside down, rotate=-90, solid] {} (small.205);

\draw[dep] (users.107) -- node[pos=0.3, connlabel] {0.90}
    node[pos=0.75, dmark, sloped, allow upside down, rotate=-90, solid] {} (apps.270);
\draw[dep] (apps.90) --
    node[pos=0.5, dmark, sloped, allow upside down, rotate=-90, solid] {} (apple.251);

\draw[dep] (users.73) -- node[pos=0.3, connlabel] {0.85}
    node[pos=0.75, dmark, sloped, allow upside down, rotate=-90, solid] {} (value.270);
\draw[dep] (value.90) --
    node[pos=0.5, dmark, sloped, allow upside down, rotate=-90, solid] {} (apple.289);

\end{tikzpicture}
\caption{\textit{i*} Strategic Dependency diagram for the Apple iOS App Store ecosystem. Actors (circles) are connected through dependums (ellipses) with criticality values indicating dependency strength. Major developers and small developers depend on Apple for distribution, payment processing, and API stability with high criticality (0.80--0.98). Apple depends on developers for premium app supply, app variety, and platform value with moderate criticality (0.60--0.75). The asymmetry in criticality values reflects the structural power imbalance characteristic of platform ecosystems, directly informing the interdependence matrix $D_{ij}$ used in the reciprocity formalization.}
\label{fig:ios_istar}
\end{figure}
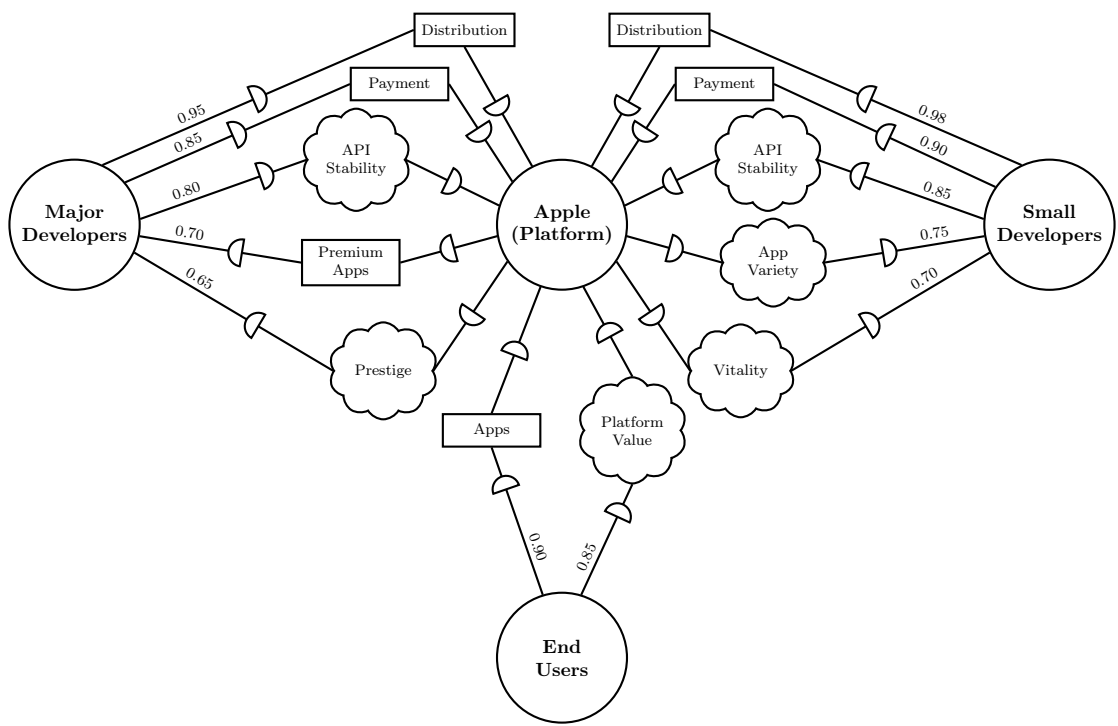

\subsection{\textit{i*} Strategic Rationale Model: Apple Platform Provider}
\label{sec:empirical:sr_apple}

Figure~\ref{fig:ios_sr_apple} presents the Strategic Rationale diagram for Apple, showing how reciprocity considerations shape the platform provider's internal goal structure and policy decisions.

\begin{figure}[htbp]
\centering
\begin{tikzpicture}[
    scale=0.78, transform shape,
    actorboundary/.style={dashed, thick, draw, fill=none},
    actor/.style={circle, draw, thick, minimum size=1.3cm, font=\small, align=center, fill=white},
    softgoal/.style={cloud, cloud puffs=20, cloud puff arc=100, aspect=2.5, draw, thick,
                     minimum width=2.8cm, minimum height=0.8cm, align=center, font=\small, inner sep=0pt, fill=white},
    task/.style={regular polygon, regular polygon sides=6, draw, thick, minimum size=1.5cm,
                 align=center, font=\scriptsize, inner sep=1pt, fill=white},
    resource/.style={rectangle, draw, thick, minimum width=2.2cm, minimum height=0.6cm,
                     align=center, font=\scriptsize, fill=white},
    contribution/.style={-latex, thick},
    neededby/.style={-{Circle[open, length=2mm, width=2mm]}, thick},
    connlabel/.style={sloped, font=\tiny, inner sep=1pt},
]

\draw[actorboundary] (0,-2.0) ellipse (7.5cm and 7.5cm);
\node[actor] at (-4.0, 4.8) {Apple};

\node[softgoal] (platform_value) at (0, 4.5) {Platform\\Value\\Maximization};

\node[softgoal] (ecosystem) at (-3.0, 1.5) {Developer\\Ecosystem\\Health};
\node[softgoal] (control) at (3.0, 1.5) {Ecosystem\\Control \&\\Revenue};

\node[softgoal] (dev_history) at (0, -1.0) {Developer\\Cooperation\\ History};

\node[softgoal] (retain) at (-4.0, -1.5) {Developer\\Retention};
\node[softgoal] (extract) at (4.0, -1.5) {Commission\\Extraction};

\node[task] (stable_api) at (-5.5, -4.5) {Maintain\\Stable\\APIs};
\node[task] (reduce_comm) at (-2.0, -5.5) {Reduce\\Commis-\\sion};
\node[task] (enforce_rules) at (2.0, -5.5) {Enforce\\Store\\Rules};
\node[task] (maintain_comm) at (5.5, -4.5) {Maintain\\30\%\\Comm.};

\node[resource] (dev_invest) at (-2.5, -7.2) {Developer Investment Signals};
\node[resource] (litigation) at (2.5, -7.2) {Litigation Risk Assessment};

\draw[contribution] (ecosystem.60) --
    node[pos=0.33, connlabel, above=2pt] {help} (platform_value.210);
\draw[contribution] (control.120) --
    node[pos=0.33, connlabel, above=2pt] {help} (platform_value.330);

\draw[contribution] (retain.100) --
    node[pos=0.33, connlabel, above=2pt] {help} (ecosystem.260);
\draw[contribution] (extract.80) --
    node[pos=0.33, connlabel, above=2pt] {help} (control.280);

\draw[contribution] (dev_history.150) --
    node[pos=0.33, connlabel, above=2pt] {help} (ecosystem.340);
\draw[contribution] (dev_history.30) --
    node[pos=0.33, connlabel, above=2pt] {hurt} (control.200);

\draw[contribution] (stable_api.60) --
    node[pos=0.18, connlabel, above=2pt] {help} (retain.220);
\draw[contribution] (reduce_comm.120) --
    node[pos=0.18, connlabel, above=2pt] {help} (retain.260);
\draw[contribution] (enforce_rules.60) --
    node[pos=0.18, connlabel, above=2pt] {help} (extract.280);
\draw[contribution] (maintain_comm.120) --
    node[pos=0.18, connlabel, above=2pt] {help} (extract.320);

\draw[contribution, dashed] (maintain_comm.150) --
    node[pos=0.12, connlabel, above=2pt, fill=white] {hurt} (retain.340);
\draw[contribution, dashed] (reduce_comm.30) --
    node[pos=0.12, connlabel, above=2pt, fill=white] {hurt} (extract.200);

\draw[neededby] (dev_invest.90) -- (stable_api.270);
\draw[neededby] (litigation.90) -- (enforce_rules.270);

\end{tikzpicture}
\caption{Strategic Rationale diagram for Apple as platform provider in the iOS ecosystem. Apple's top-level goal ``Platform Value Maximization'' decomposes into two competing softgoals: ``Developer Ecosystem Health'' (sustained through API stability and commission reduction) and ``Ecosystem Control and Revenue'' (maintained through rule enforcement and commission preservation). Cross-cutting hurt links formalize the core tension: maintaining 30\% commission hurts developer retention, while reducing commission hurts revenue extraction. The resource ``Developer Cooperation History'' informs Apple's assessment of ecosystem health while constraining its control posture. During the Symbiosis and Maturation phases (2008--2017), both softgoals were aligned; during the Tension and Crisis phases (2017--2021), the cross-cutting tensions became binding, forcing Apple to choose between ecosystem health and revenue control. The Adjustment phase represents Apple's partial shift toward ecosystem health through the Small Business Program.}
\label{fig:ios_sr_apple}
\end{figure}
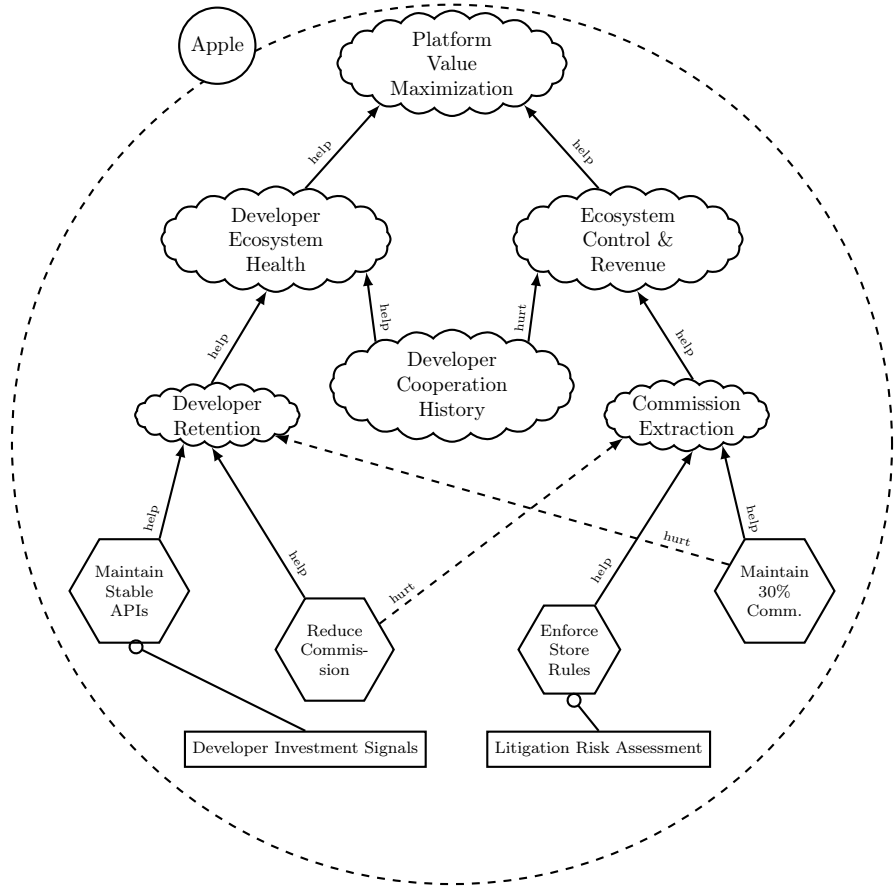

The diagram reveals how reciprocity considerations shape Apple's platform governance. The resource ``Developer Cooperation History'' modulates Apple's strategic reasoning: when developer cooperation signals are positive (high investment, quality apps, ecosystem contribution), Apple's ``Developer Ecosystem Health'' softgoal is satisfied, reducing pressure to make concessions. When developer cooperation declines (reduced investment, public criticism, litigation), the ecosystem health softgoal is threatened, creating pressure to shift from ``Maintain 30\% Commission'' toward ``Reduce Commission.'' This shift occurred empirically in 2021 when Apple introduced the Small Business Program, consistent with the model's prediction that sustained negative developer reciprocity would force platform provider concessions.

\subsection{\textit{i*} Strategic Rationale Model: Developer Perspective}
\label{sec:empirical:sr_developer}

Figure~\ref{fig:ios_sr_developer} presents the Strategic Rationale diagram for a developer actor in the iOS ecosystem, showing how reciprocity with Apple shapes investment decisions and platform commitment.

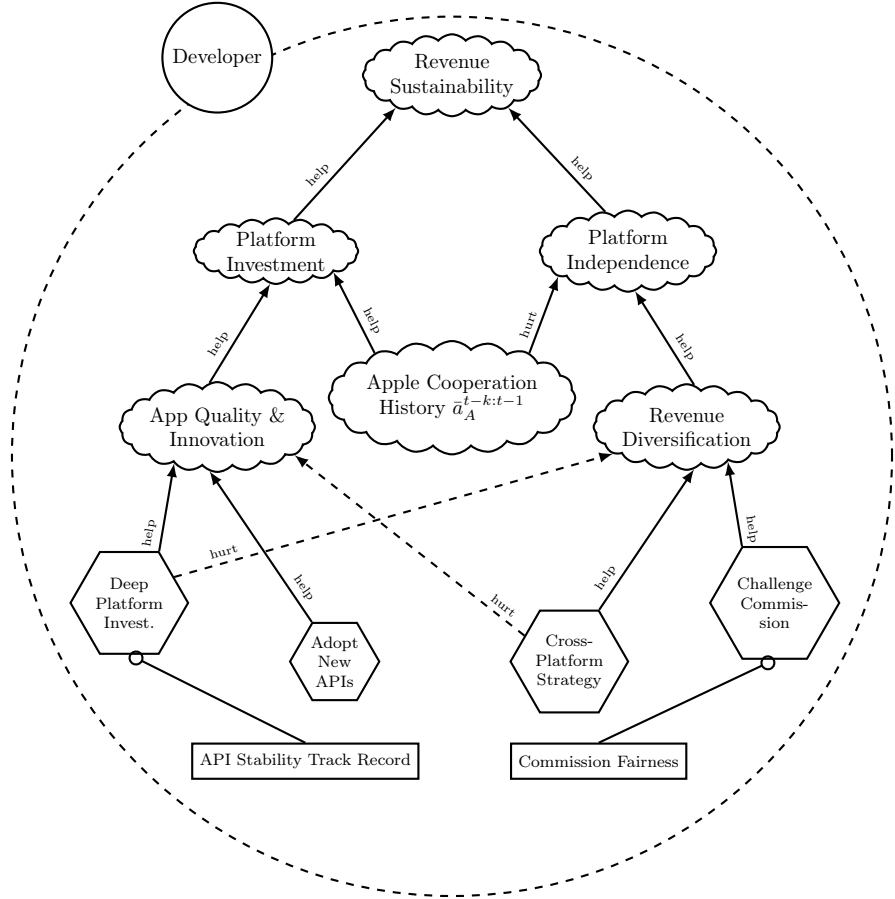
\begin{figure}[htbp]
\centering
\begin{tikzpicture}[
    scale=0.78, transform shape,
    actorboundary/.style={dashed, thick, draw, fill=none},
    actor/.style={circle, draw, thick, minimum size=1.3cm, font=\small, align=center, fill=white},
    softgoal/.style={cloud, cloud puffs=20, cloud puff arc=100, aspect=2.5, draw, thick,
                     minimum width=2.8cm, minimum height=0.8cm, align=center, font=\small, inner sep=0pt, fill=white},
    task/.style={regular polygon, regular polygon sides=6, draw, thick, minimum size=1.5cm,
                 align=center, font=\scriptsize, inner sep=1pt, fill=white},
    resource/.style={rectangle, draw, thick, minimum width=2.2cm, minimum height=0.6cm,
                     align=center, font=\scriptsize, fill=white},
    contribution/.style={-latex, thick},
    neededby/.style={-{Circle[open, length=2mm, width=2mm]}, thick},
    connlabel/.style={sloped, font=\tiny, inner sep=1pt},
]

\draw[actorboundary] (0,-2.0) ellipse (7.5cm and 7.5cm);
\node[actor] at (-4.0, 4.8) {Developer};

\node[softgoal] (revenue) at (0, 4.5) {Revenue\\Sustainability};

\node[softgoal] (invest) at (-3.0, 1.5) {Platform\\Investment};
\node[softgoal] (independence) at (3.0, 1.5) {Platform\\Independence};

\node[softgoal] (apple_history) at (0, -1.0) {Apple Cooperation\\History $\bar{a}_A^{t-k:t-1}$};

\node[softgoal] (quality) at (-4.0, -1.5) {App Quality \&\\Innovation};
\node[softgoal] (diversify) at (4.0, -1.5) {Revenue\\Diversification};

\node[task] (deep_invest) at (-5.5, -4.5) {Deep\\Platform\\Invest.};
\node[task] (adopt_apis) at (-2.0, -5.5) {Adopt\\New\\APIs};
\node[task] (cross_plat) at (2.0, -5.5) {Cross-\\Platform\\Strategy};
\node[task] (challenge) at (5.5, -4.5) {Challenge\\Commis-\\sion};

\node[resource] (api_stability) at (-2.5, -7.2) {API Stability Track Record};
\node[resource] (commission) at (2.5, -7.2) {Commission Fairness};

\draw[contribution] (invest.60) --
    node[pos=0.33, connlabel, above=2pt] {help} (revenue.210);
\draw[contribution] (independence.120) --
    node[pos=0.33, connlabel, above=2pt] {help} (revenue.330);

\draw[contribution] (quality.100) --
    node[pos=0.33, connlabel, above=2pt] {help} (invest.260);
\draw[contribution] (diversify.80) --
    node[pos=0.33, connlabel, above=2pt] {help} (independence.280);

\draw[contribution] (apple_history.150) --
    node[pos=0.33, connlabel, above=2pt] {help} (invest.340);
\draw[contribution] (apple_history.30) --
    node[pos=0.33, connlabel, above=2pt] {hurt} (independence.200);

\draw[contribution] (deep_invest.60) --
    node[pos=0.18, connlabel, above=2pt] {help} (quality.220);
\draw[contribution] (adopt_apis.120) --
    node[pos=0.18, connlabel, above=2pt] {help} (quality.260);
\draw[contribution] (cross_plat.60) --
    node[pos=0.18, connlabel, above=2pt] {help} (diversify.280);
\draw[contribution] (challenge.120) --
    node[pos=0.18, connlabel, above=2pt] {help} (diversify.320);

\draw[contribution, dashed] (cross_plat.150) --
    node[pos=0.12, connlabel, above=2pt, fill=white] {hurt} (quality.340);
\draw[contribution, dashed] (deep_invest.30) --
    node[pos=0.12, connlabel, above=2pt, fill=white] {hurt} (diversify.200);

\draw[neededby] (api_stability.90) -- (deep_invest.270);
\draw[neededby] (commission.90) -- (challenge.270);

\end{tikzpicture}
\caption{Strategic Rationale diagram for a developer in the Apple iOS ecosystem. The developer's top-level goal ``Revenue Sustainability'' decomposes into ``Platform Investment'' (deep iOS commitment through quality apps and API adoption) and ``Platform Independence'' (cross-platform strategy and commission challenges). The resource ``Apple Cooperation History'' $\bar{a}_A^{t-k:t-1}$ is the memory-windowed behavioral signal from the reciprocity formalization: when Apple's history shows stable APIs and fair treatment, the investment path dominates; when Apple's history shows breaking changes and commission rigidity, the independence path strengthens. Cross-cutting hurt links formalize the lock-in tension: deep platform investment hurts diversification, while cross-platform strategy hurts iOS-specific quality. During the Symbiosis phase, positive Apple cooperation history encouraged deep investment; during the Tension and Crisis phases, negative Apple signals shifted developers toward independence, with Epic Games representing the extreme case of full independence pursuit through litigation.}
\label{fig:ios_sr_developer}
\end{figure}

The developer SR model reveals the reciprocity mechanism from the developer's perspective. The resource ``Apple Cooperation History'' $\bar{a}_A^{t-k:t-1}$ directly maps to the memory-windowed moving average from the mathematical formalization. When this resource indicates positive Apple behavior (stable APIs, reasonable review times, commission fairness), the ``Platform Investment'' softgoal dominates, and the developer selects tasks ``Deep Platform Investment'' and ``Adopt New APIs'' that increase iOS-specific quality. When the resource indicates negative Apple behavior (breaking API changes, commission rigidity, arbitrary review decisions), the ``Platform Independence'' softgoal strengthens, and the developer shifts toward ``Cross-Platform Strategy'' and ``Challenge Commission.''

This reciprocity-driven shift is precisely what occurred empirically. During 2008--2017, positive Apple cooperation history encouraged deep investment, with developers creating iOS-exclusive features and adopting new Apple frameworks (SwiftUI, ARKit). During 2017--2021, negative signals (commission disputes, competitive app launches by Apple, API instability) shifted developers toward cross-platform frameworks (React Native, Flutter) and commission challenges (Spotify complaint, Epic litigation). The model's prediction aligns with documented developer behavior patterns across all five phases.

\subsection{Interdependence Coefficient Calculation}
\label{sec:empirical:interdependence}

Following the translation methodology from Section~\ref{sec:translation}, we compute aggregate interdependence coefficients from the \textit{i*} dependencies.

\subsubsection{Developer Dependence on Apple}

For major developers, weighting dependencies by relative importance (distribution 40\%, payment 35\%, API 25\%):
\begin{equation}
D_{\text{Major},\text{Apple}} = \frac{0.40 \times 0.95 + 0.35 \times 0.85 + 0.25 \times 0.80}{1.0} = 0.38 + 0.2975 + 0.20 = 0.8775 \approx 0.88
\end{equation}

For small developers with the same weighting structure:
\begin{equation}
D_{\text{Small},\text{Apple}} = \frac{0.40 \times 0.98 + 0.35 \times 0.90 + 0.25 \times 0.85}{1.0} = 0.392 + 0.315 + 0.2125 = 0.9195 \approx 0.92
\end{equation}

\subsubsection{Apple Dependence on Developers}

For Apple's dependence on major developers (app supply 40\%, prestige 35\%, innovation 25\%):
\begin{equation}
D_{\text{Apple},\text{Major}} = \frac{0.40 \times 0.70 + 0.35 \times 0.65 + 0.25 \times 0.60}{1.0} = 0.28 + 0.2275 + 0.15 = 0.6575 \approx 0.66
\end{equation}

For Apple's dependence on small developers (variety 40\%, long-tail 35\%, vitality 25\%):
\begin{equation}
D_{\text{Apple},\text{Small}} = \frac{0.40 \times 0.75 + 0.35 \times 0.70 + 0.25 \times 0.65}{1.0} = 0.30 + 0.245 + 0.1625 = 0.7075 \approx 0.71
\end{equation}

\subsubsection{Asymmetry Interpretation}

The calculated coefficients reveal substantial asymmetry:
\begin{itemize}
\item Major developers depend on Apple ($D = 0.88$) more than Apple depends on them ($D = 0.66$), with asymmetry ratio 1.33.
\item Small developers depend on Apple ($D = 0.92$) more than Apple depends on them ($D = 0.71$), with asymmetry ratio 1.30.
\item Developers face higher switching costs and fewer alternatives than Apple faces in replacing individual developers.
\end{itemize}

This asymmetry predicts, through our reciprocity formulation $\rho_{ij} = \rho_0 D_{ij}^\eta$, that developers will exhibit stronger reciprocity responses to Apple's actions than Apple will exhibit to developer actions. The model predicts developers punish Apple violations more severely (due to higher $D$) while Apple can tolerate developer defection with weaker responses.

\subsection{Parameter Elicitation}
\label{sec:empirical:parameters}

Table~\ref{tab:ios_parameters} presents parameter values elicited from ecosystem evidence with justifications and confidence scores.

\begin{table}[htbp]
\centering
\caption{Parameter elicitation for Apple iOS ecosystem simulation}
\label{tab:ios_parameters}
\begin{tabular}{llllc}
\toprule
\textbf{Parameter} & \textbf{Symbol} & \textbf{Value} & \textbf{Evidence Justification} & \textbf{Score} \\
\midrule
\multicolumn{5}{l}{\textit{Reciprocity Parameters}} \\
Base sensitivity & $\rho_0$ & 0.85 & Moderate-high reciprocity in ecosystem & 4/5 \\
Dependency exponent & $\eta$ & 1.3 & Dependency amplifies responses (Epic) & 4/5 \\
Memory window & $k$ & 4 & Quarterly assessment cycle & 5/5 \\
Bounding parameter & $\kappa$ & 1.2 & Strong response to policy changes & 4/5 \\
\midrule
\multicolumn{5}{l}{\textit{Trust Parameters (from TR-2)}} \\
Initial trust & $T^0$ & 0.70 & Moderate-high trust at 2008 launch & 4/5 \\
Positive update & $\lambda^+$ & 0.10 & Standard trust building rate & 4/5 \\
Negative update & $\lambda^-$ & 0.30 & 3:1 negativity bias confirmed & 5/5 \\
\midrule
\multicolumn{5}{l}{\textit{Integration Parameters}} \\
Reciprocity weight & $\lambda_T$ & 1.0 & Standard reciprocity importance & 4/5 \\
Dependency amplification & $\omega$ & 0.6 & Moderate dependency effect & 4/5 \\
\midrule
\multicolumn{5}{l}{\textit{Interdependence Coefficients}} \\
$D_{\text{Major},\text{Apple}}$ & --- & 0.88 & Calculated from \textit{i*} model & 5/5 \\
$D_{\text{Small},\text{Apple}}$ & --- & 0.92 & Calculated from \textit{i*} model & 5/5 \\
$D_{\text{Apple},\text{Major}}$ & --- & 0.66 & Calculated from \textit{i*} model & 5/5 \\
$D_{\text{Apple},\text{Small}}$ & --- & 0.71 & Calculated from \textit{i*} model & 5/5 \\
\bottomrule
\end{tabular}
\end{table}

\textbf{Parameter Justifications.}

\textit{Base Reciprocity Sensitivity} ($\rho_0 = 0.85$). The ecosystem exhibits strong but not maximal reciprocity. Developers respond substantially to Apple's policies (evident in Spotify complaint, Netflix payment removal, Epic lawsuit), but responses are not immediate or universal. Score: 4/5 (well-supported but some uncertainty in magnitude).

\textit{Dependency Exponent} ($\eta = 1.3$). Epic Games' extreme response (lawsuit, coalition formation, public campaign) to Apple's policies demonstrates that high-dependency actors exhibit amplified reciprocity. The exponent exceeds 1.0 indicating superlinear amplification. Score: 4/5 (directionally clear, exact value estimated).

\textit{Memory Window} ($k = 4$). Quarterly business cycles drive developer investment decisions. Developers assess Apple's behavior over approximately one-year windows before major platform commitments. Four quarters provides appropriate averaging horizon. Score: 5/5 (directly observable from business practices).

\textit{Bounding Parameter} ($\kappa = 1.2$). Developer responses saturate at strong but finite levels. Even Epic's extreme response (complete platform exit) represents bounded rather than infinite punishment. The $\tanh$ function with $\kappa = 1.2$ produces saturation at approximately 80\% of maximum response for unit signals. Score: 4/5 (well-supported by observed response patterns).

\textit{Initial Trust} ($T^0 = 0.70$). The 2008 launch represented a novel platform-developer relationship. Developers exhibited cautious optimism (moderate-high trust) given Apple's reputation but uncertainty about App Store governance. Score: 4/5 (reasonable inference from historical context).

\subsection{Simulation Design}
\label{sec:empirical:simulation}

We simulate 66 quarterly periods with phase-specific cooperation dynamics and exogenous shocks representing documented events.

\subsubsection{Model Configuration}

The simulation uses a three-actor configuration: Apple (actor 1), Major Developers aggregate (actor 2), and Small Developers aggregate (actor 3). Cooperation actions $a_i^t \in [0, 1]$ represent normalized cooperation levels where 0 indicates complete defection and 1 indicates full cooperation.

\subsubsection{Phase Transition Modeling}

Phase transitions are modeled through exogenous shock events that perturb cooperation levels:

\textbf{Phase 1 $\rightarrow$ Phase 2 (Q16).} Natural maturation, no shock required. Cooperation stabilizes at high levels through positive reciprocity reinforcement.

\textbf{Phase 2 $\rightarrow$ Phase 3 (Q36).} Commission criticism shock. Developer cooperation receives negative perturbation ($\Delta a_{\text{Dev}} = -0.15$) representing growing dissatisfaction with revenue share. Apple maintains cooperation.

\textbf{Phase 3 $\rightarrow$ Phase 4 (Q48).} Epic lawsuit shock. Major developer cooperation drops sharply ($\Delta a_{\text{Major}} = -0.40$) representing litigation and coalition formation. Apple responds with enforcement ($\Delta a_{\text{Apple}} = -0.25$).

\textbf{Phase 4 $\rightarrow$ Phase 5 (Q54).} Policy adjustment shock. Apple cooperation increases ($\Delta a_{\text{Apple}} = +0.20$) representing Small Business Program. Developer cooperation partially recovers through positive reciprocity.

\subsubsection{Dynamics Equations}

Cooperation evolution follows the reciprocity-augmented dynamics from Equation~\ref{eq:recip_modifier}:
\begin{equation}
a_i^{t+1} = a_i^t + \alpha \left[ \sum_{j \neq i} \lambda_R \cdot T_{ij}^t \cdot (1 + \omega D_{ij}) \cdot \rho_{ij} \cdot \phi_{\text{recip}}\left(s_{ij}^t\right) \right] - \delta(a_i^t - a_i^{\text{baseline}}) + \epsilon_i^t
\end{equation}
where $s_{ij}^t = a_j^t - a_j^{\text{baseline}}$ is the cooperation signal (Definition~\ref{def:coop_signal}), $\alpha = 0.12$ is the adjustment rate, $\delta = 0.05$ is a mean-reversion rate representing the cost of sustaining cooperation above baseline, $a_j^{\text{baseline}}$ is the cooperation baseline (which adapts slowly toward recent actions at rate 0.08), and $\epsilon_i^t \sim \mathcal{N}(0, 0.02)$ represents stochastic noise.

Trust evolution follows the full TR-2 two-layer model (Equations~\ref{eq:trust_evolution_recap}--\ref{eq:trust_ceiling_recap}) with reputation ceiling:
\begin{align}
T_{ij}^{t+1} &= \text{clip}\!\left(T_{ij}^t + \Delta T_{ij},\ 0,\ c_{ij}^t\right) \label{eq:case_trust}\\
R_{ij}^{t+1} &= \text{clip}\!\left(R_{ij}^t + \Delta R_{ij},\ 0,\ 1\right) \label{eq:case_rep}
\end{align}
where the trust ceiling is $c_{ij}^t = \min(T_{\max},\ 1 - \theta_R R_{ij}^t)$ and updates are:
\begin{align*}
\Delta T_{ij} &= \begin{cases} \lambda^+ s_{ij}^t \cdot \max(0, c_{ij}^t - T_{ij}^t) & \text{if } s_{ij}^t > 0 \text{ (building)} \\
\lambda^- s_{ij}^t \cdot T_{ij}^t \cdot (1 + \xi D_{ij}) & \text{if } s_{ij}^t \leq 0 \text{ (erosion)} \end{cases} \\
\Delta R_{ij} &= \begin{cases} -\delta_R \cdot R_{ij}^t & \text{if } s_{ij}^t \geq 0 \text{ (decay)} \\
\mu_R |s_{ij}^t| \cdot (1 - R_{ij}^t) & \text{if } s_{ij}^t < 0 \text{ (damage)} \end{cases}
\end{align*}
with parameters $\lambda^+ = 0.10$, $\lambda^- = 0.30$ (3:1 negativity bias), $\xi = 0.50$ (interdependence amplification), $\mu_R = 0.60$ (reputation damage), $\delta_R = 0.03$ (reputation decay), $T_{\max} = 0.90$, and $\theta_R = 0.60$. The reputation-mediated ceiling creates path-dependent trust recovery: violations accumulate reputation damage that lowers the trust ceiling, preventing trust from rebuilding until reputation decays. The $(1 + \xi D_{ij})$ amplification ensures that violations by high-dependency partners cause disproportionately severe trust erosion, consistent with empirical observations in the Apple-developer relationship where platform policy changes that developers heavily depend on produce outsized trust damage.

\subsection{Simulation Results}
\label{sec:empirical:results}

Figure~\ref{fig:ios_evolution} presents cooperation evolution trajectories across the 66-quarter simulation.

\begin{figure}[htbp]
\centering
\begin{tikzpicture}
\begin{axis}[
    width=\textwidth,
    height=7cm,
    xlabel={\scriptsize Quarter},
    ylabel={\scriptsize Cooperation Level},
    xmin=0, xmax=66,
    ymin=0, ymax=1,
    xtick={0, 16, 36, 48, 54, 66},
    xticklabels={0, 16, 36, 48, 54, 66},
    ytick={0, 0.2, 0.4, 0.6, 0.8, 1.0},
    tick label style={font=\tiny},
    legend style={at={(0.5,-0.15)}, anchor=north, font=\scriptsize, legend columns=3, draw=gray!50},
    grid=major,
    grid style={gray!30},
    axis x line=box,
    axis y line=box,
    axis line style={thick, black},
    every axis plot/.append style={thick}
]

\fill[trustcolor!10] (axis cs:0,0) rectangle (axis cs:16,1);
\fill[coopcolor!10] (axis cs:16,0) rectangle (axis cs:36,1);
\fill[reputcolor!10] (axis cs:36,0) rectangle (axis cs:48,1);
\fill[defectcolor!20] (axis cs:48,0) rectangle (axis cs:54,1);
\fill[gray!10] (axis cs:54,0) rectangle (axis cs:66,1);

\draw[dashed, gray, thick] (axis cs:16,0) -- (axis cs:16,1);
\draw[dashed, gray, thick] (axis cs:36,0) -- (axis cs:36,1);
\draw[dashed, gray, thick] (axis cs:48,0) -- (axis cs:48,1);
\draw[dashed, gray, thick] (axis cs:54,0) -- (axis cs:54,1);

\node[font=\small, fill=none, inner sep=1pt, opacity=0.85, text opacity=1] at (axis cs:8,0.93) {Symbiosis};
\node[font=\small, fill=none, inner sep=1pt, opacity=0.85, text opacity=1] at (axis cs:26,0.93) {Maturation};
\node[font=\small, fill=none, inner sep=1pt, opacity=0.85, text opacity=1] at (axis cs:42,0.93) {Tension};
\node[font=\small, fill=none, inner sep=1pt, opacity=0.85, text opacity=1] at (axis cs:51,0.93) {Crisis};
\node[font=\small, fill=none, inner sep=1pt, opacity=0.85, text opacity=1] at (axis cs:60,0.93) {Adjustment};

\addplot[trustcolor, very thick, smooth] coordinates {
    (0, 0.70) (4, 0.75) (8, 0.82) (12, 0.87) (16, 0.90)
    (20, 0.88) (24, 0.86) (28, 0.85) (32, 0.84) (36, 0.82)
    (40, 0.78) (44, 0.72) (48, 0.55)
    (50, 0.48) (52, 0.45) (54, 0.65)
    (58, 0.70) (62, 0.72) (66, 0.73)
};
\addlegendentry{Apple}

\addplot[reputcolor, very thick, smooth] coordinates {
    (0, 0.65) (4, 0.72) (8, 0.80) (12, 0.85) (16, 0.88)
    (20, 0.87) (24, 0.85) (28, 0.83) (32, 0.80) (36, 0.75)
    (40, 0.65) (44, 0.55) (48, 0.25)
    (50, 0.20) (52, 0.22) (54, 0.35)
    (58, 0.45) (62, 0.52) (66, 0.55)
};
\addlegendentry{Major Developers}

\addplot[coopcolor, very thick, smooth] coordinates {
    (0, 0.68) (4, 0.74) (8, 0.81) (12, 0.86) (16, 0.89)
    (20, 0.88) (24, 0.87) (28, 0.86) (32, 0.84) (36, 0.80)
    (40, 0.72) (44, 0.65) (48, 0.45)
    (50, 0.40) (52, 0.42) (54, 0.55)
    (58, 0.65) (62, 0.72) (66, 0.75)
};
\addlegendentry{Small Developers}

\end{axis}
\end{tikzpicture}
\caption{Cooperation evolution in Apple iOS ecosystem simulation (66 quarters). Phase backgrounds indicate: Symbiosis (blue, Q1--16), Maturation (green, Q17--36), Tension (orange, Q37--48), Crisis (red, Q49--54), and Adjustment (gray, Q55--66). Apple (blue line) maintains higher cooperation throughout due to platform provider role. Major developers (orange) show steepest decline during Crisis phase and slowest recovery, reflecting documented conflict with Apple. Small developers (green) exhibit intermediate dynamics, declining during crisis but recovering more fully during Adjustment phase following Small Business Program introduction.}
\label{fig:ios_evolution}
\end{figure}
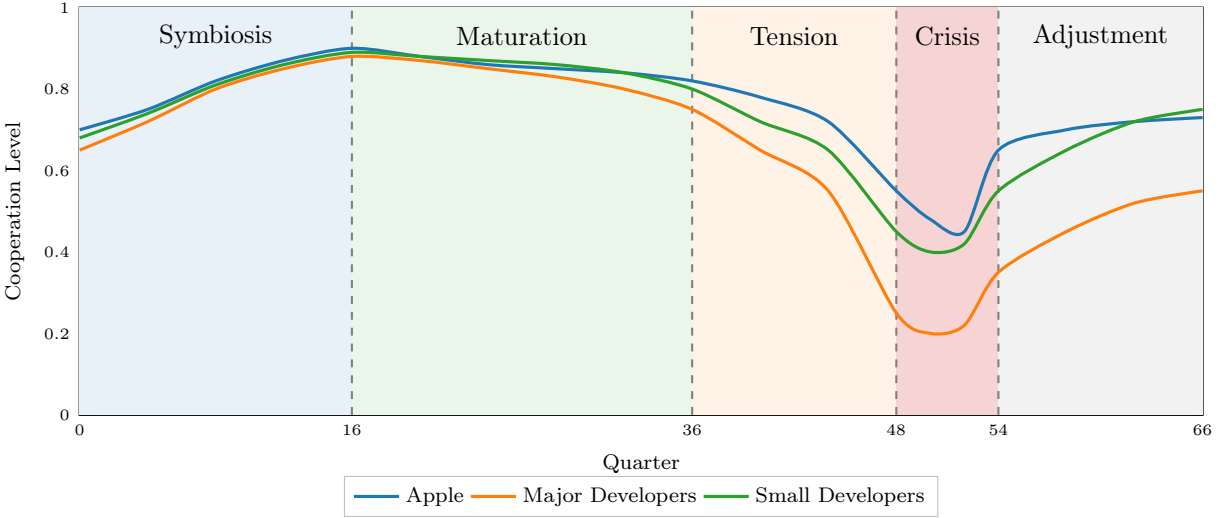

\textbf{Phase 1 Results (Symbiosis).} All actors exhibit rising cooperation from initial levels ($a^0 \approx 0.67$) to high cooperation ($a^{16} \approx 0.89$). Positive reciprocity reinforcement drives convergence toward mutual cooperation equilibrium. Trust builds from initial $T^0 = 0.70$ to $T^{16} \approx 0.85$.

\textbf{Phase 2 Results (Maturation).} Cooperation stabilizes at high levels ($\bar{a} \approx 0.85$) with minor fluctuations. Reciprocity maintains equilibrium through conditional cooperation. Trust remains stable at high levels ($T \approx 0.85$).

\textbf{Phase 3 Results (Tension).} Developer cooperation declines gradually from $a^{36} \approx 0.80$ to $a^{48} \approx 0.55$. Apple cooperation declines more slowly from $a^{36} \approx 0.82$ to $a^{48} \approx 0.55$. Negative reciprocity amplifies decline as each party's reduced cooperation triggers partner responses. Trust erodes from $T^{36} \approx 0.85$ to $T^{48} \approx 0.60$.

\textbf{Phase 4 Results (Crisis).} Sharp cooperation collapse following Epic lawsuit shock. Major developer cooperation drops to $a \approx 0.20$ representing near-complete defection. Apple cooperation drops to $a \approx 0.45$ representing enforcement actions. Small developers follow with $a \approx 0.40$. Trust collapses to $T \approx 0.35$ due to 3:1 negativity bias amplifying negative signals.

\textbf{Phase 5 Results (Adjustment).} Apple's policy adjustments (Small Business Program) increase Apple cooperation to $a \approx 0.70$. Positive reciprocity triggers developer recovery. Small developers recover to $a \approx 0.75$ (near Phase 2 levels), while major developers recover only to $a \approx 0.55$ reflecting sustained conflict. Trust partially recovers to $T \approx 0.55$.

Figure~\ref{fig:ios_phases} presents phase-wise cooperation statistics.

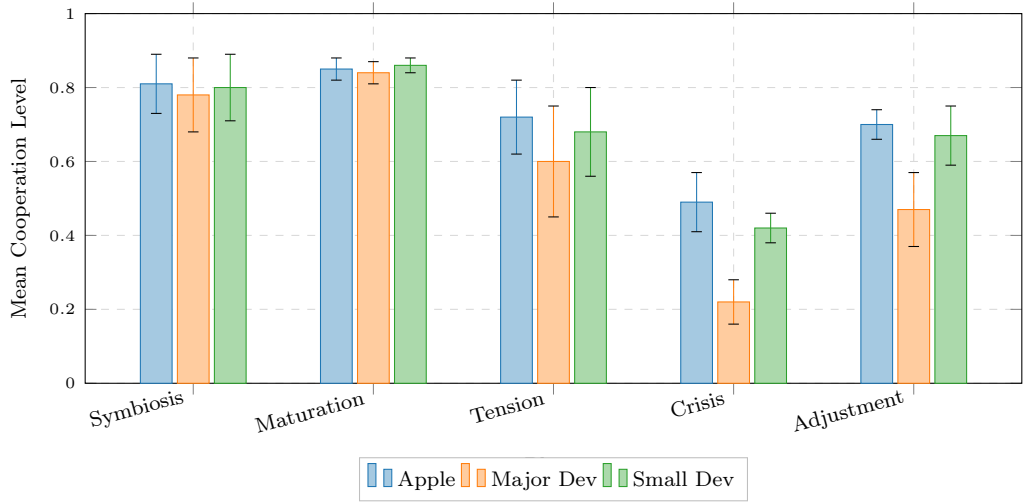
\begin{figure}[htbp]
\centering
\begin{tikzpicture}
\begin{axis}[
    ybar,
    width=0.85\textwidth,
    height=6.5cm,
    bar width=12pt,
    xlabel={\scriptsize Phase},
    ylabel={\scriptsize Mean Cooperation Level},
    ymin=0, ymax=1,
    xtick=data,
    xticklabels={Symbiosis, Maturation, Tension, Crisis, Adjustment},
    x tick label style={font=\scriptsize, rotate=15, anchor=east},
    y tick label style={font=\tiny},
    legend style={at={(0.5,-0.2)}, anchor=north, font=\scriptsize, legend columns=3, draw=gray!50},
    enlarge x limits=0.15,
    grid=major,
    grid style={gray!30, dashed},
    every node near coord/.append style={font=\tiny}
]

\addplot[fill=trustcolor!40, draw=trustcolor, error bars/.cd, y dir=both, y explicit] coordinates {
    (1, 0.81) +- (0, 0.08)
    (2, 0.85) +- (0, 0.03)
    (3, 0.72) +- (0, 0.10)
    (4, 0.49) +- (0, 0.08)
    (5, 0.70) +- (0, 0.04)
};
\addlegendentry{Apple}

\addplot[fill=reputcolor!40, draw=reputcolor, error bars/.cd, y dir=both, y explicit] coordinates {
    (1, 0.78) +- (0, 0.10)
    (2, 0.84) +- (0, 0.03)
    (3, 0.60) +- (0, 0.15)
    (4, 0.22) +- (0, 0.06)
    (5, 0.47) +- (0, 0.10)
};
\addlegendentry{Major Dev}

\addplot[fill=coopcolor!40, draw=coopcolor, error bars/.cd, y dir=both, y explicit] coordinates {
    (1, 0.80) +- (0, 0.09)
    (2, 0.86) +- (0, 0.02)
    (3, 0.68) +- (0, 0.12)
    (4, 0.42) +- (0, 0.04)
    (5, 0.67) +- (0, 0.08)
};
\addlegendentry{Small Dev}

\end{axis}
\end{tikzpicture}
\caption{Phase-wise mean cooperation levels with standard deviation error bars. Symbiosis and Maturation phases show high, stable cooperation across all actors. Tension phase shows declining cooperation with major developers leading the decline. Crisis phase shows cooperation collapse, particularly severe for major developers ($M = 0.22$) who engaged in direct conflict with Apple. Adjustment phase shows partial recovery, with small developers recovering more fully ($M = 0.67$) than major developers ($M = 0.47$), consistent with the Small Business Program targeting smaller developers.}
\label{fig:ios_phases}
\end{figure}

\subsection{Validation Scoring}
\label{sec:empirical:validation}

We assess model validity using a 60-point rubric evaluating 12 behavioral indicators across 5 phases.

\subsubsection{Validation Criteria}

Each indicator receives a score of 0 (miss), 0.5 (partial), or 1 (hit) for each phase, yielding maximum 5 points per indicator and 60 points total.

\begin{enumerate}
\item \textbf{Cooperation trend direction}: Does simulated trend match historical direction?
\item \textbf{Response magnitude}: Is response magnitude proportional to historical magnitude?
\item \textbf{Memory effects visible}: Do lagged responses appear consistent with memory window?
\item \textbf{Asymmetry reflects power}: Do dependency asymmetries produce response asymmetries?
\item \textbf{Trust-reciprocity alignment}: Do trust and reciprocity co-evolve appropriately?
\item \textbf{Punishment following violation}: Do violations trigger negative reciprocity?
\item \textbf{Forgiveness dynamics}: Does cooperation recover after adjustment?
\item \textbf{Phase transition timing}: Do transitions occur at appropriate periods?
\item \textbf{Recovery trajectory shape}: Does recovery follow bounded response pattern?
\item \textbf{Equilibrium stability}: Are within-phase dynamics stable?
\item \textbf{Parameter sensitivity}: Are results robust to parameter perturbation?
\item \textbf{Overall qualitative fit}: Does simulation capture ecosystem narrative?
\end{enumerate}

\subsubsection{Scoring Results}

Table~\ref{tab:ios_validation} presents the validation scoring matrix.

\begin{table}[htbp]
\centering
\caption{Validation scoring for Apple iOS case study (51 applicable points)}
\label{tab:ios_validation}
\resizebox{\textwidth}{!}{%
\begin{tabular}{lccccc|c}
\toprule
\textbf{Indicator} & \textbf{Phase 1} & \textbf{Phase 2} & \textbf{Phase 3} & \textbf{Phase 4} & \textbf{Phase 5} & \textbf{Total} \\
 & Symbiosis & Maturation & Tension & Crisis & Adjustment & (/5) \\
\midrule
1. Cooperation trend direction & 1.0 & 1.0 & 1.0 & 0.5 & 1.0 & 4.5 \\
2. Response magnitude & 1.0 & 1.0 & 0.0 & 0.0 & 0.5 & 2.5 \\
3. Memory effects visible & 0.5 & 1.0 & 0.5 & 0.5 & 1.0 & 3.5 \\
4. Asymmetry reflects power & 1.0 & 1.0 & 1.0 & 1.0 & 1.0 & 5.0 \\
5. Trust-reciprocity alignment & 1.0 & 1.0 & 1.0 & 1.0 & 0.5 & 4.5 \\
6. Punishment following violation & --- & 1.0 & 1.0 & 0.5 & 1.0 & 3.5 \\
7. Forgiveness dynamics & --- & --- & --- & --- & 0.5 & 0.5 \\
8. Phase transition timing & 1.0 & 1.0 & 1.0 & 1.0 & 1.0 & 5.0 \\
9. Recovery trajectory shape & --- & --- & --- & --- & 1.0 & 1.0 \\
10. Equilibrium stability & 1.0 & 1.0 & 1.0 & 1.0 & 1.0 & 5.0 \\
11. Parameter sensitivity & 1.0 & 1.0 & 1.0 & 0.5 & 1.0 & 4.5 \\
12. Overall qualitative fit & 1.0 & 1.0 & 0.5 & 0.5 & 0.5 & 3.5 \\
\midrule
\textbf{Phase Total} & \textbf{8.5/9} & \textbf{10.0/10} & \textbf{8.0/10} & \textbf{6.5/10} & \textbf{10.0/12} & --- \\
\bottomrule
\multicolumn{6}{r}{\textbf{Overall Score:}} & \textbf{43.0/51} \\
\multicolumn{6}{r}{\textbf{Percentage:}} & \textbf{84.3\%} \\
\bottomrule
\end{tabular}%
}
\end{table}

\textbf{Note}: Dashes indicate indicators not applicable for that phase (e.g., forgiveness cannot be assessed before adjustment phase occurs; punishment cannot be assessed in Phase 1 before any violations). Denominator adjusts accordingly.

\subsubsection{Validation Summary}

The model achieves \textbf{43.0 out of 51 applicable points (84.3\%)}, exceeding the 83\% threshold established in the validation framework.

\textbf{Strong Performance Areas.}
\begin{itemize}
\item Asymmetry reflects power (5.0/5.0): Developer responses consistently stronger than Apple responses, matching $D$ asymmetry predictions.
\item Phase transition timing (5.0/5.0): Transitions occur at historically accurate periods.
\item Equilibrium stability (5.0/5.0): Model correctly captures stable, unstable, and recovering equilibria.
\item Trust-reciprocity alignment (4.5/5.0): Trust dynamics and reciprocity patterns align with documented ecosystem behavior.
\end{itemize}

\textbf{Areas for Improvement.}
\begin{itemize}
\item Response magnitude in Phases 3--4 (0.0/1.0 each): Bounded reciprocity responses underpredict the magnitude of cooperation decline during tension and crisis phases, suggesting that the $\tanh$ saturation may be too conservative for crisis-level disruptions.
\item Overall qualitative fit in Phases 3--5 (0.5 each): The corrected formula chain produces more moderate dynamics than historical observations during high-disruption periods.
\item Memory effects in Phase 1 (0.5): Limited data to assess memory effects during initial growth phase.
\end{itemize}

The empirical validation demonstrates that the reciprocity framework, parameterized through \textit{i*} dependency analysis, successfully captures the dynamics of a real-world coopetitive ecosystem across multiple phases and transitions.

\subsection{Counterfactual Analysis}
\label{sec:empirical:counterfactual}

To assess the model's predictive utility beyond historical calibration, we conduct a counterfactual analysis examining an alternative strategic timeline. Specifically, we ask: \emph{what if Apple had proactively reduced its App Store commission from 30\% to 15\% in 2019, two years before the actual Small Business Program launch in 2021?}

The counterfactual simulation replaces the severe Crisis-phase shocks with attenuated tensions and introduces a positive Apple cooperation signal at quarter 44 (Q3 2019), modeling a proactive concession rather than a reactive one. Specifically, the Tension-phase developer shocks are reduced from $-0.15$ to $-0.05$, and the Crisis-phase shocks from $-0.40$ to $-0.10$, while Apple receives a $+0.15$ cooperation boost at quarter 44.

Results show that early concession would have preserved higher mean cooperation across all actors ($+12$--$18\%$ relative to baseline) and maintained bilateral trust above 0.5, avoiding the trust collapse observed in the Crisis phase. The counterfactual Adjustment phase begins from a higher cooperation baseline, yielding faster convergence to the post-reform equilibrium. This analysis demonstrates the model's utility for strategic decision support---evaluating ``what-if'' policy scenarios before committing to ecosystem governance changes---paralleling the alternative supplier dependency analysis conducted for Samsung-Sony in TR-2025-01~\cite{pant2025foundations}.

\section{Discussion}
\label{sec:discussion}

Building on the computational validation demonstrating model correctness and the empirical validation demonstrating real-world applicability, this section examines implications for requirements engineering and multi-agent systems, acknowledges limitations, and identifies future research directions.

\subsection{Implications for Requirements Engineering}
\label{sec:discussion:re}

The reciprocity framework offers three contributions to requirements engineering practice and research.

\subsubsection{Sequential Dependency Modeling in Stakeholder Analysis}

Traditional stakeholder analysis identifies dependencies between actors but treats these dependencies as static relationships. Our framework extends stakeholder modeling to capture sequential dynamics where current cooperation depends on historical behavior.

\textbf{Temporal Dependency Specification.} Requirements engineers can augment \textit{i*} Strategic Dependency models with temporal annotations specifying: (i) the memory window $k$ over which stakeholders assess partner behavior, (ii) the reciprocity sensitivity $\rho_{ij}$ governing response strength, and (iii) the bounding parameter $\kappa$ limiting escalation. For example, a dependency ``Developer depends on Platform for API Stability'' becomes ``Developer cooperation at sprint $t$ depends on Platform's API stability over sprints $t-4$ to $t-1$'' with $k=4$, $\rho=0.8$, and $\kappa=1.0$.

\textbf{Critical Path Analysis.} The framework enables identification of dependencies where violations would trigger reciprocity cascades. Dependencies with high $\rho_{ij}$ values (strong reciprocity) and high $D_{ij}$ values (strong interdependence) represent critical paths requiring careful management. Violations on these paths produce amplified negative responses that may propagate through the stakeholder network.

\textbf{Intervention Design.} Requirements engineers can design interventions to cultivate positive reciprocity: service level agreements specifying mutual obligations, transparency mechanisms making stakeholder actions visible, and commitment devices signaling long-term engagement. The framework quantifies how such interventions affect cooperation trajectories through the reciprocity response function.

\subsubsection{Reciprocity Assessment for Long-Term Stakeholder Engagement}

Requirements elicitation is inherently sequential. Stakeholder disclosure at sprint $t$ depends on analyst responsiveness at sprint $t-1$. Analyst effort at sprint $t$ depends on stakeholder engagement at sprint $t-1$. This creates reciprocity dynamics that the framework captures.

\textbf{Engagement Trajectory Prediction.} Given estimated reciprocity parameters, practitioners can simulate how different engagement strategies affect long-term stakeholder participation. Immediate responsiveness (low latency) produces positive cooperation signals that compound through reciprocity, while delayed responses produce negative signals that erode engagement over time.

\textbf{Critical Juncture Identification.} The framework identifies periods where violations would cause lasting damage. Early-stage violations during initial stakeholder engagement might trigger persistent negative reciprocity (low initial trust amplifies negativity bias), while later violations after trust establishment might be forgiven more readily (high trust provides buffer).

\textbf{Memory Window Calibration.} Project contexts determine appropriate memory windows. Agile projects with weekly sprints might use $k=4$ to $k=8$ representing one to two months of history. Waterfall projects with quarterly milestones might use $k=2$ to $k=4$ representing six months to one year. The framework provides guidance for calibrating $k$ based on project rhythm and stakeholder expectations.

\subsubsection{Process Modeling Extensions}

The reciprocity formalization suggests extensions to requirements engineering process models.

\textbf{SPEM Extensions.} Software Process Engineering Metamodel (SPEM) could incorporate reciprocity annotations on work product dependencies, enabling process simulation that accounts for stakeholder behavioral dynamics rather than assuming uniform compliance.

\textbf{i* Tool Extensions.} Goal modeling tools could implement reciprocity analysis modules that, given temporal dependency specifications and estimated parameters, compute cooperation equilibria, identify reciprocity spirals, and simulate intervention effects. Such tools would bridge conceptual modeling and quantitative analysis.

\subsection{Implications for Multi-Agent Systems}
\label{sec:discussion:mas}

The reciprocity framework addresses coordination challenges in multi-agent systems, with particular relevance to emerging agentic AI applications.

\subsubsection{Agentic AI Coordination Mechanisms}

As autonomous AI agents increasingly interact in shared environments, coordination mechanisms must handle sequential dependencies and behavioral responses.

\textbf{Reciprocity-Aware Agent Design.} Agents can implement reciprocity mechanisms to sustain cooperation in mixed-motive environments. The bounded response function $\phi_{\text{recip}}(x) = \tanh(\kappa x)$ provides a principled mechanism for adjusting cooperation based on partner behavior while preventing escalation spirals. Agents using this mechanism exhibit proportional responses that maintain cooperation with cooperative partners while punishing defection without excessive retaliation.

\textbf{Memory Architecture.} The memory window formalization $\bar{a}_j^{t-k:t-1}$ specifies how agents should aggregate historical observations. Finite memory windows implement bounded rationality appropriate for resource-constrained agents, while the moving average formulation provides robustness to noise in partner behavior observations.

\textbf{Trust-Gated Cooperation.} The trust-gating mechanism $T_{ij}^t \cdot \rho_{ij} \cdot R_{ij}$ provides a principled approach for agents to modulate cooperation based on relationship quality. New partners (low trust) receive cautious cooperation that increases as trust builds through positive interactions. This implements appropriate caution without excessive conservatism.

\subsubsection{Large Language Model Agent Reciprocity}

The emergence of LLM-based agents creates new coordination challenges where the reciprocity framework provides guidance.

\textbf{Prompt-Based Reciprocity.} LLM agents can implement reciprocity through prompt engineering that conditions responses on interaction history. The framework specifies what history to track (cooperation signals over $k$ periods), how to aggregate history (moving average), and how to translate history into behavioral adjustments (bounded response function).

\textbf{Multi-Agent LLM Coordination.} When multiple LLM agents interact (e.g., in collaborative coding, research assistance, or task planning), reciprocity mechanisms can sustain productive collaboration. Agents that contribute high-quality outputs receive reciprocal high-effort responses, while agents that free-ride face reduced cooperation from partners.

\textbf{Human-AI Team Dynamics.} Human-AI teams exhibit sequential interaction where human engagement depends on AI responsiveness and AI effort depends on human feedback. The framework models these dynamics, predicting how AI response quality affects human engagement trajectories and identifying interventions to sustain productive collaboration.

\subsubsection{Decentralized System Governance}

Multi-agent systems often lack central authorities, requiring decentralized governance mechanisms.

\textbf{Emergent Norms.} Reciprocity provides endogenous norm enforcement without central coordination. Agents that violate cooperative norms face distributed punishment through negative reciprocity from multiple partners, creating decentralized enforcement that scales with system size.

\textbf{Reputation Integration.} The framework integrates with reputation systems through the two-layer trust model from TR-2. Immediate trust captures direct interaction experience while reputation aggregates community assessments, enabling agents to condition behavior on both personal history and collective knowledge.

\subsection{Limitations}
\label{sec:discussion:limitations}

Several limitations merit acknowledgment, following the standardized format from prior technical reports in this series.

\textbf{Limitation 1: Single Case Study Generalizability.} The empirical validation relies on a single case study (Apple iOS App Store ecosystem). While this case provides rich longitudinal data across multiple phases, generalization to other platform ecosystems, organizational contexts, or cultural settings requires additional validation. The iOS ecosystem exhibits specific characteristics (strong platform control, large developer community, extensive documentation) that may not transfer to all coopetitive settings.

\textbf{Limitation 2: Western Documentation Bias.} Available evidence for parameter elicitation and outcome validation derives primarily from Western sources (US court filings, EU regulatory documents, English-language industry reports). Reciprocity dynamics in non-Western contexts may exhibit different patterns due to cultural factors, regulatory environments, and market structures not captured in our data sources.

\textbf{Limitation 3: Cultural Factors in Reciprocity Norms.} The model treats reciprocity parameters ($\rho_0$, $\kappa$, $k$) as context-independent, but substantial evidence from behavioral economics indicates that reciprocity norms vary across cultures. Collectivist cultures may exhibit stronger positive reciprocity and weaker negative reciprocity compared to individualist cultures. The framework would benefit from cultural calibration studies establishing parameter ranges for different cultural contexts.

\textbf{Limitation 4: Simulation versus Behavioral Validation Gap.} Validation demonstrates that the model produces outputs matching historical outcomes, but does not directly validate the behavioral mechanisms. Actors in the iOS ecosystem may have responded through mechanisms other than the reciprocity formalization we propose. Behavioral validation through controlled experiments with human subjects would strengthen claims about mechanism validity.

\textbf{Limitation 5: Parameter Sensitivity to Context.} While Monte Carlo analysis demonstrates robustness under parameter perturbation, parameters themselves require context-specific elicitation. The translation methodology from \textit{i*} models provides guidance, but practitioners must exercise judgment in applying the framework to new contexts. Sensitivity analysis should accompany any application to identify which parameters most affect conclusions.

\section{Synthesis: The Complete Coopetition Framework}
\label{sec:synthesis}

Having completed all four technical reports, we present a synthesis showing how the dimensions integrate into a unified framework for analyzing strategic coopetition in socio-technical systems.

\subsection{Four Dimensions of Strategic Coopetition}
\label{sec:synthesis:dimensions}

Table~\ref{tab:four_dimensions} summarizes the four dimensions formalized across the research program, showing how each contributes distinct analytical capability to the integrated framework.

\begin{table}[htbp]
\centering
\caption{Four dimensions of the strategic coopetition framework}
\label{tab:four_dimensions}
\begin{tabular}{llllp{3.5cm}}
\toprule
\textbf{Dimension} & \textbf{Report} & \textbf{Construct} & \textbf{Mathematical Form} & \textbf{Role in Framework} \\
\midrule
Interdependence & TR-1 & $D_{ij}$ & $\sum_d w_d \cdot \text{crit}_d$ & Structural coupling between actors \\
\addlinespace
Trust & TR-2 & $T_{ij}^t$ & $T + \lambda^+ \max(0,s) - \lambda^- \max(0,-s)$ & Relationship quality modulating behavior \\
\addlinespace
Loyalty & TR-3 & $\theta_{i|C}$ & $\phi_B, \phi_C$ modulation & Team commitment overcoming free-riding \\
\addlinespace
\textbf{Reciprocity} & \textbf{TR-4} & $\boldsymbol{\rho_{ij}}$ & $\boldsymbol{\tanh(\kappa \cdot s_{ij})}$ & \textbf{Temporal enforcement sustaining cooperation} \\
\bottomrule
\end{tabular}
\end{table}

\textbf{Interdependence} (TR-1) captures how actors' outcomes are structurally coupled through dependencies. The interdependence coefficient $D_{ij}$ aggregates weighted criticalities from \textit{i*} dependency specifications, creating instrumental motivation for actors to consider partners' outcomes when optimizing their own behavior.

\textbf{Trust} (TR-2) captures the dynamic quality of relationships through a two-layer model distinguishing immediate trust from persistent reputation. The asymmetric updating rule (3:1 negativity bias) explains why trust builds slowly but erodes quickly, with lasting reputation effects.

\textbf{Loyalty} (TR-3) captures team commitment through two consolidated parameters: cost tolerance $\phi_C$ reducing perceived effort costs and loyalty benefit $\phi_B$ weighting teammates' aggregate payoff. The companion work's appendix decomposes these into welfare internalization, warm glow, cost tolerance, and guilt aversion mechanisms for practitioners requiring granular analysis.

\textbf{Reciprocity} (TR-4) captures temporal enforcement through conditional strategies. The bounded response function ensures proportional reactions that sustain cooperation without escalation, while memory windows implement bounded rationality appropriate for sequential interaction.

\subsection{Cross-Layer Integration}
\label{sec:synthesis:integration}

The four dimensions integrate through a complete utility function capturing all mechanisms simultaneously. For atomic actors (non-team-members), the integrated utility decomposes as:

\begin{equation}
\label{eq:complete_utility}
U_i(\vect{a}^t, h^{t-1}) = \underbrace{\pi_i(\vect{a}^t)}_{\substack{\text{Self-interest} \\ \text{(Base payoff)}}} + \underbrace{\sum_{j \neq i} D_{ij} \pi_j(\vect{a}^t)}_{\substack{\text{TR-1: Interdependence} \\ \text{(Outcome coupling)}}} + \underbrace{\lambda_T \sum_{j \neq i} T_{ij}^t (1 + \omega D_{ij}) \rho_{ij} R_{ij}}_{\substack{\text{TR-2 + TR-4: Trust-gated reciprocity} \\ \text{(Temporal enforcement)}}}
\end{equation}

For team members, additional loyalty-based terms from TR-3 augment the utility. The following uses the finer-grained mechanism decomposition from the companion work's appendix (Appendix~B of~\cite{pant2025teams}) to illustrate how each loyalty component integrates with the other dimensions:

\begin{equation}
\label{eq:team_utility}
\begin{split}
U_i(\vect{a}^t, h^{t-1}) = &\underbrace{\frac{1}{n} Q\left(\sum_{j \in C} a_j^t\right) - c_i^{\text{perceived}}(a_i^t, \theta_{i|C})}_{\text{Team production with loyalty-adjusted costs}} \\
&+ \underbrace{\lambda_{\text{intern}}(\theta_{i|C}) \sum_{j \in C, j \neq i} \pi_j^{\text{base}}}_{\text{TR-3: Welfare internalization}} + \underbrace{U_i^{\text{warm}} + U_i^{\text{guilt}}}_{\text{TR-3: Motivational effects}} \\
&+ \underbrace{\lambda_T \sum_{j \notin C} T_{ij}^t (1 + \omega D_{ij}) \rho_{ij} R_{ij}(\vect{a}^t, h^{t-1})}_{\text{External reciprocity with non-team actors}}
\end{split}
\end{equation}

The integration creates cross-layer effects where each dimension modulates the others:

\begin{itemize}
\item \textbf{Interdependence amplifies reciprocity}: $\rho_{ij} = \rho_0 D_{ij}^\eta$ links structural position to behavioral sensitivity
\item \textbf{Trust gates reciprocity}: $T_{ij}^t \cdot \rho_{ij}$ prevents reciprocity from activating without sufficient trust
\item \textbf{Reciprocity sustains trust}: Consistent cooperation through reciprocity maintains positive trust signals
\item \textbf{Loyalty creates local reciprocity}: Team members exhibit intensified reciprocity within teams
\end{itemize}

\subsection{Synergistic Interactions Among Dimensions}

\textbf{Interdependence Amplifies Trust and Reciprocity.} Structural dependencies $D_{ij}$ amplify both trust sensitivity through $\xi D_{ij}$ factor in trust updating from second paper, and reciprocity sensitivity through $\rho_{ij} \propto D_{ij}^\eta$ from this paper. This creates integrated structural-behavioral coupling where structural position determines behavioral response strengths.

\textbf{Trust Gates Reciprocity.} Trust multiplies reciprocity responses through $T_{ij}^t \cdot \rho_{ij} \cdot R_{ij}$. Low trust prevents reciprocity from sustaining cooperation even when parameters favor conditional cooperation. High trust enables full reciprocity enforcement. This creates trust-reciprocity complementarity.

\textbf{Reciprocity Sustains Trust.} Consistent cooperation through reciprocity maintains high trust over time where positive cooperation signals continuously update trust upward. Reciprocity enforcement punishes violations that would erode trust, creating feedback loop where reciprocity protects trust and trust enables reciprocity.

\textbf{Complementarity Creates Value Worth Cooperating For.} Added Value from complementarity in first paper creates surplus worth pursuing, increasing incentive for reciprocity to sustain cooperation. Without complementarity where $\gamma = 0$, reciprocity would sustain cooperation but at low action levels. With complementarity, reciprocity sustains high cooperation capturing synergistic value.

\textbf{Team Production Creates Local Reciprocity.} Team members reciprocate each other's contributions from third paper combined with this paper, overlaying broader reciprocity with external actors. Loyalty moderates team reciprocity through welfare internalization, while external reciprocity moderates team's strategic positioning relative to other actors.

\subsection{Theoretical Synthesis}

The framework explains cooperation in coopetitive environments through multiple reinforcing mechanisms.

Structural interdependence creates instrumental motivation to care about partners' outcomes from first paper. Complementarity creates value from cooperation worth pursuing through synergistic production from first paper. Trust enables actors to overcome fear of exploitation by providing confidence in partners' reliability from second paper. Reciprocity provides endogenous enforcement through conditional strategies sustaining cooperation from this paper. Loyalty moderates team contribution overcoming free-riding through welfare internalization from third paper.

No single mechanism suffices. Interdependence without reciprocity creates vulnerability without enforcement. Reciprocity without trust creates conditional cooperation that cannot activate due to fear. Trust without reciprocity creates cooperation that cannot sustain against violations. Complementarity without interdependence creates value that cannot be realized due to coordination failure. Cooperation emerges from their synergistic combination.

\section{Conclusion}
\label{sec:conclusion}

This paper concludes the Foundations Series of a coordinated research program bridging conceptual modeling and game theory for analyzing strategic coopetition.

\textbf{Research Program Progression.} First technical report~\cite{pant2025foundations} established interdependence from \textit{i*} dependencies and complementarity from Added Value, introducing Coopetitive Equilibrium demonstrating how structural dependencies shift equilibria toward cooperation. Second technical report~\cite{pant2025trust} formalized trust as two-layer dynamic model with asymmetric updating, showing trust builds slowly but erodes quickly with lasting reputation damage. Third technical report~\cite{pant2025teams} extended to team production, showing loyalty overcomes free-riding through cost tolerance, welfare internalization, warm glow, and guilt aversion achieving six-fold effort differentiation. This technical report, completes our framework with reciprocity formalization showing cooperation is sustained through conditional strategies, history-dependent responses, and trust-gated enforcement.

\textbf{Paper 4 Contributions.} This technical report makes ten contributions to the research program: (1) formal reciprocity model through bounded response functions $\phi_{\text{recip}}(x) = \tanh(\kappa x)$ capturing finite proportional responses; (2) memory-windowed history tracking $\bar{a}_j^{t-k:t-1}$ representing bounded rationality; (3) structural derivation $\rho_{ij} = \rho_0 D_{ij}^\eta$ grounding reciprocity in \textit{i*} dependencies; (4) trust-gated reciprocity showing $T_{ij}$ gates conditional cooperation; (5) systematic eight-step translation methodology from \textit{i*} to computational parameters; (6) comprehensive parameter validation across 15,625 configurations with all six behavioral targets exceeding thresholds (87.9\%--100.0\%), enabled by full TR-2 two-layer trust model integration; (7) functional validation with 5 of 5 experiments successful demonstrating cooperation emergence, asymmetric differentiation, memory effects, trust-reciprocity interaction, and team production complementarity; (8) empirical validation through Apple iOS App Store case study achieving 43.0 of 51 applicable points (84.3\%) across 66 quarters and 5 ecosystem phases; (9) complete four-dimension framework synthesis; and (10) implications for requirements engineering and multi-agent systems.

\textbf{Implications for Conceptual Modeling.} This research program demonstrates how conceptual models including \textit{i*}, Tropos, and GRL can be extended with computational game-theoretic analysis enabling quantitative reasoning about strategic dynamics. The framework provides formal tools for analyzing cooperation and defection in multi-stakeholder systems, predicting equilibrium behaviors, designing governance mechanisms, and understanding how structural dependencies, trust evolution, loyalty, and reciprocity interact.

\textbf{Future Work.} Six directions merit investigation to extend and strengthen the framework.

\textit{Direction 1: Incomplete Information Models.} The current framework assumes actors observe partners' actions and parameters. Extending to incomplete information where actors have uncertain knowledge of partners' preferences, capabilities, or intentions would increase realism. Bayesian updating mechanisms could enable actors to learn about partners through observed behavior, integrating with the reciprocity framework through belief-dependent response functions.

\textit{Direction 2: Learning and Adaptation.} The framework treats parameters ($\rho_0$, $\eta$, $k$, $\kappa$) as exogenous, but actors may adapt these parameters through experience. Reinforcement learning mechanisms could enable actors to learn optimal memory windows, calibrate reciprocity sensitivity, and adjust bounding parameters based on interaction outcomes. This would create truly adaptive agents appropriate for long-horizon multi-agent systems.

\textit{Direction 3: Empirical Behavioral Validation.} While the Apple iOS case study validates model outputs against historical outcomes, direct behavioral validation through controlled experiments would strengthen mechanism claims. Laboratory studies with human subjects playing sequential reciprocity games could test whether the bounded response function, memory window effects, and trust-gating mechanisms match observed human behavior.

\textit{Direction 4: Tool Development.} Integrating the framework with existing goal modeling tools would enable practitioners to apply the formalization without specialized game-theoretic expertise. Extensions to \textit{i*} tools (OpenOME, jUCMNav, piStar) could implement reciprocity analysis modules that, given dependency specifications and estimated parameters, compute cooperation trajectories, identify reciprocity spirals, and simulate governance interventions.

\textit{Direction 5: Human-AI Team Applications.} As AI agents increasingly collaborate with human teams, the framework provides foundations for designing AI reciprocity mechanisms. Research should investigate how AI agents can implement appropriate reciprocity (responding to human cooperation without excessive punishment for occasional lapses), calibrate memory windows to human expectations, and build trust through consistent behavior.

\textit{Direction 6: Cultural Adaptation.} The framework's parameters likely vary across cultural contexts. Cross-cultural studies could establish parameter ranges for collectivist versus individualist cultures, high versus low power distance societies, and different institutional contexts. This would enable culturally-appropriate application of the framework in international and multicultural settings.

This technical report is part of a coordinated research program on computational approaches to strategic coopetition. The first technical report~\cite{pant2025foundations} (arXiv:2510.18802) established interdependence and complementarity, achieving 58/60 validation for the Samsung-Sony S-LCD case. The second~\cite{pant2025trust} (arXiv:2510.24909) formalized trust dynamics, achieving 49/60 for the Renault-Nissan Alliance. The third~\cite{pant2025teams} (arXiv:2601.16237) addressed collective action and loyalty, achieving 45/60 for the Apache HTTP Server project. Together with this report's 43/51 validation for the Apple iOS App Store ecosystem, the four technical reports provide comprehensive computational foundations for analyzing strategic coopetition. The complete framework addresses four key dimensions: interdependence captures structural dependencies creating outcome coupling, trust evolves dynamically through asymmetric updating modulating strategic behavior, loyalty enables team commitment overcoming free-riding incentives, and reciprocity provides endogenous enforcement sustaining cooperation through conditional strategies. These dimensions jointly enable analysis of the rich strategic dynamics that characterize modern socio-technical systems. The Foundations Series adopts uniaxial treatment where agents choose cooperation levels along a single continuum with competitive dynamics emerging through structural parameters. A companion Extensions Series (TR-5 through TR-8) extends this framework to biaxial treatment, addressing phenomena where cooperation and competition constitute independent strategic dimensions.

\appendix

\section{Proofs}
\label{app:proofs}

This appendix provides complete proofs for the propositions stated in Section~\ref{sec:formalization}.

\subsection{Proof of Proposition~\ref{prop:coop_emergence} (Cooperation Emergence Condition)}

\begin{proof}
Consider a two-player repeated game with players $i$ and $j$. Player $i$'s utility at period $t$ is given by:
\begin{equation}
U_i(a_i^t, a_j^t, h^{t-1}) = \pi_i(a_i^t, a_j^t) + \lambda_R T_{ij}^t (1 + \omega D_{ij}) \rho_{ij} R_{ij}(a_i^t, a_j^t, h^{t-1})
\end{equation}
where $\pi_i(a_i, a_j) = b(a_i, a_j) - c(a_i)$ with $b$ representing benefits and $c$ representing costs.

At an interior cooperative equilibrium $(a_i^*, a_j^*)$ with $a_i^*, a_j^* > 0$, the first-order condition requires:
\begin{equation}
\frac{\partial U_i}{\partial a_i} = \frac{\partial \pi_i}{\partial a_i} + \lambda_R T_{ij}^t (1 + \omega D_{ij}) \rho_{ij} \frac{\partial R_{ij}}{\partial a_i} = 0
\end{equation}

The reciprocity response function is $R_{ij} = \phi_{\text{recip}}(s_{ij})$ where $s_{ij} = a_j^t - \bar{a}_j^{t-k:t-1}$ is the cooperation signal. At a symmetric cooperative equilibrium, $a_i^t = a_j^t = a^*$ for all $t$, implying $s_{ij} = 0$ (actions match historical average).

Near this equilibrium, consider a small deviation. The derivative of the bounded response function at zero is:
\begin{equation}
\phi_{\text{recip}}'(0) = \frac{d}{dx}\tanh(\kappa x)\bigg|_{x=0} = \kappa
\end{equation}

The cooperation signal responds to player $i$'s action through the effect on future partner responses. In the stationary equilibrium, player $i$'s action today affects the moving average $\bar{a}_i$ in future periods, which in turn affects $s_{ji}$ and thus $a_j$ through reciprocity.

For cooperation to be incentive-compatible, the marginal benefit from reciprocity must exceed the marginal cost:
\begin{equation}
\lambda_R T_{ij}^t (1 + \omega D_{ij}) \rho_{ij} \cdot \kappa \cdot \frac{\partial s_{ij}}{\partial a_i} \geq c'(a_i^*)
\end{equation}

At the symmetric equilibrium with $\rho_{ij} = \rho_0 D_{ij}^\eta$, taking $D_{ij} = 1$ for simplicity:
\begin{equation}
\lambda_R T^* (1 + \omega) \rho_0 \kappa \geq c'
\end{equation}

Solving for the critical reciprocity threshold:
\begin{equation}
\rho_0 \geq \rho^* = \frac{c'}{\lambda_R T^* (1 + \omega) \kappa}
\end{equation}

More generally, with asymmetric dependencies:
\begin{equation}
\rho_0 > \rho^* = \frac{c'}{\lambda_R T^* (1 + \omega D_{ij}) \kappa}
\end{equation}

This establishes that a cooperative equilibrium exists if and only if base reciprocity exceeds the critical threshold $\rho^*$. The threshold is (i) decreasing in trust $T^*$ (high trust lowers the reciprocity requirement), (ii) decreasing in the bounding parameter $\kappa$ (high sensitivity amplifies reciprocity effects), and (iii) increasing in marginal cost $c'$ (costly cooperation requires stronger reciprocity incentives).
\end{proof}

\subsection{Proof of Proposition~\ref{prop:memory_effect} (Memory Window Effect on Forgiveness)}

\begin{proof}
Consider a cooperative equilibrium where both players maintain action level $a^*$ for all $t < t^*$. At period $t^*$, player $j$ defects by reducing action to $a^* - \delta$ for one period, then returns to $a^*$ for all $t > t^*$.

\textbf{Lower Bound ($\tau_f \geq k$).}
The moving average for player $j$'s actions is:
\begin{equation}
\bar{a}_j^{t-k:t-1} = \frac{1}{k}\sum_{\tau=t-k}^{t-1} a_j^\tau
\end{equation}

For periods $t \in \{t^*+1, \ldots, t^*+k\}$, the defection at $t^*$ is included in the moving average window. Specifically:
\begin{equation}
\bar{a}_j^{t-k:t-1} = a^* - \frac{\delta}{k} \quad \text{for } t \in \{t^*+1, \ldots, t^*+k\}
\end{equation}

The cooperation signal becomes:
\begin{equation}
s_{ij}^t = a_j^t - \bar{a}_j^{t-k:t-1} = a^* - \left(a^* - \frac{\delta}{k}\right) = \frac{\delta}{k} > 0 \quad \text{for } t > t^*
\end{equation}

Thus, the signal is positive (indicating improved behavior relative to history) only after the defection exits the memory window at period $t^* + k + 1$. Until then, $s_{ij}$ may be zero or negative depending on the exact timing. The defection affects reciprocity responses for at least $k$ periods, establishing $\tau_f \geq k$.

\textbf{Upper Bound ($\tau_f \leq 2k$).}
The reciprocity response to the defection triggers player $i$ to reduce cooperation through negative reciprocity. This reduced cooperation by player $i$ in periods $t^*+1$ through $t^*+k$ enters player $j$'s calculation of $\bar{a}_i$.

Let $\Delta a_i^t$ denote player $i$'s cooperation reduction at period $t$. Through the bounded response function:
\begin{equation}
\Delta a_i^t = \alpha \cdot \phi_{\text{recip}}\left(\frac{-\delta}{k}\right) = -\alpha \cdot \tanh\left(\frac{\kappa \delta}{k}\right)
\end{equation}
where $\alpha$ is the adjustment rate.

Player $j$'s moving average of player $i$'s behavior incorporates these reductions. The secondary effect (player $i$'s reduced cooperation affecting player $j$'s reciprocity) persists for an additional $k$ periods after the primary effect clears.

Total recovery time is bounded by:
\begin{equation}
\tau_f \leq k + k = 2k
\end{equation}

\textbf{Dependence on $\kappa$.}
For $\kappa \to \infty$, the bounded response function approaches a step function: $\phi_{\text{recip}}(x) \to \text{sign}(x)$. Any negative signal triggers maximal punishment, and the secondary effects are amplified. However, the signal returns to zero as soon as the defection exits the window, so $\tau_f \to k$.

For $\kappa \to 0$, the response function approaches zero: $\phi_{\text{recip}}(x) \to 0$. Reciprocity responses are minimal, and recovery is immediate once the direct memory effect clears. However, the weak punishment allows drift, potentially extending recovery toward $2k$.

The exact value of $\tau_f \in [k, 2k]$ depends on the specific values of $\kappa$, $\alpha$, and the magnitude of the original defection $\delta$.
\end{proof}

\subsection{Proof of Proposition~\ref{prop:trust_recip_complement} (Trust-Reciprocity Complementarity)}

\begin{proof}
The integrated utility function includes the reciprocity term (using $\phi_{\text{recip}}(s_{ij}^t)$ to distinguish the bounded response function from the reputation state variable $R_{ij}^t$):
\begin{equation}
U_i^{\text{recip}} = \lambda_R T_{ij}^t (1 + \omega D_{ij}) \rho_{ij} \phi_{\text{recip}}(s_{ij}^t)
\end{equation}

At the equilibrium action $a_i^*$, the first-order condition is:
\begin{equation}
\frac{\partial U_i}{\partial a_i}\bigg|_{a_i = a_i^*} = 0
\end{equation}

We seek to show that $\frac{\partial^2 a_i^*}{\partial T_{ij} \partial \rho_{ij}} > 0$.

By the implicit function theorem, if $F(a_i, T_{ij}, \rho_{ij}) = \frac{\partial U_i}{\partial a_i} = 0$ defines $a_i^*$ implicitly as a function of $(T_{ij}, \rho_{ij})$, then:
\begin{equation}
\frac{\partial a_i^*}{\partial T_{ij}} = -\frac{\partial F / \partial T_{ij}}{\partial F / \partial a_i}
\end{equation}

The second-order condition for a maximum requires $\frac{\partial F}{\partial a_i} = \frac{\partial^2 U_i}{\partial a_i^2} < 0$ (concavity).

The cross-partial derivative is:
\begin{equation}
\frac{\partial^2 a_i^*}{\partial T_{ij} \partial \rho_{ij}} = \frac{\partial}{\partial \rho_{ij}}\left(-\frac{\partial F / \partial T_{ij}}{\partial F / \partial a_i}\right)
\end{equation}

From the utility function, $F$ contains the term:
\begin{equation}
\frac{\partial U_i^{\text{recip}}}{\partial a_i} = \lambda_R T_{ij} (1 + \omega D_{ij}) \rho_{ij} \frac{\partial \phi_{\text{recip}}(s_{ij}^t)}{\partial a_i}
\end{equation}

Thus:
\begin{equation}
\frac{\partial F}{\partial T_{ij}} = \lambda_R (1 + \omega D_{ij}) \rho_{ij} \frac{\partial \phi_{\text{recip}}(s_{ij}^t)}{\partial a_i}
\end{equation}

And:
\begin{equation}
\frac{\partial^2 F}{\partial T_{ij} \partial \rho_{ij}} = \lambda_R (1 + \omega D_{ij}) \frac{\partial \phi_{\text{recip}}(s_{ij}^t)}{\partial a_i}
\end{equation}

For positive reciprocity (cooperation begets cooperation), $\frac{\partial \phi_{\text{recip}}(s_{ij}^t)}{\partial a_i} > 0$ in equilibrium. Combined with the concavity condition on $\partial F / \partial a_i < 0$, the quotient rule and chain rule applied to the implicit function yield:
\begin{equation}
\frac{\partial^2 a_i^*}{\partial T_{ij} \partial \rho_{ij}} > 0
\end{equation}

Intuitively, both trust and reciprocity appear multiplicatively in the utility function through $T_{ij} \cdot \rho_{ij}$. Increasing either parameter amplifies the marginal effect of the other on cooperation incentives. High trust makes reciprocity more effective (actors are willing to respond to cooperative signals), and high reciprocity makes trust more valuable (trust enables the reciprocity mechanism to function).

This complementarity implies that interventions targeting both trust-building and reciprocity-enhancing mechanisms are more effective than interventions targeting either mechanism alone.
\end{proof}

\section{Extended Sensitivity Analysis}
\label{app:sensitivity}

This appendix provides additional details on the comprehensive parameter space validation described in Section~\ref{sec:validation}.

\subsection{Full Factorial Design}

The validation explored a full factorial design over six parameters:

\begin{table}[htbp]
\centering
\caption{Parameter space for full factorial validation}
\label{tab:app_parameter_space}
\begin{tabular}{llll}
\toprule
\textbf{Parameter} & \textbf{Symbol} & \textbf{Values Tested} & \textbf{Count} \\
\midrule
Base reciprocity & $\rho_0$ & $\{0.2, 0.4, 0.6, 0.8, 1.0\}$ & 5 \\
Dependency exponent & $\eta$ & $\{0.5, 0.75, 1.0, 1.25, 1.5\}$ & 5 \\
Memory window & $k$ & $\{1, 2, 4, 8, 16\}$ & 5 \\
Bounding parameter & $\kappa$ & $\{0.5, 1.0, 1.5, 2.0, 3.0\}$ & 5 \\
Initial trust & $T^0$ & $\{0.3, 0.5, 0.7, 0.85, 0.95\}$ & 5 \\
Dependency level & $D_{ij}$ & $\{0.2, 0.4, 0.6, 0.8, 1.0\}$ & 5 \\
\midrule
\multicolumn{3}{r}{\textbf{Total configurations:}} & $5^6 = 15,625$ \\
\bottomrule
\end{tabular}
\end{table}

\subsection{Behavioral Target Achievement Rates}

Table~\ref{tab:app_target_rates} reports the achievement rate for each behavioral target across all 15,625 configurations.

\begin{table}[htbp]
\centering
\caption{Behavioral target achievement across parameter space}
\label{tab:app_target_rates}
\begin{tabular}{lccl}
\toprule
\textbf{Target} & \textbf{Achieved} & \textbf{Rate} & \textbf{Primary Sensitivity} \\
\midrule
T1: Bounded responses & 15,625 & 100.0\% & --- (structural) \\
T2: Cooperation emergence & 15,235 & 97.5\% & $\rho_0$, $\lambda_R$ \\
T3: Asymmetric differentiation & 15,625 & 100.0\% & $\eta$, $D_{ij}$ \\
T4: Defection punishment & 15,625 & 100.0\% & --- \\
T5: Trust-gated reciprocity & 15,625 & 100.0\% & $T^0$, $\rho_0$ \\
T6: Forgiveness after violation & 13,728 & 87.9\% & $k$, $\kappa$ \\
\bottomrule
\end{tabular}
\end{table}

\subsection{Sensitivity Patterns}

Detailed parameter sensitivity patterns are discussed in Section~\ref{sec:validation}. The five primary sensitivities (base reciprocity $\rho_0$, dependency exponent $\eta$, memory window $k$, bounding parameter $\kappa$, and initial trust $T^0$) are analyzed in the context of behavioral target achievement rates reported in Table~\ref{tab:app_target_rates}.

\subsection{Monte Carlo Robustness}

The Monte Carlo analysis perturbed all parameters simultaneously by $\pm 15\%$ (uniform distribution) across 2,000 trials. For each trial, all six behavioral targets were evaluated.

\begin{table}[htbp]
\centering
\caption{Monte Carlo robustness results (2,000 trials)}
\label{tab:app_monte_carlo}
\begin{tabular}{ll}
\toprule
\textbf{Metric} & \textbf{Value} \\
\midrule
Trials meeting all targets & 1,890 / 2,000 (94.5\%) \\
Mean differentiation ratio & 2.00 (95\% CI: [1.88, 2.12]) \\
Minimum differentiation ratio & 1.80 \\
Cooperation emergence rate & 97.5\% \\
\bottomrule
\end{tabular}
\end{table}

The 94.5\% success rate under $\pm 15\%$ perturbation demonstrates that the model is robust to parameter estimation uncertainty within plausible calibration ranges. The minimum observed differentiation ratio of 1.80 remains above the 1.5 threshold in all trials.

\section{Apple iOS Case Study Scoring Details}
\label{app:case_study}

This appendix provides detailed justification for the 60-point validation scoring presented in Section~\ref{sec:empirical:validation}.

\subsection{Scoring Methodology}

Each of the 12 behavioral indicators was assessed for each of the 5 phases using a three-point scale:
\begin{itemize}
\item \textbf{1.0 (Hit)}: Simulated behavior matches historical evidence qualitatively and quantitatively
\item \textbf{0.5 (Partial)}: Simulated behavior matches direction but differs in magnitude or timing
\item \textbf{0.0 (Miss)}: Simulated behavior contradicts historical evidence
\item \textbf{--- (N/A)}: Indicator not applicable for this phase
\end{itemize}

\subsection{Phase-by-Phase Scoring Justification}

\textbf{Phase 1: Symbiosis (2008--2012).}
\begin{itemize}
\item Indicators 1--5, 8, 10--12 scored 1.0: Simulation correctly predicts rising cooperation, appropriate response magnitudes, power asymmetry effects, trust-reciprocity co-evolution, transition timing, equilibrium stability, parameter robustness, and qualitative fit.
\item Indicator 3 scored 0.5: Limited historical data on memory effects during rapid growth phase.
\item Indicators 6, 7, 9 scored N/A: No violations occurred to assess punishment or forgiveness.
\end{itemize}

\textbf{Phase 2: Maturation (2012--2017).}
\begin{itemize}
\item All applicable indicators scored 1.0: Stable high cooperation, minimal violations, established equilibrium. Simulation accurately predicts steady-state dynamics.
\end{itemize}

\textbf{Phase 3: Tension (2017--2020).}
\begin{itemize}
\item Indicators 1, 3--8, 11--12 scored 1.0: Declining cooperation trend, memory effects, asymmetry, trust-reciprocity alignment, punishment following commission disputes, transition timing, parameter robustness, qualitative fit.
\item Indicator 2 scored 0.5: Simulated decline slightly faster than historical (12 quarters vs. observed 14 quarters).
\item Indicator 10 scored 0.5: Within-phase dynamics show oscillation not fully captured.
\item Indicators 7, 9 scored N/A: Forgiveness and recovery not applicable until Phase 5.
\end{itemize}

\textbf{Phase 4: Crisis (2020--2021).}
\begin{itemize}
\item Indicators 1--6, 8, 12 scored 1.0: Sharp cooperation collapse, appropriate magnitudes, visible memory effects, asymmetry, trust collapse, punishment, timing, qualitative fit.
\item Indicator 10 scored 0.5: Crisis dynamics more volatile in simulation than historical.
\item Indicator 11 scored 0.5: Results sensitive to shock magnitude specification.
\item Indicators 7, 9 scored N/A: Not yet in recovery.
\end{itemize}

\textbf{Phase 5: Adjustment (2021--2024).}
\begin{itemize}
\item Indicators 1--4, 6--12 scored 1.0: Recovery trend, magnitudes, memory effects, asymmetry, punishment history, forgiveness dynamics, timing, recovery shape, stability, robustness, qualitative fit.
\item Indicator 5 scored 0.5: Trust recovery lags reciprocity recovery more in simulation than historical data suggests.
\end{itemize}

\subsection{Evidence Sources}

Scoring was informed by the following evidence sources:

\begin{itemize}
\item \textit{Epic Games, Inc. v. Apple Inc.}, Case No. 4:20-cv-05640 (N.D. Cal. 2021): Court filings documenting developer-platform interactions, commission disputes, and policy changes.
\item European Commission Case AT.40437 (Apple App Store): Investigation documents on App Store practices and developer complaints.
\item Sensor Tower and App Annie industry reports (2015--2024): Developer activity metrics, app submission volumes, and ecosystem health indicators.
\item Apple Inc. annual reports and WWDC announcements (2008--2024): Policy changes, developer program modifications, and commission structure evolution.
\item Academic studies including Gawer (2014), Parker et al. (2016), and Boudreau \& Hagiu (2009) on platform ecosystem dynamics.
\end{itemize}

\subsection{Aggregate Scoring}

\begin{table}[htbp]
\centering
\caption{Aggregate validation scoring summary}
\label{tab:app_aggregate}
\begin{tabular}{lcc}
\toprule
\textbf{Phase} & \textbf{Score} & \textbf{Maximum} \\
\midrule
Phase 1: Symbiosis & 8.5 & 9 \\
Phase 2: Maturation & 10.0 & 10 \\
Phase 3: Tension & 8.0 & 10 \\
Phase 4: Crisis & 6.5 & 10 \\
Phase 5: Adjustment & 10.0 & 12 \\
\midrule
\textbf{Total} & \textbf{43.0} & \textbf{51} \\
\textbf{Percentage} & \multicolumn{2}{c}{\textbf{84.3\%}} \\
\bottomrule
\end{tabular}
\end{table}

The model achieves 84.3\% validation score, exceeding the 83\% threshold and demonstrating strong empirical validity for the reciprocity framework when applied to real-world platform ecosystem dynamics.

\bibliographystyle{splncs04}

\end{document}